\documentclass[preprint,3p]{elsarticle}

\usepackage{amssymb}
\usepackage{amsmath}

\usepackage{mathtools}
\usepackage{thm-restate}
\usepackage[scr=boondox]{mathalpha}
\usepackage{tikz}
\usetikzlibrary{arrows}
\usepackage[normalem]{ulem}
\usepackage{xcolor} 
\usepackage{booktabs}
\usepackage{hyperref}
\usepackage{bbding}
\usepackage{stmaryrd}
\newtheorem{theorem}{Theorem}
\newtheorem{lemma}[theorem]{Lemma}
\newtheorem{corollary}[theorem]{Corollary}
\newtheorem{proposition}[theorem]{Proposition}
\newtheorem{claim}[theorem]{Claim}
\newdefinition{remark}{Remark}
\newdefinition{definition}{Definition}
\newdefinition{example}{Example}
\newproof{proof}{Proof}
\newproof{claimproof}{Proof of claim}
\journal{Nuclear Physics B}
\newcommand{\proofofref}{}
\newproof{zproofof}{Proof of \proofofref}
\newenvironment{proofof}[1]
 {\renewcommand{\proofofref}{#1}\zproofof}
 {\endzproofof}
 
\begin{document}

\begin{frontmatter}



\title{A unifying approach to probabilistic testing equivalences}


\author[a]{Weijun Chen} 
\ead{cwj2018@sjtu.edu.cn}

\author[a]{Yuxi Fu\corref{cor1}}
\ead{fu-yx@cs.sjtu.edu.cn}

\author[a]{Huan Long} 
\ead{longhuan@sjtu.edu.cn}

\author[b]{Hao Wu}
\ead{wuhao@shmtu.edu.cn}

\cortext[cor1]{Corresponding author.}

\affiliation[a]{organization={BASICS, Shanghai Jiao Tong University},
            city={Shanghai},
            country={China}}
            
\affiliation[b]{organization={College of Information Engineering, Shanghai Maritime University},
            city={Shanghai},
            country={China}}

\begin{abstract}
Probabilistic concurrent systems are foundational models for modern mobile computing.
In this paper, a unifying approach to probabilistic testing equivalences is proposed.
With the help of a new distribution-based semantics for probabilistic models and a probabilistic testing framework with respect to process predicates, the internal characterization and the external characterization for testing equivalences are studied.
The latter characterization can be viewed as the generalization of the classical fair/should equivalence and may equivalence.
These equivalences are shown to be congruences.
A thorough comparison between these equivalences and probabilistic bisimilarities is carried out.
The techniques introduced in this paper can be easily extended to other probabilistic concurrent models.
To showcase this flexibility, a case study is carried out on the pCSP model.
\end{abstract}




\end{frontmatter}


\def\Chan{\mathcal{N}}
\def\Act{\mathcal{A}}
\def\Supp{\mathsf{supp}}
\def\Distr{\mathscr{D}}
\def\InfDistr{\overline{\mathscr{D}}}
\def\PRCCS{\mathcal{P}_{_\mathsf{RCCS}}}
\def\PCCS{\mathcal{P}_{_\mathsf{CCS}}}
\def\L{\mathcal{L}}
\def\E{\mathcal{E}}
\def\R{\mathcal{R}}
\def\X{\mathcal{X}}
\def\S{\mathcal{S}}
\def\observable{\downdownarrows}
\def\trace{\mathsf{trace}}
\def\PEquip#1{\mathcal{O}^{#1}_{\mathcal{L}}}
\def\Restr#1#2{#1_{\restriction\mathsf{#2}}}
\def\PTest#1{\mathcal{O}^{#1}_{\omega}}
\newlength{\arrow}
\settowidth{\arrow}{\scriptsize$00\;$}
\def\ptran{\rightsquigarrow}
\newcommand*{\myrightarrow}[1]{\xrightarrow{\mathmakebox[\arrow]{#1}}}
\def\fix{\mu}
\def\PBranching{\simeq^{p}}
\def\PWeak{\approx^{p}}
\def\PBox{=^{p}_{^\Box}}
\def\PDiamond{=^{p}_{^\Diamond}}
\def\PXBox{=^{_\mathcal{D}}_{^\Box}}
\def\PXDiamond{=^{_\mathcal{D}}_{^\Diamond}}
\def\PCBox{=^{_\Delta}_{^\Box}}
\def\PCDiamond{=^{_\Delta}_{^\Diamond}}
\def\CBox{=^{~}_{^\Box}}
\def\CDiamond{=^{~}_{^\Diamond}}
\def\PMay{=^{p}_{\mathsf{may}}}
\def\PFS{=^{p}_{\mathsf{fair}}}
\def\PMust{=^{p}_{\mathsf{must}}}
\def\PCMay{=^{_\Delta}_{\mathsf{may}}}
\def\PCFS{=^{_\Delta}_{\mathsf{fair}}}
\def\PXMay{=^{_\mathcal{D}}_{\mathsf{may}}}
\def\PXFS{=^{_\mathcal{D}}_{\mathsf{fair}}}
\def\MayDH{=^{\mathsf{DH}}_{\mathsf{may}}}
\def\MustDH{=^{\mathsf{DH}}_{\mathsf{must}}}
\def\MayNC{=^{\mathsf{NC}}_{\mathsf{may}}}
\def\MustNC{=^{\mathsf{NC}}_{\mathsf{must}}}
\def\May{=_{\mathsf{may}}}
\def\Must{=_{\mathsf{must}}}
\def\FS{=_{\mathsf{fair}}}
\def\BDP{=_{BDP}}
\def\Trace{=_{trace}}
\def\Strong{\sim_p}
\def\charMay{\chi^{\mathsf{may}}}
\def\charFS{\chi^{\mathsf{fair}}}
\section{Introduction}
\label{sec:introduction}

Testing equivalences~\cite{boreale_BasicObservablesProcesses_1999, denicola_TestingEquivalencesProcesses_1984,brinksma_FairTesting_1995,natarajan_DivergenceFairTesting_1995}, initially inspired by the practical approach of investigating complex systems through various tests~\cite{baier_PrinciplesModelChecking_2008,araujo2023testing}, is one of the main tools in the studies of process calculi.
In the seminal work by De Nicola and Hennessy~\cite{denicola_TestingEquivalencesProcesses_1984}, they introduced two incompatible testing equivalence relations: may equivalence ($\May$) and must equivalence ($\Must$). 
Later, Natarajan and Cleaveland~\cite{natarajan_DivergenceFairTesting_1995}, and Brinksma et al.~\cite{brinksma_FairTesting_1995} proposed a refinement $\FS$ of $\May$, the so-called fair/should equivalence, with a more careful treatment of divergence.
The testing framework makes use of a class of specialized testing processes called \emph{observers} that can perform a special action $\omega$ to indicate a successful test outcome.
This framework is sometimes termed the \textit{external} characterization of testing equivalences.
A slightly different approach is to seek alternative characterizations without introducing external observers for testing, giving rise to the trace equivalence~\cite{denicola_TestingEquivalencesProcesses_1984} and the contextual preorder $\preccurlyeq_{\mathscr{L}}^{c}$ in~\cite{boreale_BasicObservablesProcesses_1999}.
Since these characterizations are more intrinsic and rely solely on the behaviors of the processes themselves, they are also referred to as the \textit{internal} characterizations of testing equivalences.

Over the past few decades, integrating probabilistic mechanisms into traditional computational models has been a major endeavor.
The studies of probabilistic testing equivalences emerged with both theoretical and practical motivations~\cite{cheung07,hierons10,gerhold28,crafa11}.
Randomness has been studied in different process models.
Some well-known models are Markov decision processes (MDP) \cite{baier_PrinciplesModelChecking_2008,etessami_RecursiveMarkovDecision_2015}, randomized Calculus of Communication Systems (RCCS) \cite{fu_ModelIndependentApproach_2021, wu_ProbabilisticWeakBisimulation_2023, zhang_UniformRandomProcess_2019}, probabilistic automata (PA) \cite{cattani_DecisionAlgorithmsProbabilistic_2002,segala_ModelingVerificationRandomized_1995,turrini_PolynomialTimeDecision_2015}, and probabilistic communicating
sequential processes (pCSP)~\cite{deng_SemanticsProbabilisticProcesses_2014}.
One would expect a theory of probabilistic testing equivalences to be \emph{unifying}: it should not depend on the specific underlying model, or, at the very least, it should be easily transferable across different models.
Such a requirement has been the central theme of Fu et al.'s theory~\cite{fu_TheoryInteraction_2016, fu_ModelIndependentApproach_2021, fu_NamePassingCalculus_2015}.
The basic idea is that different interaction models may employ fundamentally different interaction mechanisms; for example, name passing in the $\pi$-calculus~\cite{10.5555/329902} and the handshake communication in CCS are two inherently distinct paradigms. 
However, internal computations are common to all these models. 
If we can abstract away from the specific details of external actions, then our approach becomes readily adaptable to a wide range of models.

In this paper, we aim to extend the unifying approach to the setting of probabilistic testing equivalences.
We select a concrete model, the RCCS model, for a case study, which enables direct comparison with classical testing equivalences on CCS.
We primarily focus on the may equivalence and the fair equivalence, since it is well-known that the may equivalence is compatible with the fair equivalence~\cite{brinksma_FairTesting_1995,natarajan_DivergenceFairTesting_1995}, whereas it is not compatible with the must equivalence~\cite{denicola_TestingEquivalencesProcesses_1984}. 
More specifically, we provide two types of characterizations: the observer-based external characterizations and the unifying internal characterizations.
These characterizations coincide and can be regarded as probabilistic extensions of classical $\May$ and $\FS$.
Our testing equivalences are all \emph{congruence relations}: they are closed under all the operations of our concern.
This is quite significant for a full-fledged interaction model.
In most previous studies, equivalence relations are defined for models that disown the concurrent composition operator or the recursion operator. 

A comprehensive spectrum of behavioral equivalences (including testing equivalences~\cite{denicola_TestingEquivalencesProcesses_1984} and bisimilarities~\cite{milner_CommunicationConcurrency_1989,vanglabbeek_BranchingTimeAbstraction_1996}) has been given in~\cite{vanglabbeek_LinearTimeBranching_1993}.
Inspired by this, we conduct a detailed comparison of these equivalences with respect to their inclusion relations~$\subseteq$, a perspective that offers valuable insight in probabilistic settings.
Beyond comparing different probabilistic equivalences within a single model, a further natural question is whether our approach coincides with existing notions of probabilistic testing equivalences across different models. While research on probabilistic fair testing remains rather limited, probabilistic may equivalence has been extensively studied~\cite{deng_SemanticsProbabilisticProcesses_2014, yi_TestingProbabilisticNondeterministic_1992}. We believe that our notion is consistent with several established definitions, especially those introduced by Deng et al.~\cite{DBLP:journals/lmcs/DengGHM08}.
To substantiate this claim and reinforce the unifying nature of our approach, we conducted a case study on the pCSP model. 
The resulting equivalences indeed corroborate this correspondence.

The interplay between nondeterminism and randomness is a central issue.
Some studies~\cite{CLEAVELAND199993} resolved all nondeterministic choices probabilistically, limiting their applicability to general probabilistic concurrent systems.
The RCCS model in the paper incorporates both nondeterminism and randomness, allowing us to demonstrate a proper treatment of nondeterminism in the presence of randomness. 
To support our unifying testing approach, we develop a novel distribution-based semantics that directly captures the evolution of processes in a probabilistic testing scenario.
It is worth noting that this differs from existing methods, such as the tree-based semantics~\cite{fu_ModelIndependentApproach_2021, wu_AnalyzingDivergenceNondeterministic_2024a, wu_ProbabilisticWeakBisimulation_2023, yi_TestingProbabilisticNondeterministic_1992}, the scheduler-based approach~\cite{turrini_PolynomialTimeDecision_2015}, and the lifting techniques~\cite{deng_SemanticsProbabilisticProcesses_2014, DENG201758}.
In brief, our distribution-based approach exhibits several appealing properties.
It is \emph{simple}, as it avoids intricate constructions such as schedulers and infinite unfolding trees (see Lemma~\ref{lem:linearity_of_probabilistic_transition_strengthened}); \emph{general}, since the theory relies only on the underlying probabilistic labelled transition system rather than model-specific syntactic details (see the case study on pCSP); and \emph{semantically well-founded}, as its distribution-based nature aligns closely with the essence of probabilistic systems, making it broadly deployable in settings where tree-based or scheduler-based approaches are used (see in particular Section~\ref{sec:comparison}).
Further technical comparisons are provided in the main text.

\paragraph{Contributions}

The main contributions of this paper are stated as follows.
\begin{enumerate}
    \item 
    We introduce a distribution-based semantics for the RCCS model that satisfies several well-established requirements.
    Under this semantics, probabilistic transition sequences are linear with respect to convex combinations of distributions (Lemma~\ref{lem:linearity_of_probabilistic_transition_strengthened}).
    We then develop a universal probabilistic testing framework using process predicates and extend the classical concrete testing outcome to a convex set (Lemma~\ref{lem:linearity_of_testing_outcome}). 
    Some approximation properties are discussed in Section~\ref{subsec:approximation}.
    \item
    Within our testing framework, we introduce two parametric probabilistic testing equivalences, the box equivalence $\PXBox$ and the diamond equivalence $\PXDiamond$, as internal characterizations, where $\mathcal{D}$ is a given testing context.
    We prove that $\PXBox$ is strictly finer than $\PXDiamond$ (Theorem~\ref{thm:probabilistic_box_diamond}).
    We define the external characterizations $\PXFS$ and $\PXMay$ and show that they coincide with $\PXBox$ and $\PXDiamond$, respectively (Theorem~\ref{thm:model_indenpent_characterizations}).
    We prove that our testing equivalences are all congruences (Theorem~\ref{thm:congr}). 
    At the same time, we applied our framework to the pCSP model and obtained consistent results.
    \item 
    Finally, we reveal the inclusion relationships between the probabilistic weak bisimilarity ($\PWeak$) and our probabilistic testing equivalences
    (Theorem~\ref{thm:probabilistic_hierarchy}).
    We also show that our internal characterizations are 
    conservative extensions of the ones in~\cite{fu_NamePassingCalculus_2015} (Proposition~\ref{prop:classical_generalization} and~\ref{prop:testing_classical_generalization}), which give alternative proofs for their core results (Proposition~\ref{prop:classical_model_independent_characterization}) in a probabilistic setting.
\end{enumerate}

A schematic overview of our main results is given in Figure~\ref{fig:summary}.

\begin{figure}[t]
    \centering
\begin{tikzpicture}
\tikzset{vertex/.style = {font=\scriptsize}}
\tikzset{edge/.style = {->,> = latex', font=\scriptsize}}
\definecolor{mycolor1}{HTML}{0aa344}
\node[vertex] (b) at (3,0) {$\PWeak$};
\node[vertex] (c) at (6,0) {$\PFS$};
\node[vertex] (d) at (9,0) {$\PBox$};
\node[vertex] (e) at (12,0) {$\PDiamond$};
\node[vertex] (f) at (15,0) {$\PMay$};
\node[vertex] (c1) at (6,1.5) {$\PCFS$};
\node[vertex] (d1) at (9,1.5) {$\PCBox$};
\node[vertex] (e1) at (12,1.5) {$\PCDiamond$};
\node[vertex] (f1) at (15,1.5) {$\PCMay$};
\node[vertex] (c2) at (6,3) {$\FS$};
\node[vertex] (d2) at (9,3) {$\CBox$};
\node[vertex] (e2) at (12,3) {$\CDiamond$};
\node[vertex] (f2) at (15,3) {$\May$};
\draw[edge] (b) to node[above,text = black, scale = 0.6] {$\subsetneq$} node[below,text = black, scale=1] { Theorem~\ref{thm:pweak_vs_pbox}} (c);
\draw[edge, <->] (c) to node[above,text = black, scale=1] {$=$} node[below,text = black, scale=1] { Theorem~\ref{thm:model_indenpent_characterizations}} (d);
\draw[edge] (d) to node[above,text = black, scale=1] {$\subsetneq$} node[below,text = black, scale=1] {Theorem~\ref{thm:probabilistic_box_diamond}} (e);
\draw[edge, <->] (e) to node[above,text = black, scale=1] {$=$} node[below,text = black, scale=1] { Theorem~\ref{thm:model_indenpent_characterizations}} (f);

\draw[violet, dashed, opacity=0.25] (4.7,-0.75) rectangle (7.3,4);
\node[vertex,text = violet, scale=1,align=center] (ex1) at (6,3.85) {external};
\node[vertex,text = violet, scale=1,align=center] (ex1) at (6,3.6) {characterizations};
\draw[violet, dashed, opacity=0.25] (13.7,-0.75) rectangle (16.3,4);
\node[vertex,text = violet, scale=1,align=center] (ex2) at (15,3.85) {external};
\node[vertex,text = violet, scale=1,align=center] (ex2) at (15,3.6) {characterizations};
\draw[mycolor1, dashed, opacity=0.25] (7.7,-0.75) rectangle (13.3,4);
\node[vertex,text = mycolor1, scale=1] (in) at (10.5,3.65) {internal characterizations};

\draw[blue, dashed, opacity=0.5] (1,-0.35) rectangle (17,0.35);
\node[vertex,text = blue, scale=1] (p) at (2.7,-0.55) {probabilistic testing context};

\draw[red, dashed, opacity=0.5] (1,1.15) rectangle (17,1.85);
\node[vertex,text = red, scale=1] (de) at (3,1.5) {classical testing context};

\draw[cyan, dashed, opacity=0.5] (1,2.65) rectangle (17,3.35);
\node[vertex,text = cyan, scale=1] (cl) at (3,3) {classical equivalences};

\draw[edge, <->] (c1) to node[above,text = black, scale=1] {$=$} node[below,text = black, scale=1] { Theorem~\ref{thm:model_indenpent_characterizations}} (d1);
\draw[edge] (d1) to node[above,text = black, scale=1] {$\subsetneq$} node[below,text = black, scale=1] { Theorem~\ref{thm:probabilistic_box_diamond}} (e1);
\draw[edge, <->] (e1) to node[above,text = black, scale=1] {$=$} node[below,text = black, scale=1] { Theorem~\ref{thm:model_indenpent_characterizations}} (f1);
\draw[edge, <->] (c2) to node[above,text = black, scale=1] {$=$} node[below,text = black, scale=1] { Proposition~\ref{prop:classical_model_independent_characterization}} (d2);
\draw[edge] (d2) to node[above,text = black, scale=1] {$\subsetneq$} node[below,text = black, scale=1] { Proposition~\ref{prop:classical_model_independent_characterization}} (e2);
\draw[edge, <->] (e2) to node[above,text = black, scale=1] {$=$} node[below,text = black, scale=1] { Proposition~\ref{prop:classical_model_independent_characterization}} (f2);

\node[vertex] (p) at (10.5,0.35) {};
\node[vertex] (de) at (10.5,1.15) {};
\draw[edge] (p) to node[left,text = black, scale=1] {pairwise $\subsetneq$~~} node[right,text = black, scale=1,align=center] {~Proposition~\ref{prop:testing_classical_generalization}\\~Proposition~\ref{prop:classical_generalization}} (de);
\node[vertex] (de) at (10.5,1.85) {};
\node[vertex] (cl) at (10.5,2.65) {};
\draw[edge] (de) to node[left,text = black, scale=1, align=center] {pairwise conservative\\generalization~} node[right,text = black, scale=1, align=center] {~Proposition~\ref{prop:testing_classical_generalization}\\ ~Proposition~\ref{prop:classical_generalization}} (cl);
\end{tikzpicture}
    \caption{The inclusion relations among probabilistic equivalences presented in this paper, where the arrow from one equivalence to the other indicates that the former is finer than the latter. The superscripts $\Delta$ and $p$ denote, respectively, classical testing contexts without probabilistic capabilities and probabilistic testing contexts with full probabilistic power.}
    \label{fig:summary}
\end{figure}
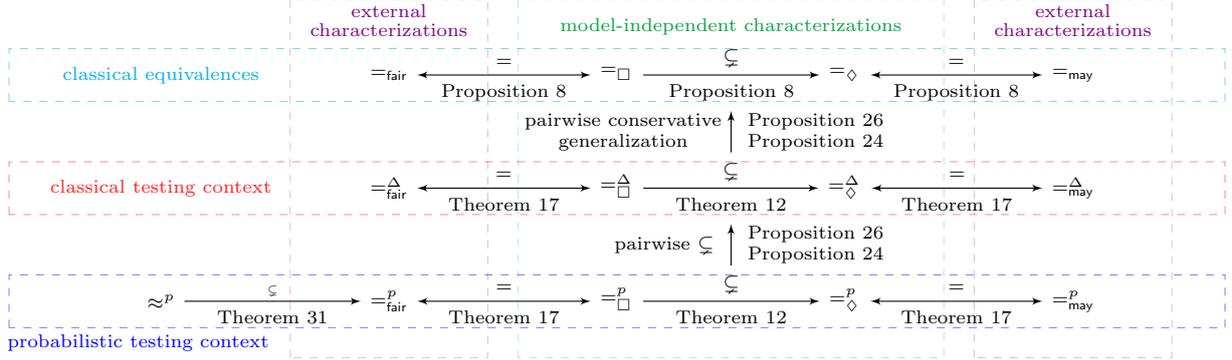

\paragraph*{Organization}
Section~\ref{sec:preliminary} revisits the CCS model and classical testing theory.
Section~\ref{sec:rccs_model} introduces the RCCS model and its distribution-based semantics.
Section~\ref{sec:testing_over_distributions} develops our general predicate-based testing framework for probabilistic process models.
Section~\ref{sec:model_independent} formally introduces probabilistic diamond/box equivalences as well as their external characterizations.
A case study on pCSP is also presented.
Section~\ref{sec:comparison} compares our work with related work and establishes the spectrum shown in Figure~\ref{fig:summary}. 
Section~\ref{sec:conclusion} concludes.
\section{Revisiting the classical testing theory}
\label{sec:preliminary}

We start by revisiting the foundational concepts of Milner's CCS model and the theory of classical testing equivalence, which should be familiar to the readers. 
This review serves as a basis for introducing the RCCS model and distribution-based semantics, which will be discussed in detail in Section~\ref{sec:rccs_model}.

\subsection{The CCS model}
\label{subsec:ccs}

Let $\Chan$ be an infinite set of channels, ranged over by $a,b,c$, and $\overline{\Chan} =\left\{\overline{a}\mid a \in \Chan\right\}$.
Let $\L = \Chan \cup \overline{\Chan}$, ranged over by $\ell$, be the set of labels representing the external actions.
We usually use $L$ to denote a finite subset of $\Chan$.
We use a special symbol $\tau \notin \L$ to represent the internal action (also referred to as a silent action).
The set of actions is $\Act = \L \cup \left\{\tau\right\}$, ranged over by $\alpha, \beta, \gamma$.
The empty string over $\Act$ is denoted by $\varepsilon$.

The grammar of the CCS model is defined as:
\begin{equation}
\label{eq:ccs_grammer}
    S, T:=  \mathbf{0} ~\Big{|}~ X ~\Big{|}~ \sum_{i \in I} \alpha_i.T_i ~\Big{|}~  S\mid T ~\Big{|}~ (L)T ~\Big{|}~ \fix X.T,\tag{$*$}
\end{equation}
where the non-empty index set $I$ is finite.
In (\ref{eq:ccs_grammer}), $\mathbf{0}$ is the \emph{nil term} that cannot perform any action and $X$ is a \emph{process variable}.
The remaining four combinators are used to generate new terms from existing ones. 
Their intuitive meanings are explained as follows:
\begin{itemize}
    \item $\sum_{i \in I} \alpha_i.T_i$ (\emph{nondeterministic choice term}) is a term that may in a nondeterministic manner choose to perform action $\alpha_i$ and then continue as $T_i$.
    \item $S \mid T$ (\emph{parallel composition term}) denotes the concurrent execution of $S$ and $T$.
    \item $(L)T$ (\emph{localization term}) behaves like $T$, except that it cannot perform actions in the set $L$.
    \item $\fix X.T$ (\emph{fixpoint term}) represents a term defined recursively by the equation $X = T$, where $T$ typically contains some occurrences of the process variable $X$.
    A substitution $T\{P_1/X_1, \ldots, P_n/X_n\}$ (abbreviated as $T\{\widetilde{P}/\widetilde{X}\}$ when $n$ is clear) denotes the term obtained by simultaneously replacing the process variables $X_1, \ldots, X_n$ in $T$ by $P_1, \ldots, P_n$, respectively.
\end{itemize}

A trailing $\mathbf{0}$ that appears at the end of a term is often omitted, e.g., $\tau.a$ represents $\tau.a.\mathbf{0}$.
If $L=\{a_1,a_2,\dots,a_k\}$, then $(L)T$ can be simplified as $(a_1a_2\dots a_k)T$.
Sometimes we will use the infix notation of $\sum$ to specify particular summands in nondeterministic choice terms, writing for example, $\sum_{i \in I'} \alpha_i. T_i + \beta.T' + \gamma.T''$.
As usual, the variable $X$ in $\fix X.T$ is \emph{bound}. 
A variable in a term is \emph{free} if it is not bound. 
A term is a \emph{process} if it contains no free variables. 
We write $A,B,P,Q$ for processes. 
The set of all CCS processes is denoted by $\PCCS$.
Given a process $P\in\PCCS$, we use $\Chan(P)$ to denote the channels in $P$.
The channels in $L$ are said to be \emph{local} in $(L)P$.
A channel in a term is \emph{global} if it is not local.
A \emph{renaming function} $f:\Chan\to\Chan$ is a partial function with its domain $\mathsf{dom}(f)$ finite.
We use $P[f]$ to denote the process obtained from $P$ by renaming all global names in $\Chan(P)\cap \mathsf{dom}(f)$ using $f$.
A renaming function that maps $a$ to $b$ can be simply written as $[a\mapsto b]$.

The operational semantics of $\PCCS$ is given by the \emph{labeled transition system} (LTS for short) in Figure \ref{fig:LTS}, where $\ell 
\in \L$, $\alpha_i, \beta \in \Act$ and the transition relation $\longrightarrow\;\subseteq \PCCS\times \Act \times \PCCS$.
We use $\Longrightarrow$ to denote the reflexive and transitive closure of $\myrightarrow{\tau}$.
Given a process $P$, an \emph{internal action sequence} of $P$ is a possibly infinite sequence of the form $P\myrightarrow{\tau}P_1\myrightarrow{\tau}\cdots$.
A process $P$ is \emph{divergent} if there exists an infinite internal action sequence starting from $P$.

\begin{figure}[t]	
    \vspace{-5mm}
    \begin{center}
        \begin{displaymath}
        \frac{}{~\sum_{i \in I} \alpha_i.T_i \myrightarrow{\alpha_i} T_i~} \qquad
        \frac{T \myrightarrow{\beta} T'}{~(L)T \myrightarrow{\beta}(L)T' ~} ~\beta, \overline{\beta} \notin L \qquad
        \frac{T\left\{\fix X.T/X\right\} \myrightarrow{\beta} T'}{\fix X.T \myrightarrow{\beta} T'}
        \end{displaymath}	
        \begin{displaymath}
            \frac{~S \myrightarrow{\beta} S'~}{~S\mid T \myrightarrow{\beta} S'\mid T~} \qquad
            \frac{~T \myrightarrow{\beta} T'~}{~S\mid T \myrightarrow{\beta} S\mid T'~} \qquad
            \frac{~S\myrightarrow{\ell} S' \quad T\myrightarrow{\overline{\ell}} T'~}{S\mid T\myrightarrow{\tau} S' \mid T'}
        \end{displaymath}	
    \end{center}
    \vspace{-2mm}
    \caption{LTS for $\mathrm{CCS}$.}
    \label{fig:LTS}
    \vspace{-2mm}
\end{figure}

\subsection{Classical testing equivalence}
\label{subsec:may_must}

Given a set $\S$, we use $\R$ and $\E$ to denote a relation and an equivalence relation over $\S$, respectively.
We write $\S / \E$ for the set of equivalence classes induced by $\E$.
The equivalence class containing $x\in \S$ is denoted by $[x]_{\E}$.
The subscript is often omitted when $\E$ is clear from the context.
The testing machinery of De Nicola and Hennessy~\cite{denicola_TestingEquivalencesProcesses_1984} can be summarized as follows.

\begin{enumerate}
    \item 
    A special action $\omega \notin \Act$ is used to report success.
    \item 
    \emph{Observers}, ranged over by $O$, are obtained from CCS processes by replacing some occurrences of $\mathbf{0}$ by the process $\boldsymbol{\omega}:=\omega.\mathbf{0}$.
    The operational semantics of observers are defined under the action set $\Act\cup\left\{\omega\right\}$.
    \item 
    A \emph{testing  sequence $s^{P\mid O}_{\omega}$ of process $P$ by observer $O$} is a (possibly infinite) internal action sequence of $P \mid O$ satisfies that
    \begin{enumerate}
        \item 
        no non-terminating state of $s^{P\mid O}_{\omega}$ can perform the $\omega$ action, and
        \item 
        if $s^{P\mid O}_{\omega}$ is a finite sequence, then the final state of $s^{P\mid O}_{\omega}$ either can perform the $\omega$ action or cannot perform any $\tau$ action.
    \end{enumerate}
    \item 
    A testing  sequence $s^{P\mid O}_{\omega}$ of $P$ by $O$ is \emph{DH-successful} if $s^{P\mid O}_{\omega}$ is a finite sequence with its final state being able to perform the $\omega$ action.
    \item The test of $P$ by $O$ is called \emph{fair-successful} if whenever $P\mid O \Longrightarrow P'\mid O'$ then some testing sequence of $P'$ by $O'$ is DH-successful.
\end{enumerate}

We then introduce a set of predicates for binary relations to describe the relative observability of a pair of processes under any observer, tailored to meet various requirements.

\begin{definition}
Suppose $\R$ is a relation on $\PCCS$.
\begin{enumerate}
    \item 
    $\R$ satisfies \emph{may predicate} if $P\,\R\, Q$ implies that, for every observer $O$, \emph{some} testing sequence of $P$ by $O$ is DH-successful if and only if \emph{some} testing sequence of $Q$ by $O$ is DH-successful.
    \item 
    $\R$ satisfies \emph{must predicate} if $P\,\R\, Q$ implies that, for every observer $O$, \emph{all} testing sequences of $P$ by $O$ are DH-successful if and only if \emph{all} testing sequences of $Q$ by $O$ are DH-successful.
    \item 
    $\R$ satisfies \emph{fair predicate} if $P\,\R\, Q$ implies that, for every observer $O$, the test of $P$ by $O$ is fair-successful if and only if 
    the test of $Q$ by $O$ is fair-successful.
\end{enumerate} 
\end{definition}

It is straightforward to see that the largest relation defined by each of these predicates is reflexive, symmetric, transitive, and is closed under the union operation, hence the following.

\begin{definition}[Classical testing equivalences \cite{denicola_TestingEquivalencesProcesses_1984}]
\label{def:may_equivalence}
~
\begin{enumerate}
    \item The \emph{may equivalence} $\May$ is the largest relation on $\PCCS$ that satisfies the may predicate.
    \item The \emph{must equivalence} $\Must$ is the largest relation on $\PCCS$ that satisfies the must predicate.
    \item The \emph{fair equivalence} $\FS$ is the largest relation on $\PCCS$ that satisfies the fair predicate.
\end{enumerate}
\end{definition}

Intuitively, an observer $O$ abstracts the external testing behavior, and different testing sequences capture the possible outcomes of applying a test under different nondeterministic branches.
The may and must testing check whether the system \emph{can} or \emph{must} eventually pass the test, respectively.
In contrast, the fair testing checks if a system remains capable of passing the test from \emph{any} reachable state.
The meanings of the corresponding equivalence relations naturally follow from these interpretations.
The relationships among the three testing equivalences are summarized below.

\begin{proposition}[\cite{denicola_TestingEquivalencesProcesses_1984,brinksma_FairTesting_1995,natarajan_DivergenceFairTesting_1995}]
    The three testing equivalences introduced above satisfy the following statements.
    \begin{enumerate}
        \item $\May\;\not\subseteq\;\Must$ and $\Must\;\not\subseteq\;\May$.
        \item $\FS\;\subsetneq\;\May$.
    \end{enumerate}
\end{proposition}

As a result, may and must equivalence are incomparable, and the same holds for fair and must equivalence.
Formal proofs can be found in the literature; here, we present several illustrative examples in Example~\ref{ex:classical_testing} to build intuition, which will help understand the spectrum of testing equivalences in the probabilistic setting discussed later.

\begin{example}
\label{ex:classical_testing}
Consider the four processes shown in Figure~\ref{fig:classical_testing}
\begin{equation}
    P_1:=\tau.a,\; P_2:=\tau.a+\tau.(\fix X.\tau.X),\; P_3 := \fix X.(\tau.a+\tau.X), \text{and~}P_4:=\fix X.\tau.X.
\end{equation}

\begin{figure}[htb]
\centering
\begin{minipage}[c]{.2\textwidth}
\centering
\begin{tikzpicture}
\tikzset{vertex/.style = {}}
\tikzset{edge/.style = {->,> = latex'}}
\node[vertex] (a) at  (0,0) {\footnotesize$a$};
\node[vertex] (b) at  (0,1.5) {\footnotesize$P_1$};
\draw[edge] (b) to node[left] {\footnotesize$\tau$} (a);
\end{tikzpicture}
\end{minipage}
\hspace{0.03\textwidth}
\begin{minipage}[c]{.22\textwidth}
\centering
\begin{tikzpicture}
\tikzset{vertex/.style = {}}
\tikzset{edge/.style = {->,> = latex'}}
\node[vertex] (a) at  (0,0) {\footnotesize$a$};
\node[vertex] (b) at  (1,1) {\footnotesize$P_2$};
\node[vertex] (c) at (2,0) {\footnotesize$\cdot$};
\draw[edge] (b) to node[left=3, pos=0] {\footnotesize$\tau$} (a);
\draw[edge] (b) to node[right=3, pos=0] {\footnotesize$\tau$} (c);
\draw[edge] (c) to[loop below, in=300, out=240, looseness=8] node {\footnotesize$\tau$} (c);
\end{tikzpicture}
\end{minipage}
\hspace{0.03\textwidth}
\begin{minipage}[c]{.2\textwidth}
\centering
\begin{tikzpicture}
\tikzset{vertex/.style = {}}
\tikzset{edge/.style = {->,> = latex'}}
\node[vertex] (a) at  (0,0) {\footnotesize$a$};
\node[vertex] (b) at  (0,1.5) {\footnotesize$P_3$};
\draw[edge] (b) to node[left] {\footnotesize$\tau$} (a);
\draw[edge] (b) to[loop above, in=60, out=120, looseness=6] node {\footnotesize$\tau$} (b);
\end{tikzpicture}
\end{minipage}
\hspace{0.03\textwidth}
\begin{minipage}[c]{.15\textwidth}
\centering
\begin{tikzpicture}
\tikzset{vertex/.style = {}}
\tikzset{edge/.style = {->,> = latex'}}
\node[vertex] (b) at  (0,0) {\footnotesize$P_4$};
\draw[edge] (b) to[loop above, in=60, out=120, looseness=8] node {\footnotesize$\tau$} (b);
\end{tikzpicture}
\end{minipage}
\caption{The counterexamples of classical testing equivalences.}
\label{fig:classical_testing}
\end{figure}
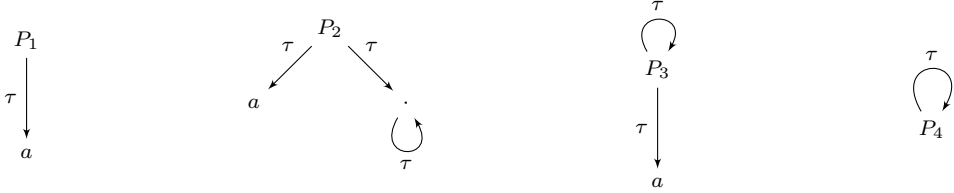 
We have the following observations.
\begin{itemize}
    \item 
    $P_1$, $P_2$ and $P_3$ all can perform the external action $a$, whereas $P_4$ cannot. Thus only $P_4$ \emph{cannot} pass the test $\overline{a}.\boldsymbol{\omega}$, and $P_1\May P_2\May P_3 \not\May P_4$.
    \item 
    $P_2$, $P_3$ and $P_4$ all have infinite internal action sequences, whereas $P_1$ does not. Thus only $P_1$ \emph{must} pass the test $\overline{a}.\boldsymbol{\omega}$, and $P_1 \not\Must P_2\Must P_3\Must P_4$.
    \item 
    For any $Q$ such that $P_3 \Longrightarrow Q$, we have $P_1 \May Q$, and they can both pass the test $\overline{a}.\boldsymbol{\omega}$. Thus the test of $P_1$ (and $P_3$) by $\overline{a}.\boldsymbol{\omega}$ is fair-successful.
    Since $P_2 \myrightarrow{\tau} P_4$ and $P_4$ can only diverge, neither the test of $P_2$ nor the test of $P_4$ by any observer is fair-successful.
    Therefore, $P_1\FS P_3\not\FS P_2\FS P_4$.

\end{itemize}

In summary, $\May$ and $\Must$, as well as $\FS$ and $\Must$, are incomparable under inclusion; however, the pairwise intersections of these three relations are nontrivial (containing process pairs beyond just the reflexive relation).
Notably, the proof of $\FS\;\subseteq\;\May$ is slightly non-trivial, and Theorem~\ref{thm:probabilistic_box_diamond} in the following section provides a probabilistic version of this inclusion.

\end{example}

From Example~\ref{ex:classical_testing}, we can also observe that $\Must$ treats divergence rather coarsely: all divergent processes are considered must-equivalent.
In contrast, $\FS$ inspects every state along divergent paths, thereby partially addressing this limitation.
Due to this drawback and the incomparability discussed above, the remainder of this paper will focus on probabilistic may and fair testing.

\section{Randomized CCS with distribution-based semantics}
\label{sec:rccs_model}

In this section, we formalize the randomized CCS model and present a new distribution-based semantics. 
To begin, we introduce the necessary mathematical notations related to probability.
A \textit{(discrete) distribution} over a countable set $\S\neq\emptyset$ is a function $\rho: \S \to [0,1]$ such that $\sum_{x \in \S} \rho(x) = 1$, whose \emph{support} is $\Supp(\rho):=\left\{x\in\S: \rho(x)>0\right\}$.
We denote by $\Distr(\S)$ the set of distributions over $\S$ with finite support.
We also express a distribution $\rho \in\Distr(\S)$ as $\left\{(x, \rho(x)): x \in \Supp(\rho)\right\}$.
For $\X \subseteq \S$, we define $\rho(\X) := \sum_{x \in \X} \rho(x)$.
For any $x\in\S$, the \emph{Dirac distribution} $\delta_x$ w.r.t. $x$ is defined by $\delta_x(x)=1$ and $\delta_x(y)=0$ for all $y\ne x$.
Given a finite family of distributions $\left\{\rho_i\right\}_{i\in I}$ and real numbers $\{p_i\}_{i\in I}$ such that $p_i\in [0,1]$ for each $i\in I$ and $\sum_{i\in I}p_i=1$, we say that $\sum_{i\in I}p_i\rho_i$ is a \emph{convex combination} of $\{\rho_i\}_{i\in I}$.

\subsection{Randomized CCS model}
\label{subsec:sytax}

The syntax of the randomized CCS model, RCCS~\cite{fu_ModelIndependentApproach_2021}, is defined by
\begin{equation}
\label{eq:grammer_rccs}
    S, T:=  \mathbf{0} ~\Big{|}~ X ~\Big{|}~ \sum_{i \in I} \alpha_i.T_i ~\Big{|}~ \bigoplus_{i \in I}p_i\tau.T_i ~\Big{|}~ S\mid T ~\Big{|}~ (L)T ~\Big{|}~ \fix X.T,\tag{$**$}
\end{equation}
where $I$ is a non-empty finite set.
In (\ref{eq:grammer_rccs}), $\bigoplus_{i \in I}p_i\tau.T_i$ is a newly introduced \emph{probabilistic choice term}, where $p_i\in(0,1)$ for each $i\in I$ and $\sum_{i\in I}p_i=1$.
The $\bigoplus$ operator is also commonly written in infix notation.
Intuitively, this term represents an internal probabilistic transition within the system, which can evolve into $T_i$ with probability $p_i$.
The meanings of the remaining terms are consistent with those introduced in Section~\ref{subsec:ccs}.
An RCCS term $S$ is $n$-ary if $S$ contains at most $n$ free process variables, and a \emph{process} is a $0$-ary term.
The set of all RCCS processes is denoted by $\PRCCS$.
Hereafter, unless otherwise specified, distributions $\mu, \nu, \rho,\dots$, refer specifically to the elements of $\Distr(\PRCCS)$;
$n$-ary distributions $\vartheta, \varsigma, \varrho, \dots$, refer specifically to distributions over RCCS terms such that all terms in their supports are $n$-ary.
Some operations are extended from processes to ($n$-ary) distributions as follows.
\begin{itemize}
    \item The composition of $\rho_1$ and $\rho_2$ is 
$
\rho_1\mid\rho_2:=\sum_{A\in\Supp(\rho_1)}\sum_{B\in\Supp(\rho_2)}\rho_1(A)\rho_2(B)\delta_{A\mid B}.
$
    \item The localization of $\rho$ to $L$ is $(L)\rho:=\sum_{A\in\Supp(\rho)}\rho(A)\delta_{(L)A}$.
    \item The renaming of $\rho$ by $f$ is $\rho[f]:=\sum_{A\in\Supp(\rho)}\rho(A)\delta_{A[f]}$.
    \item The channels appearing in $\rho$ is defined by $\Chan(\rho):=\bigcup_{A\in\Supp(\rho)}\Chan(A)$.
    \item The substitution $\vartheta\{\widetilde{P}/\widetilde{X}\}:=\sum_{R\in\Supp(\vartheta)}\vartheta(R)\delta_{R\{\widetilde{P}/\widetilde{X}\}}$.
\end{itemize}

\subsection{A new distribution-based semantics}
\label{subsec:semantics}

The operational semantics of $\mathrm{RCCS}$ consists of two parts defined by the \emph{probabilistic labeled transition systems} (pLTS for short) given in Figure \ref{fig:pLTS1} and  Figure \ref{fig:pLTS2}. 
\begin{itemize}
    \item To Part I of the pLTS (Figure \ref{fig:pLTS1}), the transition relation $\longrightarrow\;\subseteq \PRCCS\times \Act \times \Distr(\PRCCS)$, where $\ell\in \L$, $\alpha_i, \beta \in \Act$.
    Notably, $\bigoplus_{i \in I} p_i \tau.T_i$ can evolve to the process $T_i$ with probability $p_i$, thereby reaching the distribution $\sum_{i\in I}p_i\delta_{T_i}$.
    Other rules are consistent with the classical CCS model (refer to Figure~\ref{fig:LTS}).
    This part is defined from the perspective of processes. 
    In other words, it only specifies how processes evolve into distributions through a single transition.  
    \item Part II of the pLTS (Figure \ref{fig:pLTS2}) then defines how distributions will continue to evolve further.
    Intuitively, $\mu\xrightarrow{(\alpha,p)} \nu$ means that $\nu$ is obtained from $\mu$ by activating some transition $P\myrightarrow{\alpha}\rho$ for process $P\in\Supp(\mu)$ with probability $p$, and $P$ remains unchanged with probability $1-p$.
    Beyond its theoretical convenience (Lemma~\ref{lem:linearity_of_probabilistic_transition_strengthened}), this approach is also motivated by practical scenarios, such as in queueing theory, where arrival events generally conform to specific distributions.
\end{itemize}
    
\begin{figure}[t]	
    \begin{center}
        \begin{displaymath}
        \frac{}{~\sum_{i \in I} \alpha_i.T_i \myrightarrow{\alpha_i} \delta_{T_i}~} \quad
        \frac{}{~\bigoplus_{i \in I}p_i\tau.T_i \myrightarrow{\tau} \sum_{i\in I}p_i\delta_{T_i}~} \quad 
        \frac{~S \myrightarrow{\beta} \rho~}{~S\mid T \myrightarrow{\beta} \rho\mid\delta_{T}~}
        \end{displaymath}	
        \begin{displaymath}
            \frac{~T \myrightarrow{\beta} \rho~}{~S\mid T \myrightarrow{\beta} \delta_{S}\mid \rho~} \quad
            \frac{~S\myrightarrow{\ell}\rho_1 \quad T\myrightarrow{\overline{\ell}}\rho_2~}{S\mid T\myrightarrow{\tau}\rho_1 \mid \rho_2} \quad
            \frac{T \myrightarrow{\beta} \rho}{~(L)T \myrightarrow{\beta}(L)\rho ~} ~\beta,\overline{\beta} \notin L \quad
            \frac{T\{\fix X.T/X\} \myrightarrow{\beta} \rho}{\fix X.T \myrightarrow{\beta} \rho}
        \end{displaymath}
    \end{center}
    \vspace{-3mm}
    \caption{pLTS for $\mathrm{RCCS}$ - Part I.}
    \label{fig:pLTS1}
\end{figure}

\begin{figure}[t]	
    \small
    \begin{center}
   $\mu\xrightarrow{(\alpha,p)} \nu \quad \left(\begin{array}{c}
        \text{where~} \mu,\;\nu\in\Distr(\PRCCS),\;\alpha\in \Act,\; p\in(0,1],
        \text{and there exists a process }\\P\in\Supp(\mu)\text{~such that~}
        P\myrightarrow{\alpha}\rho \text{~for some $\rho$ and~}\nu=\mu+\mu(P)p(\rho-\delta_P).
    \end{array}\right)$  
    \end{center}
    \vspace{-3mm}
    \caption{pLTS for $\mathrm{RCCS}$ - Part II.}
    \label{fig:pLTS2}
\end{figure}

It can be also observed that distributions with finite support, which cover most practical scenarios, are closed under finite-step transitions.
Therefore, it is well-justified to consider only the distributions in $\Distr(\PRCCS)$ in probabilistic testing.
The following example helps comprehend our new semantics.
\begin{example}
\label{ex:semantics}
Let $Q_1:=\frac{1}{2}\tau.a\oplus\frac{1}{2}\tau.b$, $Q_2:=\fix X.\left(\frac{1}{3}\tau.a\oplus\frac{1}{3}\tau.b\oplus\frac{1}{3}\tau.X\right)$, $Q_3:=\tau.a+\tau.b.a$ and $Q_4:=\frac{1}{2}\tau.a\oplus\frac{1}{2}\tau.b.a$, as shown in Figure~\ref{fig:exp_semantics}.

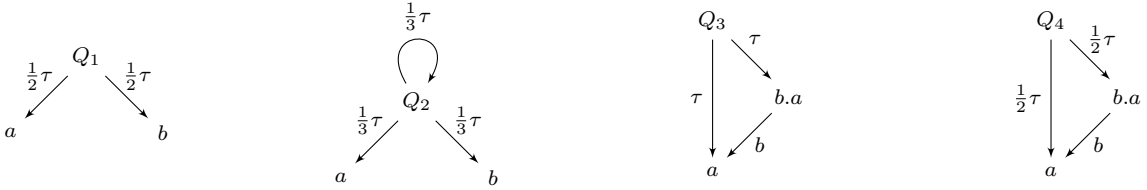
\begin{figure}[htb]
\centering
    \begin{minipage}[c]{.2\textwidth}
        \centering
        
\begin{tikzpicture}
\tikzset{vertex/.style = {}}
\tikzset{edge/.style = {->,> = latex'}}
\node[vertex] (a) at  (0,0) {\footnotesize$a$};
\node[vertex] (b) at  (1,1) {\footnotesize$Q_1$};
\node[vertex] (c) at  (2,0) {\footnotesize$b$};
\draw[edge] (b) to node[left=2, pos=0] {\footnotesize$\frac{1}{2}\tau$} (a);
\draw[edge] (b) to node[right=3, pos=0] {\footnotesize$\frac{1}{2}\tau$} (c);
\end{tikzpicture}

    \end{minipage}
\hspace{0.05\textwidth}
    \begin{minipage}[c]{.2\textwidth}
        \centering

\begin{tikzpicture}
\tikzset{vertex/.style = {}}
\tikzset{edge/.style = {->,> = latex'}}
\node[vertex] (a) at  (0,0) {\footnotesize$a$};
\node[vertex] (b) at  (1,1) {\footnotesize$Q_2$};
\node[vertex] (c) at  (2,0) {\footnotesize$b$};
\draw[edge] (b) to node[left=2, pos=0] {\footnotesize$\frac{1}{3}\tau$} (a);
\draw[edge] (b) to node[right=3, pos=0] {\footnotesize$\frac{1}{3}\tau$} (c);
\draw[edge] (b) to[loop above, in=60, out=120, looseness=8] node {\footnotesize$\frac{1}{3}\tau$} (b);
\end{tikzpicture}

    \end{minipage}
\hspace{0.05\textwidth}
    \begin{minipage}[c]{.2\textwidth}
        \centering

\begin{tikzpicture}
\tikzset{vertex/.style = {}}
\tikzset{edge/.style = {->,> = latex'}}
\node[vertex] (a) at  (1,-1) {\footnotesize$a$};
\node[vertex] (b) at  (1,1) {\footnotesize$Q_3$};
\node[vertex] (c) at  (2,0) {\footnotesize$b.a$};
\draw[edge] (b) to node[left] {\footnotesize$\tau$} (a);
\draw[edge] (b) to node[right=3, pos=0] {\footnotesize$\tau$} (c);
\draw[edge] (c) to node[right=2, pos=0.7] {\footnotesize$b$} (a);
\end{tikzpicture}

    \end{minipage}
\hspace{0.05\textwidth}
    \begin{minipage}[c]{.2\textwidth}
        \centering

\begin{tikzpicture}
\tikzset{vertex/.style = {}}
\tikzset{edge/.style = {->,> = latex'}}
\node[vertex] (a) at  (1,-1) {\footnotesize$a$};
\node[vertex] (b) at  (1,1) {\footnotesize$Q_4$};
\node[vertex] (c) at  (2,0) {\footnotesize$b.a$};
\draw[edge] (b) to node[left] {\footnotesize$\frac{1}{2}\tau$} (a);
\draw[edge] (b) to node[right=3, pos=0] {\footnotesize$\frac{1}{2}\tau$} (c);
\draw[edge] (c) to node[right=2, pos=0.7] {\footnotesize$b$} (a);
\end{tikzpicture}

    \end{minipage}
\vspace{-3mm}
\caption{Examples of the distribution-based semantics.}
\label{fig:exp_semantics}
\end{figure}

\begin{itemize}
    \item 
    According to the first part of the pLTS in Figure \ref{fig:pLTS1}, one has that 
    \begin{equation}
        \begin{aligned}
            &Q_1\myrightarrow{\tau}\rho_1:=\frac{1}{2}\delta_a+\frac{1}{2}\delta_b,\quad Q_2\myrightarrow{\tau}\rho_2:=\frac{1}{3}\delta_a+\frac{1}{3}\delta_b +\frac{1}{3}\delta_{Q_2},\\
            &Q_3\myrightarrow{\tau}\rho_3:=\delta_a,\quad 
            Q_3\myrightarrow{\tau}\rho_3':=\delta_{b.a},\quad
            Q_4\myrightarrow{\tau}\rho_4:=\frac{1}{2}\delta_a+\frac{1}{2}\delta_{b.a}.
        \end{aligned}
    \end{equation}
    Processes $Q_3$ and $Q_4$ illustrate a crucial distinction between the operators $+$ and $\oplus$: $\delta_{Q_3}$ can reach the distribution $\delta_{b.a}$, while in any distribution reachable from $\delta_{Q_4}$, the probability of $b.a$ never exceeds $\frac{1}{2}$.
    \item 
    Let $\mu:=\left\{\left(Q_1,\frac{1}{2}\right), \left(Q_2,\frac{1}{2}\right)\right\}$.
    According to the second part of the pLTS in Figure \ref{fig:pLTS2}, if we activate $Q_1\myrightarrow{\tau}\rho_1$ with probability $\frac{1}{2}$, then we have 
    \begin{equation}
        \mu\xrightarrow{\left(\tau, \frac{1}{2}\right)}\mu+\frac{1}{2}\mu(Q_1)\left(\rho_1-\delta_{Q_1}\right)=
        \frac{1}{4}\delta_{Q_1}+\frac{1}{2}\delta_{Q_2}+\frac{1}{8}\delta_a+\frac{1}{8}\delta_b.
    \end{equation}
    If we activate $Q_2\myrightarrow{\tau}\rho_2$ with probability $\frac{1}{2}$, then we have 
    \begin{equation}
        \mu\xrightarrow{\left(\tau, \frac{1}{2}\right)}\mu+\frac{1}{2}\mu(Q_2)\left(\rho_2-\delta_{Q_2}\right)=
        \frac{1}{2}\delta_{Q_1}+\frac{1}{3}\delta_{Q_2}+\frac{1}{12}\delta_a+\frac{1}{12}\delta_b.
    \end{equation}
\end{itemize}

\end{example}

\subsection{Linearity over probabilistic transition sequences}
\label{subsec:linearity}

Given a word $\pi = (\alpha_1,p_1) \cdots (\alpha_k,p_k)$ with length $\vert\pi\vert:=k$, and two distributions $\mu,\nu$, 
a \emph{(probabilistic) transition sequence} $\mu \myrightarrow{\pi} \nu$ takes the form $\mu \xrightarrow{(\alpha_1,p_1)} \dots \xrightarrow{(\alpha_k,p_k)} \nu$.
We refer to $\pi$ as a \emph{witness} and $\nu$ as a \emph{descendant} of $\mu$.
If $\pi\in(\{\tau\}\times(0,1])^{*}$, we say that it is an \emph{internal transition sequence} and further abbreviate it to $\mu\rightsquigarrow\nu$, which can be viewed as a probabilistic extension of the classical silent transition sequence $\Longrightarrow$. 
If $\pi\in(\Act\times\{1\})^{*}$, we say it is \emph{degenerate} and simply treat $\pi$ as identical to its action sequence, i.e., a word in $\Act^{*}$.
For clarity, we shall abbreviate $(a,1)$ to $a$. 

\begin{example}
    \label{ex:probabilistic_transition}
    We can observe that the behaviors of the processes $Q_1$ and $Q_2$ shown in Example~\ref{ex:semantics} are ``consistent'' under our semantics. 
    For instance, $\delta_{Q_1}$ can simulate $\delta_{Q_2}\myrightarrow{\tau}\rho_2$ by $\delta_{Q_1}\xrightarrow{\left(\tau, \frac{2}{3}\right)}\rho_2':=\frac{1}{3}\delta_a+\frac{1}{3}\delta_b +\frac{1}{3}\delta_{Q_1}$.
    On the other hand, we can prove by induction that 
    \begin{equation}
        \label{eq:Q2}
        \delta_{Q_2}\rightsquigarrow\nu_k:=\left\{\left(Q_2, \frac{1}{3^k}\right), \left(a, \frac{3^k-1}{2\cdot 3^k}\right), \left(b, \frac{3^k-1}{2\cdot 3^k}\right)\right\}.
    \end{equation}
    Since $\lim_{k\to\infty}\nu_k=\rho_1$, we see that $\delta_{Q_2}$ can ``almost'' simulate $\delta_{Q_1}\myrightarrow{\tau}\rho_1$.
    In fact, they are equivalent with respect to the probabilistic weak bisimilarity $\PWeak$, to be introduced in Section~\ref{sec:comparison}.
\end{example}

Since distributions are closed under convex combinations, it is natural that the probabilistic transition sequences are also closed under convex combinations in some way. 
The following lemma formalizes this idea (see Appendix \ref{appendix:proof_rccs} for its proof), which can be further extended to the convex combinations of any finite number of distributions.

\begin{restatable}{lemma}{LinOfProbTransStr}
    \label{lem:linearity_of_probabilistic_transition_strengthened}
    Let $\mu_1,\mu_2\in \Distr(\PRCCS)$ be two distributions and $p\in[0,1]$ be a real number.
    \begin{enumerate}
        \item 
        If $\mu_1\myrightarrow{\pi_1} \nu_1$ and $\mu_2\myrightarrow{\pi_2}\nu_2$, then $p\mu_1+(1-p)\mu_2\myrightarrow{\pi} p\nu_1+(1-p)\nu_2$ for some $\pi$ such that $\vert\pi\vert\le\vert\pi_1\vert+\vert\pi_2\vert$.
        \item 
        If $p\mu_1+(1-p)\mu_2\myrightarrow{\pi} \nu$, then there exist two distributions $\nu_1,\nu_2$ such that $\mu_1\myrightarrow{\pi_1}\nu_1$, $\mu_2\myrightarrow{\pi_2}\nu_2$ and $\nu=p\nu_1+(1-p)\nu_2$, where $\vert\pi_1\vert,\vert\pi_2\vert\le\vert\pi\vert$.
    \end{enumerate}
\end{restatable}

Lemma~\ref{lem:linearity_of_probabilistic_transition_strengthened} provides a concise mathematical description of the convex combination and decomposition operations for probabilistic transition sequences, which exemplifies a key advantage of our proposed semantics.
Note that the lengths of new transition sequences are bounded, which is crucial when performing inductions.
If $\pi \in (\{\tau\}\times(0,1])^{*}$, we can immediately obtain the following commonly used corollary for internal transition sequences.

\begin{corollary}
    \label{cor:linearity_of_probabilistic_transition}
    Let $\mu_1,\mu_2\in \Distr(\PRCCS)$ be two distributions and $p\in[0,1]$ be a real number.
    \begin{enumerate}
        \item \label{item:linear_combination} 
        If $\mu_1\rightsquigarrow \nu_1$ and $\mu_2\rightsquigarrow\nu_2$, then $p\mu_1+(1-p)\mu_2\rightsquigarrow p\nu_1+(1-p)\nu_2$.
        \item \label{item:linear_decomposition}
        If $p\mu_1+(1-p)\mu_2\rightsquigarrow \nu$, then there exist two distributions $\nu_1,\nu_2$ such that $\mu_1\rightsquigarrow\nu_1$, $\mu_2\rightsquigarrow\nu_2$ and $\nu=p\nu_1+(1-p)\nu_2$.
    \end{enumerate}
\end{corollary}

We have previously stated that our distribution-based approach is semantically well-founded.
Here, we conclude this section with a few additional discussions in support of this claim.

\begin{itemize}
    \item 
    Our semantics is more suitable for characterizing the continuous nature of probabilities.
    For instance, to characterize branching bisimulation, Fu~\cite{fu_ModelIndependentApproach_2021} has employed a conditional probability strategy, where $Q_1$ and $Q_2$ in Example~\ref{ex:semantics} are equivalent in the sense that if $Q_i\myrightarrow{\tau} X$ then
    $\mathsf{Pr}(X\equiv A\mid X\not\equiv Q_i)=\frac{1}{2}$ for $A=a,b$ and $i=1,2$.
    Even so, this explanation cannot account for multi-step simulations, 
    However, under our semantics, such a multi-step simulation only requires, in a natural way, that equivalent distributions can almost reach equivalent distributions, as explained in Example~\ref{ex:probabilistic_transition}.
    \item 
    For probabilistic testing, the tree-based method often makes too fine a distinction between nodes, where different nodes are treated as distinct, even if they are labeled with the same process.
    For instance, in Example~\ref{ex:semantics}, $Q_3$ can reach leaf distribution $\{(a,1)\}$, $Q_4$ can reach leaf distribution $\left\{\left(a,\frac{1}{2}\right), \left(a,\frac{1}{2}\right)\right\}$, and
    these two distributions are distinct in the sense that only the latter can reach leaf distribution $\left\{\left(a, \frac{1}{2}\right), \left(\mathbf{0},\frac{1}{2}\right)\right\}$.
    The scheduler-based method~\cite{turrini_PolynomialTimeDecision_2015} also has this issue, where the resulting $a$'s from different branches of $Q_3$ are not treated as identical, as their traces differ.
    In contrast, our approach is entirely distribution-based, which
    eliminates such unnecessary differences since $\delta_a\xrightarrow{\left(a,\frac{1}{2}\right)}\left\{\left(a, \frac{1}{2}\right), \left(\mathbf{0},\frac{1}{2}\right)\right\}$.
\end{itemize}

At the same time, the discussion above clarifies in what sense our semantics improves upon tree-based approaches.
Since scheduler-based semantics involves substantially more intricate definitions, we defer a formal comparison to Section~\ref{subsec:comparison_weak}.

\section{A predicate-based framework for probabilistic testing}
\label{sec:testing_over_distributions}

In this section, we propose a predicate-based framework for probabilistic testing.
We begin by formalizing the testing outcome of a distribution.
Recall that the classical way (in Section~\ref{subsec:may_must}) is to introduce a special action $\omega$ to specify success. 
Here, we take \emph{predicates} on processes, i.e., the subsets of $\PRCCS$, as the extension. 
Given a distribution $\mu$ and a predicate $\varphi \subseteq \PRCCS$, we use $\mu(\varphi) = \sum_{P \in \varphi} \mu(P)$ to denote the probability that the predicate $\varphi$ is satisfied under the distribution $\mu$, and call it the \emph{satisfaction probability of $\varphi$ under $\mu$}. 
The \emph{testing outcomes of $\mu$ with respect to $\varphi$} is then defined by the set $\mathcal{O}^{\mu}_\varphi := \left\{\nu(\varphi) : \mu\rightsquigarrow\nu\right\}$, that is, the collection of satisfaction probabilities of $\varphi$ under the distributions reachable from $\mu$ by probabilistic internal transition sequences.
We shall further abbreviate $\mathcal{O}^{\delta_P}_\varphi$ to $\mathcal{O}^{P}_\varphi$.

\subsection{The structure of testing outcomes}
\label{subsec:structure_of_testing_outcome}

Classical testing equivalences are obtained by imposing varying degrees of constraints on the testing results. 
Here, we show that the testing outcomes defined above are determined by their boundaries.
Therefore, it suffices only to consider the boundary points when defining probabilistic testing equivalences.
Let $O_1, O_2$ be two sets of real numbers, and $k\in\mathbb{R}$ be a scalar. 
The Minkowski sum and scalar product for sets are defined as:
$O_1 + O_2 := \{ x + y : x \in O_1, y \in O_2\},
kO_1 := \{ kx : x \in O_1 \}$.
As we have had Lemma \ref{lem:linearity_of_probabilistic_transition_strengthened}, the intuition of the following lemma, which shows the linearity of testing outcomes, should be clear. 

\begin{restatable}{lemma}{linOfTestOut}
    \label{lem:linearity_of_testing_outcome}
    Let $\mu_1,\mu_2\in\Distr(\PRCCS)$ be two distributions, $\varphi\subseteq \PRCCS$ be a predicate, and $p\in[0,1]$ be a real number.
    The following statements are valid.
    \begin{enumerate}
        \item\label{item:linearity_of_testing_outcome_1} $\mathcal{O}^{p\mu_1+(1-p)\mu_2}_{\varphi}=p\mathcal{O}^{\mu_1}_{\varphi}+(1-p)\mathcal{O}^{\mu_2}_{\varphi}.$
        \item\label{item:linearity_of_testing_outcome_2} $\partial\,\mathcal{O}^{p\mu_1+(1-p)\mu_2}_{\varphi}=p\cdot\partial\,\mathcal{O}^{\mu_1}_{\varphi}+(1-p)\cdot\partial\,\mathcal{O}^{\mu_2}_{\varphi}$, for $\partial\in\{\sup,\inf\}$.
    \end{enumerate}
\end{restatable}

The proof is provided in Appendix~\ref{appendix:framework}.
Lemma~\ref{lem:linearity_of_testing_outcome} delineates the structure of testing outcomes.
If $x, y \in \mathcal{O}^\mu_\varphi$, then for any $p\in[0,1]$, we can apply Lemma~\ref{lem:linearity_of_testing_outcome} (\ref{item:linearity_of_testing_outcome_1}) and obtain that
\begin{equation}
    px + (1 - p)y \in p\mathcal{O}^\mu_\varphi + (1 - p)\mathcal{O}^\mu_\varphi = \mathcal{O}^{p\mu + (1 - p)\mu}_\varphi = \mathcal{O}^\mu_\varphi.
\end{equation}
Therefore, $\mathcal{O}^\mu_\varphi$ is a nonempty convex subset of $[0, 1]$, i.e., an interval or a singleton.
The supremum and infimum can almost fully determine $\mathcal{O}^\mu_\varphi$, except for the two boundary points.
This also explains why we need Lemma~\ref{lem:linearity_of_testing_outcome} (\ref{item:linearity_of_testing_outcome_2}).
In fact, when $\varphi$ is chosen appropriately, $\sup \mathcal{O}^\mu_\varphi$ corresponds to the may equivalence.
For fair testing, we focus on the quantity $\inf \left\{\sup \mathcal{O}^{\nu}_{\varphi}: \mu \rightsquigarrow \nu\right\}$, which represents the worst case success probability of distribution $\mu$ in any fair test $\mu \rightsquigarrow \nu$.
Let $\charMay_\varphi(\mu):=\sup \mathcal{O}^\mu_\varphi$ and $\charFS_{\varphi}(\mu):=\inf \left\{\sup \mathcal{O}^{\nu}_{\varphi}: \mu \rightsquigarrow \nu\right\}$, which are called the \emph{may} and \emph{fair characteristic} of $\mu$ with respect to $\varphi$.
Similarly, a shorthand notation is adopted if $\mu$ is a Dirac distribution $\delta_P$.
See the examples below for an intuitive understanding.
More explanations will be given in Sections~\ref{sec:model_independent}.

\begin{example}
\label{ex:linearity_testing_outcome}

Consider the two processes $P_1:=\tau.a, P_2:=\tau.a+\tau.(\fix X.\tau.X)$ in Example~\ref{ex:classical_testing}.
Define $\psi_a:=\{P \in \PRCCS:P\myrightarrow{a}\}$, which characterizes processes that can immediately perform action $a$.
Then we have $a\in \psi_a$ whereas $P_1 \notin \psi_a$.
Since $\delta_{P_1}\rightsquigarrow{\delta_a}$ and $\delta_{P_1}\rightsquigarrow\delta_{P_1}$ and $\delta_{P_1}(\psi_a)=0$ both hold, we have $\{0,1\}\subseteq\mathcal{O}^{P_1}_{\psi_a}$.
Then by the above analysis, we obtain that $\mathcal{O}^{P_1}_{\psi_a}=[0,1]$.
Similarly, we can deduce that $\mathcal{O}^{P_2}_{\psi_a}=[0,1]$.
Note that $\charMay_{\psi_a}(P_1)=\charMay_{\psi_a}(P_2)=1$ holds, which intuitively indicates that $P_1$ and $P_2$ ``may pass'' the test by observer $\overline{a}.\boldsymbol{\omega}$ in classical settings.

As $\delta_{P_2}\rightsquigarrow\delta_{\fix X.\tau.X}$ and $\mathcal{O}_{\psi_a}^{{\fix X.\tau.X}}=0$, we have $\charFS_{\psi_a}(P_2)=\inf\{\sup\mathcal{O}^{\nu}_{\psi_a}:\delta_{P_2}\rightsquigarrow\nu\}=0$.
However, for any $\nu$ such that $\delta_{P_1}\rightsquigarrow\nu$, we have $\nu=p\delta_{P_1}+(1-p)\delta_a$ for some $p\in[0,1]$.
Applying Lemma~\ref{lem:linearity_of_testing_outcome}, we have 
\begin{equation}
    \sup\mathcal{O}^{\nu}_{\psi_a}=p\sup\mathcal{O}^{P_1}_{\psi_a}+(1-p)\sup\mathcal{O}^{a}_{\psi_a}=1.
\end{equation}
Hence, $\charFS_{\psi_a}(P_1)=1\ne \charFS_{\psi_a}(P_2)$, which is consistent with the fact that $P_1 \not\FS P_2$.
\end{example}

\begin{example}
\label{ex:oplus_testing_outcome}

Consider the process $Q_2:=\fix X.\left(\frac{1}{3}\tau.a\oplus\frac{1}{3}\tau.b\oplus\frac{1}{3}\tau.X\right)$ involving the probabilistic choice operation in Example~\ref{ex:semantics}.
Let $\psi_a$ be the predicate defined above.
By Equation~\ref{eq:Q2}, we have $\frac{3^k-1}{2\cdot 3^k}\in\mathcal{O}^{Q_2}_{\psi_{a}}$ for each $k \in \mathbb N$.
Since $\frac{1}{2}$ can never be achieved by any descendant  of $\delta_{Q_2}$, one has that $\mathcal{O}^{Q_2}_{\psi_a}=[0,\frac{1}{2})$.
However, $\charMay_{\psi_a}(Q_2)=\sup \mathcal{O}^{Q_2}_{\psi_a}=\frac{1}{2}$, which intuitively means that there exists a distribution $\nu$, which can be obtained from $\delta_{Q_2}$ by an ``infinite way'' (see Figure~\ref{fig:infinite_case}), such that $\nu(\psi_a)=\frac{1}{2}$.
Applying the supremum and infimum can prevent us from such complex infinite cases, therefore making the formalization concise.

\begin{figure}[htb]
\centering
\begin{tikzpicture}
\tikzset{vertex/.style = {}}
\tikzset{edge/.style = {->,> = latex'}}
\node[vertex] (a) at  (0,1) {\footnotesize $Q_2$};
\node[vertex] (b) at  (1,-0.2) {\footnotesize $b$};
\node[vertex] (c) at  (2,1) {\footnotesize $Q_2$};
\node[vertex] (d) at  (1,2.2) {\footnotesize $a$};
\draw[edge] (a) to node[left, pos=0.7] {\footnotesize$\frac{1}{3}\tau$} (b);
\draw[edge] (a) to node[above] {\footnotesize$\frac{1}{3}\tau$} (c);
\draw[edge] (a) to node[left, pos=0.7] {\footnotesize$\frac{1}{3}\tau$} (d);
\node[vertex] (b) at  (3,-0.2) {\footnotesize $b$};
\node[vertex] (a) at  (4,1) {\footnotesize $Q_2$};
\node[vertex] (d) at  (3,2.2) {\footnotesize $a$};
\draw[edge] (c) to node[left, pos=0.7] {\footnotesize$\frac{1}{3}\tau$} (b);
\draw[edge] (c) to node[above] {\footnotesize$\frac{1}{3}\tau$} (a);
\draw[edge] (c) to node[left, pos=0.7] {\footnotesize$\frac{1}{3}\tau$} (d);
\node[vertex] (b) at  (5,-0.2) {\footnotesize $b$};
\node[vertex] (c) at  (6,1) {\footnotesize $\cdots$};
\node[vertex] (d) at  (5,2.2) {\footnotesize $a$};
\draw[edge] (a) to node[left, pos=0.7] {\footnotesize$\frac{1}{3}\tau$} (b);
\draw[edge] (a) to node[above] {\footnotesize$\frac{1}{3}\tau$} (c);
\draw[edge] (a) to node[left, pos=0.7] {\footnotesize$\frac{1}{3}\tau$} (d);
\node[vertex] (e) at (8,1) {\footnotesize $Q_2$};
\draw[edge] (c) to node[above] {\footnotesize$\frac{1}{3}\tau$} (e);
\node[vertex] (b) at  (9,-0.2) {\footnotesize $b$};
\node[vertex] (a) at  (10,1) {\footnotesize $\cdots$};
\node[vertex] (d) at  (9,2.2) {\footnotesize $a$};
\draw[edge] (e) to node[left, pos=0.7] {\footnotesize$\frac{1}{3}\tau$} (b);
\draw[edge] (e) to node[above] {\footnotesize$\frac{1}{3}\tau$} (a);
\draw[edge] (e) to node[left, pos=0.7] {\footnotesize$\frac{1}{3}\tau$} (d);
\end{tikzpicture}
\caption{The infinite transition sequence of $Q_2$ to obtain the satisfactory probability $\frac{1}{2}$.}
\label{fig:infinite_case}
\end{figure}
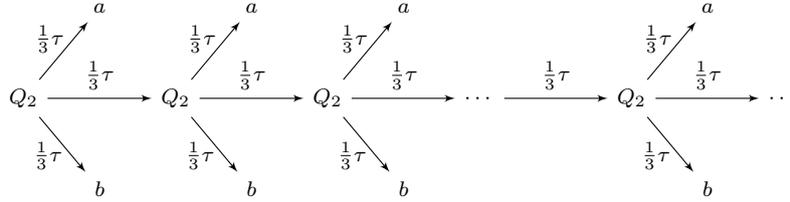

\end{example}

\paragraph{Divergence in probabilistic settings}

From Figure~\ref{fig:infinite_case}, we observe that although $Q_2$ has an infinite diverging branch, the probability that the process actually follows this branch is $0$.
In general, such ``innocent'' divergent behavior is often regarded as equivalent to non-divergence.
In fact, in our semantics, there is no intrinsic distinction between strict non-divergence and divergence with probability $0$.
For instance, under tree-based semantics, a process like $Q_1 = \frac{1}{2}\tau.a \oplus \frac{1}{2}\tau.b$, which cannot generate any diverging branch, can in our semantics fully simulate the innocent divergent behavior of $Q_2$ (see Example~\ref{ex:probabilistic_transition}).

The probabilistic testing equivalences focus on interaction aspects of probabilistic processes and are divergent-insensitive, as no external observer can distinguish between a nil process $\mathbf{0}$ and an always divergent process $\mu X.\tau.X$.
However, the predicate-based framework developed in this section can be extended to other equivalences that are divergent-sensitive.
For probabilistic processes, the relevant notion is not a binary ``divergent vs. non-divergent’’ classification, but the probability with which divergence occurs.
The probability of divergence can be quantified through testing outcomes with respect to suitable predicates. We can define two predicates:
\begin{equation}
    \begin{aligned}
         \psi_{\tau} &:= \{P \in \PRCCS : P \myrightarrow{\tau}\text{ only} \}.\\
        \psi_{\mathrm{div}} &:= \{P \in \PRCCS : P \in \text{some \emph{$\tau$-end component}} \},
    \end{aligned}
\end{equation}
where $\tau$-end components can be seen as probabilistic version of $\tau$-cycles and characterizes processes that divergent with probability $1$. We omit here the formal definition of  $\tau$-end component, and readers can refer \cite{HE2023105033} for its details.
Then, for a distribution $\mu$ over finite-state processes, we can use $\mathrm{div}_{\min}^\mu:=\inf \mathcal{O}_{\psi_\tau}^{\mu}$ (resp. $\mathrm{div}_{\max}^\mu:=\sup \mathcal{O}_{\psi_{\mathrm{div}}}^{\mu}$) to quantify the minimal (resp. maximal) probability of divergence under all probabilistic internal transition sequences. In this way, a distribution $\mu$ may diverge with any probability within the interval $[{\rm div}_{\min}^\mu, {\rm div}_{\max}^\mu]$.
Now consider three processes $P_1 = \mu X.(\frac{1}{2}\tau.a \oplus \frac{1}{2}\tau.X)$, $P_2 = \tau.a + \tau.(\mu X.\tau.X)$ and $P_3 = \frac{1}{2}\tau.a \oplus \frac{1}{2}\tau.\mu X.\tau.X$, we can verify that 
\begin{equation}
    [\mathrm{div}_{\min}^{P_1}, \mathrm{div}_{\max}^{P_1}] = \{0\},\quad [\mathrm{div}_{\min}^{P_2}, \mathrm{div}_{\max}^{P_2}] = [0, 1],\quad \text{and}\quad [\mathrm{div}_{\min}^{P_3}, \mathrm{div}_{\max}^{P_3}]  = \{\frac{1}{2}\}.
\end{equation}
Therefore, we see that ``innocent'' divergent processes and  ``pathological'' divergent processes (diverge with positive probability) can be formally distinguished by this divergence testing.

\subsection{Approximation by degenerate sequences}
\label{subsec:approximation}

Since both the external and internal characterization of probabilistic testing will be defined within this framework, we provide several general technical lemmas in this part.
Recall that a witness $\pi$ of an internal transition sequence is degenerate if $\pi\in \tau^{*}$, i.e., without uncertainty.
Lemma~\ref{lem:bounded_by_degenerate_witness} shows that any reachable satisfaction probability can be both lower and upper approximated by some probabilistic transition sequences with degenerate witnesses. 

\begin{restatable}{lemma}{bdByDegenWit}
    \label{lem:bounded_by_degenerate_witness}
    Let $\varphi\subseteq\PRCCS$ be a predicate and $\mu,\nu\in\Distr(\PRCCS)$ be two distributions.
    If $\mu\rightsquigarrow\nu$, then there exist two distributions $\nu_1,\nu_2\in\Distr(\PRCCS)$ such that 
    \begin{enumerate}
        \item \label{item:bounded_by_degenerate_witness_1}
        for $i=1,2$,  $\mu\rightsquigarrow\nu_i$ with a degenerate witness $\pi_i\in\tau^{*}$, and
        \item \label{item:bounded_by_degenerate_witness_2}
        $\nu_1(\varphi)\le\nu(\varphi)\le\nu_2(\varphi)$.
    \end{enumerate}
\end{restatable}

The proof of Lemma~\ref{lem:bounded_by_degenerate_witness} (see Appendix~\ref{appendix:framework}) does not rely on the fact that $\varphi$ is a predicate.
If we regard $\varphi$ as its double dual, i.e., a function $\varphi^{*}: \Distr(\PRCCS) \to [0,1]$, which maps any distribution $\mu$ to $\mu(\varphi)$,
then for any distributions $\mu_1$ and $\mu_2$ and $p\in[0,1]$, we have that
\begin{equation}
    \varphi^{*}(p\mu_1 + (1-p)\mu_2) = p\mu_1(\varphi) + (1-p)\mu_2(\varphi) = p\varphi^{*}(\mu_1) + (1-p)\varphi^{*}(\mu_2).
\end{equation}
This means that $\varphi^{*}$ is a linear functional, which is the only crucial property of $\varphi$ used in the proof. 
Therefore, if we replace $\varphi^{*}$ with any other linear functional, the lemma still holds.
Lemma~\ref{lem:linearity_of_testing_outcome} verifies that the functional $\mu \mapsto \sup\,\mathcal{O}_{\varphi}^{\mu}$, is also linear.
Combining Lemma~\ref{lem:linearity_of_testing_outcome} and the proof of Lemma~\ref{lem:bounded_by_degenerate_witness}, we get the following lemma. 

\begin{lemma}
    \label{lem:outcome_bounded_by_degenerate_witness}
    Let $\varphi\subseteq\PRCCS$ be a predicate and $\mu,\nu\in\Distr(\PRCCS)$ be two distributions.
    If $\mu\rightsquigarrow\nu$ , then there exist two distributions $\nu_1,\nu_2\in\Distr(\PRCCS)$ such that 
    \begin{enumerate}
        \item \label{item:outcome_bounded_by_degenerate_witness_1}
        for $i=1,2$,  $\mu\rightsquigarrow\nu_i$ with a degenerate witness $\pi_i\in\tau^{*}$, and
        \item\label{item:outcome_bounded_by_degenerate_witness_2}
        $\sup\,\mathcal{O}_{\varphi}^{\nu_1}\le\sup\,\mathcal{O}_{\varphi}^{\nu}\le\sup\,\mathcal{O}_{\varphi}^{\nu_2}$.
    \end{enumerate}
\end{lemma}

A major application of Lemma~\ref{lem:bounded_by_degenerate_witness} and Lemma~\ref{lem:outcome_bounded_by_degenerate_witness} is that when examining the boundaries of testing outcomes, it suffices to focus on transition sequences with degenerate witnesses. 

\begin{corollary}
    \label{cor:calculated_by_degenerate_witness}
    Given $\mu\in\Distr(\PRCCS)$ and $\varphi\subseteq\PRCCS$, the following equalities are valid. 
    \begin{enumerate} 
        \item \label{item:calculated_by_degenerate_witness_1}
        $\charMay_{\varphi}(\mu)=\sup\,\left\{\nu(\varphi) : \mu \myrightarrow{\tau^{*}} \nu\right\}$,
        \item \label{item:calculated_by_degenerate_witness_2}
        $\charFS_{\varphi}(\mu)= \inf \left\{\sup \mathcal{O}^{\nu}_{\varphi} : \mu \myrightarrow{\tau^{*}} \nu\right\}$.
    \end{enumerate}
\end{corollary}

Corollary~\ref{cor:calculated_by_degenerate_witness} provides a simplification of the definition of characteristics.
However, it is important to note that convex combinations of two degenerate sequences are not necessarily degenerate.
Therefore, the simplified form on the right cannot be viewed as the essence of the testing characteristics.
\section{Unifying approach to probabilistic testing equivalences}
\label{sec:model_independent}
In this section, we introduce the unifying internal characterization for probabilistic testing equivalences ($\PXDiamond$ and $\PXBox$) and explore their relationship.
Then we extend the classical may and fair equivalences to the RCCS model and show that they are exactly the external characterizations for $\PXDiamond$ and $\PXBox$, respectively. 
All these equivalences are proved to be congruences.
We then transferred our approach to the pCSP model and obtained similar conclusions, which further support the generality of our framework.
Detailed proofs are provided in Appendix~\ref{appendix:proof_model_independent}.

\subsection{Towards a unifying approach to probabilistic testing}
In \cite{fu_NamePassingCalculus_2015}, the authors formulated a unifying characterization of testing equivalences for the $\pi$-calculus.
Since this approach abstracts away from the specific details of external actions, it can be adapted to CCS with minimal modifications as follows.
We say that a process $P$ is \emph{observable}, written $P\Downarrow $, if $P \Longrightarrow \myrightarrow{\ell} P'$ for some $\ell\in \mathcal{L}$ and $P'$.
A process $P$ is \emph{strongly observable}, written $P\observable$, if $P' \Downarrow$ for all $P'$ such that $P \Longrightarrow P'$.
Let $\R$ be a relation on $\PCCS$.
We define the following three properties.
\begin{enumerate}
    \item $\R$ is \emph{extensional} if the following property holds:
    \begin{enumerate}
        \item If $P_1 \,\R\, P_2$ and $Q_1 \,\R\, Q_2$, then 
        $(P_1 \mid Q_1) \;\R\; (P_2 \mid Q_2)$.
        \item If $P \,\R\, Q$, then $(a) P \;\R\; (a)Q$ for all $a \in \Chan$.
    \end{enumerate}
    \item
    $\R$ is \emph{equipollent} if $P\Downarrow \iff Q\Downarrow$ whenever $P \,\R\, Q$.
    \item
    $\R$ is \emph{strongly equipollent} if $P\observable \iff Q\observable$ whenever $P \,\R\, Q$.
\end{enumerate}

The \emph{diamond equality} $\CDiamond$ is the largest equipollent, extensional equivalence on $\PCCS$.
The \emph{box equality} $\CBox$ is the largest strongly equipollent, extensional equivalence on $\PCCS$.
The main results regarding $\CBox$ and $\CDiamond$ presented in \cite{fu_NamePassingCalculus_2015} are  valid in CCS as well. 
We summarize them as the following propositions.

\begin{proposition}
\label{prop:classical_model_independent_characterization}
    The following statements are valid.
    \begin{enumerate}
        \item $(\FS)=(\CBox)\;\subsetneq\;(\CDiamond)=(\May)$.
        \item Both the box equality $\CBox$ and the diamond equality $\CDiamond$ are congruences on $\PCCS$.
    \end{enumerate}
\end{proposition}

Proposition~\ref{prop:classical_model_independent_characterization} forms the top-level hierarchy structure in Figure~\ref{fig:summary}.
We will provide proofs from a probabilistic view in Section~\ref{subsec:comparison_with_classical_testing}. 
However, since these definitions and properties are based on processes, while all our characterizations rely on distributions, we introduce the following projection and lifting operations.
\begin{definition}
    \label{def:lifting}
    Let $\R$ be a relation on $\Distr(\PRCCS)$ and $\E$ be an equivalence over $\PRCCS$.
    \begin{enumerate}
        \item The projection of $\R$ onto a set $\mathcal{X}\subseteq\PRCCS$ is defined by 
        $
            \Restr{\R}{\mathcal{X}}:=\{(P_1,P_2)\in\mathcal{X}^2:\delta_{P_1}\,\R\,\delta_{P_2}\}.
        $
        \item The lifting of $\E$ is defined by
        $
            \E^\dag:=\{(\mu_1,\mu_2):\forall\;\mathcal{C}\in \PRCCS/\E, \mu_1(\mathcal{C})=\mu_2(\mathcal{C})\}.
        $
    \end{enumerate}
\end{definition}

Note that the lifting of any equivalence must be closed under the convex combination.
We take $\PCCS$ as a subset of $\PRCCS$.
The projections of $\R$ onto $\PCCS$ and $\PRCCS$ are abbreviated as $\Restr{\R}{CCS}$ and $\Restr{\R}{RCCS}$, respectively. 
Furthermore, since the projection of $\R$ is always clear, we usually omit such subscripts unless we want to emphasize the operation itself.
For convenience, we formalize the probabilistic strong bisimulation within our semantics as a tool, which means two processes can simulate every transition of each other.

\begin{definition}[Probabilistic strong bisimulation]
    \label{def:p_strong}
    A relation $\R$ on $\PRCCS$ is a \emph{probabilistic strong bisimulation} if for any $\alpha\in\mathcal{A}$, the following statements are valid.
    \begin{enumerate}
        \item If $P\,\R\,Q$ and $P\myrightarrow{\alpha}\sum_{i\in I}p_i\delta_{P_i}$, then $Q\myrightarrow{\alpha}\sum_{i\in I}p_i\delta_{Q_i}$ and $P_i\R Q_i$ for all $i\in I$.
        \item If $P\,\R\,Q$ and $Q\myrightarrow{\alpha}\sum_{i\in I}p_i\delta_{Q_i}$, then $P\myrightarrow{\alpha}\sum_{i\in I}p_i\delta_{P_i}$ and $P_i\R Q_i$ for all $i\in I$.
    \end{enumerate}
\end{definition}

The \emph{probabilistic strong bisimilarity} over distributions is defined to be $\Strong^\dag$, where equivalence $\Strong$ is the largest probabilistic strong bisimulation over $\PRCCS$.
It should be clear that both $\Strong$ and $\Strong^\dag$ are well-defined equivalences enjoying the following structural properties.
\begin{restatable}{proposition}{propStrongBisim}
    \label{prop:property_strong_bisimulation}
    Suppose $P, Q, A\in\PRCCS$, and $T, R$ are RCCS terms containing free variable $X$.
    Let $\mu_1,\mu_2, \nu_1\in\Distr(\PRCCS)$.
    \begin{enumerate}
        \item $P \mid Q\Strong Q\mid P$, $P\mid\mathbf{0}\Strong P$, $(P\mid Q)\mid A\Strong P\mid(Q\mid A)$, and $\fix X.T\Strong T\{\fix X.T/X\}$.
        \item If $L\subseteq\Chan$ and $L\cap \Chan(Q)=\emptyset$, then $(L)Q\Strong Q$ and $(L)P\mid Q\Strong (L)(P\mid Q)$.
        \item\label{item:property_strong_bisimulation_3} If $P\Strong Q$, then $R\{P/X\}\Strong R\{Q/X\}$.
        \item
        If $\mu_1\Strong^\dag\mu_2$ and $\mu_1\rightsquigarrow\nu_1$, then $\mu_2\rightsquigarrow\nu_2\Strong^\dag\nu_1$ for some $\nu_2$.
    \end{enumerate}
\end{restatable}

\begin{remark}
The probabilistic strong bisimilarity puts rather strong constraints on distributions such that no external tests based on observable properties can distinguish such an equivalent pair.
Proposition~\ref{prop:property_strong_bisimulation} underpins this claim.
For instance, $\tau.a\Strong\frac{1}{2}\tau.a\oplus\frac{1}{2}\tau.a$ .
Therefore, in the subsequent discussion, we do not need to distinguish them.
\end{remark}

\subsection{Probabilistic box equivalence and diamond equivalence}
\label{subsec:probabilistic_diamond_box}

\paragraph{Observability}
We now generalize the (strong) equipollence and the extensionality to the RCCS model.
We first define 
\begin{equation}
    \psi_\L:= \left\{ P \in \PRCCS : P \myrightarrow{\ell} \text{ for some } \ell \in \L\right\}
\end{equation} to indicate that a process can perform an external action
and abbreviate $\mathcal{O}^{\mu}_{\psi_\L}$, $\charMay_{\psi_\L}$, and $\charFS_{\psi_\L}$ as $\PEquip{\mu}$, $\charMay_{\L}$, and $\charFS_{\L}$, respectively.
Note that the essence of the unifying approach lies in its disregard for the specific behavior of external actions, which is effectively captured by $\psi_\L$.
However, with the introduction of probabilities, the observability of processes should no longer be  a dichotomy of observable/unobservable. 
For example, consider two processes 
\begin{align}
    &P_5:=\tau.(0.99 \tau.a\oplus 0.01\tau.\mathbf{0})+\tau.(0.01 \tau.a\oplus 0.99\tau.\mathbf{0}),\text{\quad and}\\
    &P_6:=\tau.(0.49 \tau.a\oplus 0.51\tau.\mathbf{0})+\tau.(0.51 \tau.a\oplus 0.49\tau.\mathbf{0}),
\end{align}
where we might agree that $P_5$ has a higher likelihood of being observable than $P_6$.
Fortunately, the testing framework introduced in Section~\ref{sec:testing_over_distributions} allows us to refine the definition of observability as follows.

\begin{definition}[Probabilistic equipollence]
\label{def:probabilistic_equipollence}
Let $\R$ be a relation over $\Distr(\PRCCS)$.
    Then 
    \begin{enumerate}
        \item $\R$ is \emph{probabilistically equipollent} if $\charMay_{\L}(\mu_1)=\charMay_{\L}(\mu_2)$ holds whenever $\mu_1\,\R\,\mu_2$.
        \item $\R$ is \emph{probabilistically strongly equipollent} if $\charFS_{\L}(\mu_1)=\charFS_{\L}(\mu_2)$ holds whenever $\mu_1\,\R\,\mu_2$.
    \end{enumerate}
\end{definition}
The following examples should help to understand the definition.

\begin{example}
    We can interpret $\charMay_{\L}(\mu)$ as the degree to which $\mu$ is observable.
    Accordingly, our claim that $P_5$ is more likely to be observable than $P_6$ is supported by the fact that $\charMay_{\L}(P_5) = 0.99 > 0.51 = \charMay_{\L}(P_6)$.
    Similarly, since $\charFS_{\L}(P_5)=0.01<0.49=\charFS_{\L}(P_6)$, we may say that $P_5$ is ``more strongly observable'' than $P_6$.
    Then, a probabilistically (strongly) equipollent relation requires that any pair of relevant distributions have the same degree of being (strongly) observable.
\end{example}

\begin{example}
    \label{ex:equipollence}
    Let $\R:=\{(\delta_{P_1}, \delta_{P_2})\}$, where $P_1, P_2$ are defined in Example~\ref{ex:classical_testing}.
    Since $\Chan(P_1)=\Chan(P_2)=\{a\}$, the distinguishing power of $\psi_\mathcal{L}$ is no greater than that of $\psi_a$ in this simple case.
    According to Example~\ref{ex:linearity_testing_outcome}, $\R$ is probabilistically equipollent but not probabilistically strongly equipollent.
\end{example}

From the examples above, we observe that for a distribution $\mu$, the values of $\charMay_{\L}(\mu)$ and $\charFS_{\L}(\mu)$ are largely unrelated.
However, there remains a strong underlying connection between these two notions of characteristics, which can be quantified as follows.

\begin{restatable}{lemma}{seqTrans}
\label{lem:sequence_transform}
    Let $\mu\in\Distr(\PRCCS)$ be a distribution, $L$ be a finite set of channels such that $\Chan(\mu)\subseteq L$ and $b\not\in L$ be a fresh channel.
    Define $Q=\sum_{a\in L}a+\sum_{a\in L}\overline{a}+b$ and $\mu^L=(L)(\mu\mid\delta_{Q})$. Then
    \begin{equation}
        \charMay_{\L}(\mu)+\charFS_{\L}(\mu^L)=1.
    \end{equation}
\end{restatable}

Intuitively, the only external action $\mu^L$ can perform is $b$ (through $Q\myrightarrow{b}\delta_{\mathbf 0}$).
But if $\mu$ itself can perform a external action $a$, then it can interact with $\delta_Q$ in such a way that their parallel composition loses the ability to perform $b$. 
Roughly speaking, the silent transition sequences of $\mu$ and $\mu^L$ correspond one-to-one, with the property that the satisfaction probabilities of their respective descendants with respect to $\L$ are complementary.
\begin{example}
    Let $Q:=\overline a+b$. Then it is easy to compute and verify that
    \begin{equation}
        \charFS_{\L}((a)(P_5\mid Q))=0.01=1-\charMay_{\L}(P_5),\quad\text{and}\quad
        \charFS_{\L}((a)(P_6\mid Q))=0.49=1-\charMay_{\L}(P_6).
    \end{equation}
\end{example}

\paragraph{Extensionality} 
Given $\mathcal{X}\subseteq \PRCCS$, we define $\Delta(\mathcal{X}):=\left\{\delta_{P}:P\in \mathcal{X}\right\}$.
A set $\mathcal{D}$ is a \emph{testing context} if $\Delta(\PCCS)\subseteq\mathcal{D}\subseteq\Distr(\PRCCS)$ and $\mathcal{D}$ is closed under the parallel composition, localization, and renaming operations.
Notably, we allow relations to exhibit extensionality properties only in a testing context, facilitating the establishment of connections to the classical case in Section~\ref{subsec:comparison_with_classical_testing}.

\begin{definition}[$\mathcal{D}$-Extensional]
\label{def:probabilistic_extensionality}
    Suppose $\mathcal{D}$ is a testing context.
    A relation $\R$ on $\Distr(\PRCCS)$ is \emph{$\mathcal{D}$-extensional} if for $\mu_1,\mu_2\in\Distr(\PRCCS)$ and $\nu\in\mathcal{D}$ the following property holds:
    \begin{enumerate}
        \item ($\mathcal{D}$-compositional) If $\mu_1 \,\R\, \mu_2$, then $\left(\mu_1 \mid \nu\right) \;\R\; (\mu_2 \mid \nu)$.
	\item (Localizable) If $\mu_1 \,\R\, \mu_2$, then $(L)\mu_1\;\R\; (L)\mu_2$ for all finite set $L \subseteq \Chan$.
    \end{enumerate}
\end{definition}

Now we are prepared to present the unifying internal characterization for probabilistic testing equivalences.
Definition~\ref{def:probabilistic_extensionality} and Definition~\ref{def:probabilistic_equipollence} induce the following two equivalences, both of which are parametric with respect to the testing context $\mathcal{D}$.
The main result of this section holds for all testing contexts.
Therefore, in most cases, we omit the explicit qualification ``for any $\mathcal{D}$''.

\begin{definition}[$\mathcal{D}$-diamond equivalence]
\label{def:probabilistic_diamond_equivalence}
	The \emph{$\mathcal{D}$-diamond equivalence} $\PXDiamond$ is the largest probabilistically equipollent, $\mathcal{D}$-extensional relation on $\Distr(\PRCCS)$.
\end{definition}

\begin{definition}[$\mathcal{D}$-box equivalence]
\label{def:probabilistic_box_equivalence}
	The \emph{$\mathcal{D}$-box equivalence} $\PXBox$ is the largest probabilistically strongly equipollent, $\mathcal{D}$-extensional relation on $\Distr(\PRCCS)$.
\end{definition}

The $\mathcal{D}$-extensionality can be obtained through an inductive procedure.
The \emph{$\mathcal{D}$-extensional closure} of a relation $\R$ is defined by $\R^{\mathcal{D}}:=\bigcup_{k}\R_k$, where $\R_0:=\R$ and $\R_{k+1}:=\{((L)\mu_1 \mid \nu, (L)\mu_1\mid \nu) : \mu_1\;\R_k\;\mu_2, \nu\in\mathcal{D}, L\subseteq \Chan\text{~is finite}\}$.
Given an $n$-ary RCCS term $S$, for convenience, we shall assume w.l.o.g. that the free variables in $S$ are always $X\equiv X_1, X_2\dots, X_n$ and
abbreviate $S\{\widetilde{P}/\widetilde{X}\}$ as $S[\widetilde{P}]$ for any $n$-tuple $\widetilde{P}$ of processes.
The convention for an $n$-ary distribution is similar.
Based on these concepts, Lemma~\ref{lem:show_membership} is the primary tool we use to prove two distributions are equivalent w.r.t. $\PXDiamond$ or $\PXBox$ in this paper.

\begin{restatable}{lemma}{showMem}
    \label{lem:show_membership}
    Given a testing context $\mathcal{D}$, the following statements are valid.
    \begin{enumerate}
        \item \label{item:show_membership_1} 
        Let $\mu_1, \mu_2 \in \Distr(\PRCCS)$ be two distributions.
        If $\mu_1\Strong^\dag \mu_2$, then $\mu_1\PXDiamond\mu_2$.
        \item \label{item:show_membership_2}
        Let $\R$ be a relation on $\Distr(\PRCCS)$,
        If $\R^{\mathcal{D}}$ is probabilistically equipollent, then $\R\,\subseteq\,\PXDiamond$.
        \item \label{item:show_membership_3}
        Let $P, Q\in \PRCCS$.
        If $\R:=\{(\vartheta[P], \vartheta[Q]):\vartheta\text{~is a $1$-ary distribution}\}$ is probabilistically equipollent, or equivalently, $\charMay_\L(\vartheta[P])=\charMay_\L(\vartheta[Q])$ for any $1$-ary distribution $\vartheta$, then $P\PXDiamond Q$.
    \end{enumerate}
    Moreover, the analog statements also hold for $\PXBox$.
\end{restatable}

The following example provides some insights into the relationship between $\PXBox$ and $\PXDiamond$.

\begin{example}
    \label{ex:diamond_box_conterexp}
    Let $P_1:=\tau.a$, $P_2:=\tau.a+\tau.\mu X.\tau.X, P_3:=\fix X.(\tau.a+\tau.X)$.
    \begin{itemize}
        \item Let $\R:=\{(\vartheta[P_1], \vartheta[P_2]):\vartheta\text{~is $1$-ary~}\}$.
        We can prove that $\R$ is probabilistically equipollent through an induction on transitions.
        Applying Lemma~\ref{lem:show_membership}, we have $P_1\PXDiamond P_2$.
    \item Example~\ref{ex:linearity_testing_outcome} shows that 
    $\charFS_{\psi_a}(P_1)\ne \charFS_{\psi_a}(P_2)$, which implies that $(\delta_{P_1},\delta_{P_2})\not\in\;\PXBox$.
    Therefore, we deduce that $\PXDiamond \not\subseteq \PXBox$ 
    \item Let $\R':=\{(\vartheta[P_1], \vartheta[P_3]):\vartheta\text{~is $1$-ary~}\}$.
    Following a similar argument, we obtain that $\R'$ satisfies both types of equipollence.
    Therefore, both $P_1\PXDiamond P_3$ and $P_1\PXBox P_3$ hold for any testing context $\mathcal{D}$.
    \end{itemize}
\end{example}

Example~\ref{ex:diamond_box_conterexp} supports the central theorem of this subsection: the $\mathcal{D}$-box equivalence $\PXBox$ is strictly included in the $\mathcal{D}$-diamond equivalence $\PXDiamond$.

\begin{theorem}
\label{thm:probabilistic_box_diamond}
    $\PXBox\;\subsetneq\;\PXDiamond$.
\end{theorem}

\begin{proof}
    Let $\mu_1, \mu_2\in\Distr(\PRCCS)$ be two distributions such that $\mu_1\PXBox\mu_2$.
    Take $L=\Chan(\mu_1)\cup\Chan(\mu_2)$ and $b\in\Chan\setminus L$. 
    Since $\PXBox$ is $\mathcal{D}$-extensional, one has that $\mu_1^L\PXBox\mu_2^L$ and thus $\charFS_{\L}(\mu_1^L)=\charFS_{\L}(\mu_2^L)$, where $\mu_i^L$ is defined as in Lemma~\ref{lem:sequence_transform}.
    Together with Lemma~\ref{lem:sequence_transform}, we have
    \begin{equation}
        \charMay_{\L}(\mu_1)=1-\charFS_{\L}(\mu_1^L)=1-\charFS_{\L}(\mu_2^L)=\charMay_{\L}(\mu_2).
    \end{equation} 
    Therefore, $\PXBox$ is probabilistically equipollent, thus contained in the largest relation $\PXDiamond$.
    By Example~\ref{ex:diamond_box_conterexp}, the strictness is witnessed by $\tau.a$ and $\tau.a+\tau.\fix X.\tau.X$.
    \qed
\end{proof}

\begin{remark}
The establishment of $\mu_1^L\PXBox\mu_2^L$ relies on the fact that $\Delta(\PCCS)\subseteq\mathcal{D}$. 
This requirement is natural as two probabilistic testing equivalent processes should not be distinguishable by a classical observer.
\end{remark}

\subsection{External characterization for probabilistic testing equivalences}
\label{subsec:external_characterization}

We proceed to develop the external characterizations for probabilistic testing.
A $\mathrm{RCCS}^\omega$ process is obtained from a $\mathrm{RCCS}$ process by renaming some channels by $\omega$.
The set of all $\mathrm{RCCS}^\omega$ processes will be denoted by $\PRCCS^\omega$. The semantics of $\mathrm{RCCS}^\omega$ is defined under the action set $\Act\cup\{\omega\}$ as done in Section~\ref{subsec:semantics}.
The testing machinery of De Nicola and Hennessy \cite{denicola_TestingEquivalencesProcesses_1984} can be generalized to the RCCS model as follows.
\begin{enumerate}
	\item A special action $\omega \notin \Act$ is used to mark success.
	\item 
    We call any $O\in \PRCCS^\omega$ an observer process and any $o \in \Distr(\PRCCS^\omega)$ an observer distribution.
    A testing context $\mathcal{D}\subseteq\Distr(\PRCCS)$ can induce a set of observer distributions 
    $
        \mathcal{D}^\omega:=\{o[a\mapsto \omega] : o \in\mathcal{D}, a\in\Chan(o)\}.
    $
	\item We use $\psi_{\omega}:= \left\{ P\in\PRCCS^\omega : P \myrightarrow{\omega} \right\}$ to indicate that a process is in a successful state.
	\item The testing outcomes of process distribution $\mu$ by observer distribution $o$ (with respect to $\psi_{\omega}$) is then defined by the set
	$\mathcal{O}_{\psi_{\omega}}^{\mu\mid o}$. 
    We shall abbreviate $\mathcal{O}_{\psi_{\omega}}^{\mu\mid o}$, $\charMay_{\psi_\omega}$, and $\charFS_{\psi_\omega}$ as $\PTest{\mu\mid o}$, $\charMay_\omega$ and $\charFS_\omega$, respectively.
\end{enumerate}

\begin{example}
    The observer set induced by $\Delta(\PCCS)$ corresponds to all classical observers defined in Section~\ref{subsec:may_must}.
    The observer set induced by $\Distr(\PRCCS)$ is exactly all $\mathrm{RCCS}^\omega$ distributions with full probabilistic capabilities.
\end{example}

Next, we extend the predicates over binary relations for classical processes in Section \ref{subsec:may_must} to the probabilistic settings.

\begin{definition}
\label{def:prob_may_fs_predicate}
Given a relation $\R$ on $\Distr(\PRCCS)$ and a testing context $\mathcal{D}$, we say that
\begin{itemize}
    \item 
        $\R$ satisfies \emph{$\mathcal{D}$-may predicate} if $\mu_1\,\R\,\mu_2$ implies that, for every observer distribution $o\in \mathcal{D}^\omega$, it holds that 
    $\charMay_{\omega}(\mu_1\mid o)=\charMay_{\omega}(\mu_2\mid o)$, and that
    \item 
        $\R$ satisfies \emph{$\mathcal{D}$-fair predicate} if $\mu_1\,\R\,\mu_2$ implies that, for every observer distribution $o\in\mathcal{D}^\omega$, it holds that 
    $\charFS_{\omega}(\mu_1\mid o)=\charFS_{\omega}(\mu_2\mid o).$
\end{itemize}
\end{definition}

The well-definedness of the following two equivalence relations is evident.

\begin{definition}[$\mathcal{D}$-may equivalence]
    \label{def:pmay}
	The \emph{$\mathcal{D}$-may equivalence} $\PXMay$ is the largest equivalence on $\Distr(\PRCCS)$ that satisfies the $\mathcal{D}$-may predicate.
\end{definition}
\begin{definition}[$\mathcal{D}$-fair equivalence]
    \label{def:pfs}
	The \emph{$\mathcal{D}$-fair equivalence} $\PXFS$ is the largest equivalence on $\Distr(\PRCCS)$ that satisfies the $\mathcal{D}$-fair predicate.
\end{definition}

Table~\ref{tab:comparision} compares these definitions with those introduced in Section~\ref{subsec:probabilistic_diamond_box}. 
We will show that the equivalences sharing the same characteristic, i.e., those aligned in the same column of Table~\ref{tab:comparision}, are identical.

For the soundness of external characterizations, it should be clear that all definitions and lemmas in Section~\ref{sec:testing_over_distributions} are independent of the action set $\mathcal{A}$.
The strong bisimilarity $\Strong^\dag$ can also be extended to $\mathrm{RCCS}^\omega$.
Applying these techniques, we can obtain the $\mathcal{D}$-extensionality of $\PXMay$ and $\PXFS$.

\begin{table}[t]
    \centering
    \begin{tabular}{ccccccc}
        \toprule
        Characterizations & Predicate $\varphi$ & Characteristic $\charMay_{\varphi}$ & Characteristic $\charFS_{\varphi}$ \\
        \midrule
        internal style & $\psi_\L$ & $\PXDiamond$ & $\PXBox$  \\
        external style & $\psi_\omega$ & $\PXMay$ &  $\PXFS$ \\
        \bottomrule
    \end{tabular}
    \caption{Comparison between the unifying and external characterizations.}
    \label{tab:comparision}
    \vspace{-5mm}
\end{table}

\begin{restatable}{lemma}{probMayFsExt}
\label{lem:probabilistic_may_fs_extensional}
    Both $\PXMay$ and $\PXFS$ are $\mathcal{D}$-extensional.
\end{restatable}

Lemma~\ref{lem:external_equipollent} quantifies the relationship between the testing outcomes with respect to $\psi_\L$ and $\psi_\omega$ by carefully designing an observer to absorb all external actions, from which we can derive that the probabilistic may (fair, resp.) equivalence is (strong, resp.) equipollent.

\begin{restatable}{lemma}{extEqip}
    \label{lem:external_equipollent}
    Let $\mu\in\Distr(\PRCCS)$ be a distribution and $L\subseteq \Chan$ be a finite set of channels containing $\Chan(\mu)$. 
    Define observer $O_L:=\sum_{a\in L}a.\boldsymbol\omega+\sum_{a\in L}\overline{a}.\boldsymbol\omega$ and $o_L:=\delta_{O_L}$. 
    Then
    \begin{enumerate}
        \item 
        $\charMay_{\omega}(\mu\mid o_L)=\charMay_{\L}(\mu)$, and
        \item
        $\charFS_{\omega}(\mu\mid o_L)=\charFS_{\L}(\mu)$.
    \end{enumerate}
\end{restatable}

An example is also provided to further elucidate the above lemma.
\begin{example}
    \label{ex:reduction}
    Let $Q_2:=\fix X.\left(\frac{1}{3}\tau.a\oplus\frac{1}{3}\tau.b\oplus\frac{1}{3}\tau.X\right)$ and $O_L:=\overline a.\boldsymbol{\omega}+\overline b.\boldsymbol{\omega}$.
    Note that the descendants of $\delta_{Q_2}$ can interact with $o_L:=\delta_{O_L}$ if and only if $a$ or $b$ are exposed.
    We claim that every internal transition sequence of $Q_2$ can induce a corresponding one of $\delta_{Q_2}\mid o_L$.
    For instance, as shown in Figure~\ref{fig:infinite_case}, we have 
    \begin{equation}
        \delta_{Q_2}\rightsquigarrow\frac{1}{3^k}\delta_{Q_2}+\frac{3^k-1}{2\cdot 3^k}\delta_{a}+\frac{3^k-1}{2\cdot 3^k}\delta_b.
    \end{equation}
    Then by Lemma~\ref{lem:linearity_of_probabilistic_transition_strengthened}
    \begin{equation}
        \delta_{Q_2}\mid o_L\rightsquigarrow\frac{1}{3^k}\delta_{Q_2\mid O_L}+\frac{3^k-1}{2\cdot 3^k}\delta_{a\mid O_L}+\frac{3^k-1}{2\cdot 3^k}\delta_{b\mid O_L}\rightsquigarrow\frac{1}{3^k}\delta_{Q_2\mid O_L}+\frac{3^k-1}{3^k}\delta_{\mathbf 0\mid \boldsymbol{\omega}}.
    \end{equation}
    The satisfactory probabilities of these two distributions concerning $\psi_\L$ and $\psi_\omega$ are identical (both equal to $1-\frac{1}{3^k}$).
    Based on this observation, we can reduce the calculation of $\charMay_{\L}(Q_2)$ to that of $\charMay_{\omega}(Q_2\mid O_L)$ but not vice versa since $O_L$ is fixed in this equality.
\end{example}

Then the soundness follows.

\begin{restatable}[Soundness]{corollary}{corProbMayFsEqip}
    \label{cor:probabilistic_may_fs_equipollent}
    $\PXMay\,\subseteq\,\PXDiamond$ and $\PXFS\,\subseteq\,\PXBox$.
\end{restatable}

As for completeness, since the $\mathcal{D}$-may/fair predicates rely on the external action $\omega$, which is not allowed in $\PRCCS$, we need to eliminate this action.
For any process $P\in\PRCCS$, we choose some $a\notin\Chan(P)$.
Recall that $P[\omega\mapsto a]$ and $\mu[\omega\mapsto a]$ stand for renaming $P$ and $\mu$ by function $[\omega\mapsto a]$, respectively.
The basic fact we can observe is that $P\myrightarrow{\omega}$ if and only if $P[\omega\mapsto a]\myrightarrow{a}$.
It is evident that $\mu[\omega\mapsto a]$ is a distribution in $\Distr(\PRCCS)$.
The following lemma reveals the relationship between testing outcomes of $\mu$ and $\mu[\omega\mapsto a]$.

\begin{restatable}{lemma}{probDiaBoxPred}
    \label{lem:probabilistic_diamond_box_predicate}
    Let $\mu\in\Distr(\PRCCS^\omega)$ be a distribution, $L$ be a finite set of channels containing $\Chan(\mu)$ and $a\in \Chan\setminus L$ be a fresh channel. 
    Then we have that 
    \begin{enumerate}
        \item \label{item:probabilistic_diamond_box_predicate_1}
        $\charMay_{\omega}(\mu)=\charMay_{\L}((L)(\mu[\omega\mapsto a]))$, and
        \item \label{item:probabilistic_diamond_box_predicate_2}
        $\charFS_{\omega}(\mu)=\charFS_{\L}((L)(\mu[\omega\mapsto a]))$.
    \end{enumerate}
\end{restatable}

Lemma~\ref{lem:probabilistic_diamond_box_predicate} provides the reduction in the converse direction,  with an intuition similar to that in Example~\ref{ex:reduction}.
Now we are ready to prove the following correspondence theorem.

\begin{theorem}[External characterizations]
\label{thm:model_indenpent_characterizations}
    Let $\mu_1,\mu_2\in \Distr(\PRCCS)$ be two distributions. 
    \begin{enumerate}
        \item \label{item:model_indenpent_characterizaitons_1}
        $\mu_1\PXDiamond\mu_2$ if and only if  $\mu_1\PXMay\mu_2$.
        \item \label{item:model_indenpent_characterizaitons_2}
        $\mu_1\PXBox\mu_2$ if and only if $\mu_1\PXFS\mu_2$.
        \item $\PXFS\;\subsetneq\;\PXMay$.
    \end{enumerate}
\end{theorem}

\begin{proof}
Assume that $\mu_1,\mu_2\in\Distr(\PRCCS)$ and $\mu_1\PXDiamond\mu_2$.
For any observer distribution $o\in\mathcal{D}^\omega$, we take $L:=\Chan(\mu_1)\cup\Chan(\mu_2)\cup\Chan(o)$ and select a fresh channel $a\in \Chan\setminus L$.
Note that $o[\omega\mapsto a]\in \mathcal{D}$. 
Since the equivalence $\PXDiamond$ is $\mathcal{D}$-extensional, we have $(L)(\mu_1\mid o[\omega\mapsto a])\PXDiamond(L)(\mu_2\mid o[\omega\mapsto a])$.
By the probabilistic equipollence of $\PXDiamond$ and Lemma~\ref{lem:probabilistic_diamond_box_predicate},
\begin{equation}
\charMay_\omega({\mu_1\mid o})
=\charMay_\L((L)(\mu_1\mid o[\omega\mapsto a]))
=\charMay_\L((L)(\mu_2\mid o[\omega\mapsto a]))
=\charMay_\omega({\mu_2\mid o}),
\end{equation}
which implies that $\PXDiamond$ satisfies the probabilistic $\mathcal{D}$-may predicate and $\PXDiamond\;\subseteq\;\PXMay$.
Following a similar argument, we can prove that $\PXBox\;\subseteq\;\PXFS$.
By Corollary~\ref{cor:probabilistic_may_fs_equipollent}, we can conclude that $(\PXDiamond)=(\PXMay)$ and $(\PXBox)=(\PXFS)$.
Finally, by Theorem~\ref{thm:probabilistic_box_diamond} we have $\PXFS\,\subsetneq\,\PXMay$.
\qed
\end{proof}

\paragraph{Two typical testing contexts}
Note that Theorem~\ref{thm:model_indenpent_characterizations} holds for any testing context $\mathcal{D}$.
However, we are mainly concerned with two extreme cases: one where observers have full probabilistic capability (i.e., $\Distr(\PRCCS)$) and the other where observers have no probabilistic choice at all (i.e., $\Delta(\PCCS)$). 
For clarity, we summarize the intuitions and abbreviated notations in Table~\ref{tab:symbols}.
Then Corollary~\ref{cor:probabilistic_may_fs} follows immediately from Theorem~\ref{thm:model_indenpent_characterizations} and Theorem~\ref{thm:probabilistic_box_diamond}.

\begin{table}[hbt]
    \centering
    \begin{tabular}{ccccccc}
        \toprule
        Testing context & Intuition & $\mathcal{D}$-extensionality & $\PXBox$ & $\PXDiamond$ & $\PXFS$ & $\PXMay$\\
        \midrule
        $\Distr(\PRCCS)$ & probabilistic observers & $p$-extensionality & $\PBox$ & $\PDiamond$ & $\PFS$ & $\PMay$ \\
        $\Delta(\PCCS)$ & classical observers & $\Delta$-extensionality & $\PCBox$ & $\PCDiamond$ & $\PCFS$ & $\PCMay$ \\
        \bottomrule
    \end{tabular}
    \caption{Probabilistic testing equivalences in different testing contexts.}
    \label{tab:symbols}
    \vspace{-3mm}
\end{table}

\begin{corollary}
\label{cor:probabilistic_may_fs}
    The following statements are valid.
    \begin{enumerate}
        \item
        $(\PMay)=(\PDiamond)$, $(\PFS)=(\PBox)$, $\PBox\;\subsetneq\;\PDiamond$ and $\PFS\;\subsetneq\;\PMay$.
        \item
        $(\PCMay)=(\PCDiamond)$, $(\PCFS)=(\PCBox)$, $\PCBox\;\subsetneq\;\PCDiamond$ and  $\PCFS\;\subsetneq\;\PCMay$.
    \end{enumerate}
\end{corollary}

\subsection{Testing equivalences as congruences}
\label{subsec:box_diamond_congr}
In this section, we establish the congruence property of our testing equivalences.
Since the external characterizations align with the unifying internal characterizations, we focus exclusively on the latter.
Definition~\ref{def:probabilistic_extensionality} implies that $\PXDiamond$ and $\PXBox$ are $\mathcal{D}$-compositional and localizable. 
We further show that they are preserved by the convex combination operation, which can be regarded as a distribution extension of the $\bigoplus$ operation.
Especially, Theorem~\ref{thm:congr} demonstrates that when projected onto $\PRCCS$ or $\PCCS$, our testing equivalences $\Restr{(\PDiamond)}{RCCS}$, $\Restr{(\PBox)}{RCCS}$, $\Restr{(\PCDiamond)}{CCS}$, $\Restr{(\PCBox)}{CCS}$ are all congruences.
\begin{theorem} [Congruence]
    \label{thm:congr}
    The following statements are valid.
    \begin{enumerate}
        \item The equivalences $\PXDiamond$, $\PXBox$, $\PXFS$, and $\PXMay$ are preserved by the $\mathcal{D}$-composition, localization, and convex combination operation.
        \item The equivalences $\Restr{(\PDiamond)}{RCCS}$, $\Restr{(\PBox)}{RCCS}$, $\Restr{(\PCDiamond)}{CCS}$, and $\Restr{(\PCBox)}{CCS}$ are congruences.
        \item The equivalences $\Restr{(\PMay)}{RCCS}$, $\Restr{(\PFS)}{RCCS}$, $\Restr{(\PCMay)}{CCS}$, and $\Restr{(\PCFS)}{CCS}$ are congruences.
    \end{enumerate}
\end{theorem}

We first show that $\PDiamond$ and $\PXDiamond$ satisfy the above properties.
The following proposition, derived from Lemma~\ref{lem:show_membership}, summarizes the congruence properties of $\PDiamond$ except for the fixpoint operation.

\begin{restatable}{proposition}{combChoCongr}
    \label{prop:combination_choice_congr}
    The following statements are valid.
    \begin{enumerate}
        \item \label{item:combination_choice_congr_1}
        The equivalence $\PXDiamond$ is preserved by the convex combination operation.
        \item \label{item:combination_choice_congr_2}
        The equivalence $\Restr{(\PDiamond)}{RCCS}$ is preserved by the nondeterministic choice, probabilistic choice, parallel composition, and localization operations. 
    \end{enumerate}
\end{restatable}

Thus far, to show that $\Restr{(\PDiamond)}{RCCS}$ is a congruence, we merely need to prove that $\Restr{(\PDiamond)}{RCCS}$ is closed under the fixpoint operation.
Suppose $S$ and $T$ are two $n$-ary RCCS terms.
Define $S \PDiamond T$ if for tuple $\widetilde{P}$ of processes, the relation $S[\widetilde{P}]\PDiamond T[\widetilde{P}]$ always holds.
The other equivalence can be generalized to $n$-ary distributions similarly.
The standard congruence property of fixpoint operations is formalized as follows: given two $n$-ary processes $S, T$, if $S\PDiamond T$, then $\fix X.S\PDiamond\fix X.T$.
When $n=1$, we need the following lemma.

\begin{restatable}{lemma}{addTau}
    \label{lem:add_tau}
    Let $S, T$ be two $1$-ary terms.
    The following statements are valid.
    \begin{enumerate}
        \item If $S\PDiamond T$, then $\fix X.\tau.S\PDiamond\fix X.\tau.T$.
        \item $\fix X.T\PDiamond\fix X.\tau.T$.
    \end{enumerate}
\end{restatable}

The intuition behind this lemma is similar to the classical proofs of congruence for the CCS fixpoint operation, where the proof of congruence is relatively straightforward when the term is $\tau$-guarded. 
As for testing equivalences, a key feature is that the additional $\tau$-actions do not affect external observations, allowing us to conclude that congruence still holds even in the unguarded case.
Formally, assume that $S\PDiamond T$.
By Lemma~\ref{lem:add_tau}, one can derive that $\fix.S=\fix.\tau.S=\fix.\tau.T=\fix.T$.
Now the case of $n=1$ holds, and the following corollary follows immediately.

\begin{restatable}{corollary}{fixPoint}
    \label{cor:fix_point}
    The equivalence $\Restr{(\PDiamond)}{RCCS}$ is closed under the fixpoint operation.
\end{restatable}

Notably, the crucial property about $\PDiamond$ used in the proofs above is Corollary~\ref{cor:calculated_by_degenerate_witness} (\ref{item:calculated_by_degenerate_witness_1}).
Since the corresponding statement w.r.t. $\PBox$, i.e., Corollary~\ref{cor:calculated_by_degenerate_witness} (\ref{item:calculated_by_degenerate_witness_2}), also holds, we can prove that $\PBox$ also satisfies these properties.
Moreover, projecting $\PCBox$ and $\PCDiamond$ onto $\PCCS$ also yields the similar conclusions.
Therefore, we deduce that all statements in Theorem~\ref{thm:congr} are valid.

\subsection{Case study: probabilistic testing equivalences in pCSP}
\label{subsec:case}

The pCSP model~\cite{DBLP:journals/lmcs/DengGHM08,10.1007/978-3-642-04081-8_19,DENG2007359} is a probabilistic extension of CSP~\cite{10.1145/359576.359585}.
Deng et al. provided several characterizations for may and must preorders in pCSP~\cite{DBLP:journals/lmcs/DengGHM08}.
We mentioned earlier that one major advantage of our framework is its ease of adaptation to other models.
To strengthen this claim, we present a case study in this subsection, where we apply our method to pCSP and compare it with the notions proposed therein.

\paragraph{Syntax and semantics} 
We begin by providing a brief introduction to the pCSP model.
Here we reuse the notation introduced earlier: let $\Chan$ denote the set of all visible actions, let $L \subseteq \Chan$ be a finite subset of actions, and $\Act:=\Chan\cup\{\tau\}$. 
Note that we no longer consider the dual channels $\overline{\Chan}$ here.
The syntax of pCSP is defined as follows.
\begin{equation}
    \label{eq:pCSP}
    \begin{aligned}
        P,Q &:=s ~\Big{|}~ P \oplus_p Q,\\
        s,t &:=\mathbf{0} ~\Big{|}~ a.P ~\Big{|}~ P\sqcap Q ~\Big{|}~ s~\Box~t ~\Big{|}~ s \mid_L t.
    \end{aligned}
\end{equation}

In Equation~(\ref{eq:pCSP}), the process $P \oplus_p Q$ with $p \in (0,1)$ represents a probabilistic choice between $P$ and $Q$.
The \emph{state-based processes} generated by $s$ have a CSP-like syntax: the \emph{internal choice term} $P \sqcap Q$ performs an internal action $\tau$ and then nondeterministically evolves into either $P$ or $Q$, while the \emph{external choice term} $s~\Box~t$ disables one component once the other has performed an external action.
We write $\mathsf{pCSP}$, ranged over by $P,Q$, for the set of processes defined by Equation~(\ref{eq:pCSP}), and $\mathsf{sCSP}$, ranged over by $s,t$, for the set of state-based processes.
Each process $P$ in $\mathsf{pCSP}$ induces a distribution $\llbracket P\rrbracket$ over state-based processes, defined by
\begin{align}
    \llbracket s\rrbracket :=\delta_s,\qquad
    \llbracket P\oplus_p Q\rrbracket := p\cdot \llbracket P\rrbracket + (1-p)\cdot\llbracket Q\rrbracket.
\end{align}
This provides a formal description of the meaning of $P \oplus_p Q$.
Based on these notations, we present the operational semantics of pCSP using the pLTS shown in Figure~\ref{fig:pLTS-pCSP}, where $a \in \Chan$, $\alpha \in \Act$, and $\myrightarrow{}~\subseteq \mathsf{pCSP} \times \Act \times \Distr(\mathsf{sCSP})$.
\begin{figure}[h]
    \begin{center}
        \begin{displaymath}
        \frac{}{a.P\myrightarrow{a}\llbracket P\rrbracket} \enspace
        \frac{}{P\sqcap Q \myrightarrow{\tau} \llbracket P\rrbracket} \enspace 
        \frac{}{P\sqcap Q\myrightarrow{\tau} \llbracket Q\rrbracket} \enspace
        \frac{~s_1\myrightarrow{a}\rho~}{~s_1~\Box~s_2\myrightarrow{a} \rho~}
        \frac{~s_2\myrightarrow{a}\rho~}{~s_1~\Box~s_2\myrightarrow{a} \rho~} \enspace
        \frac{~s_1\myrightarrow{\tau}\rho~}{s_1~\Box~s_2\myrightarrow{\tau} \rho~\Box~s_2}
        \end{displaymath}	
        \begin{displaymath}
            \frac{~s_2\myrightarrow{\tau}\rho~}{~s_1~\Box~s_2\myrightarrow{\tau} s_1~\Box~\rho} \quad
            \frac{s_1 \myrightarrow{\alpha} \rho \quad \alpha \notin L}{~s_1\mid_L s_2 \myrightarrow{\alpha} \rho\mid_L s_2 ~}  \quad
            \frac{s_2 \myrightarrow{\alpha} \rho \quad \alpha \notin L}{~s_1\mid_L s_2 \myrightarrow{\alpha} s_1\mid_L \rho ~} \quad
            \frac{s_1 \myrightarrow{\alpha} \rho_1\quad s_2\myrightarrow{\alpha} \rho_2\quad \alpha\in L}{~s_1\mid_L s_2 \myrightarrow{\tau} \rho_1\mid_L \rho_2 ~} 
        \end{displaymath}
    \end{center}
    \caption{pLTS for pCSP.}
    \label{fig:pLTS-pCSP}
\end{figure}

Note that $\mid_L$ and $\Box$ were originally defined as operations on state-based processes, but in Figure~\ref{fig:pLTS-pCSP} we extend them to distributions, in the same way as in Section~\ref{sec:rccs_model}.
There are also some subtle differences between the pCSP and RCCS models, as illustrated by the following example.

\begin{example}
    The prefix and probabilistic choice terms are essentially consistent in both models.
    However, in pCSP, it is not possible to syntactically construct a process of the form $\tau.P$; instead, one can only construct the behaviorally equivalent process $P \sqcap P$, which satisfies $P \sqcap P \myrightarrow{\tau} \llbracket P\rrbracket$.
    Thus, we take the former as syntactic sugar for the latter within pCSP.
    Next, we compare the remaining operations.
    \begin{itemize}
        \item 
        There are two types of non-deterministic choice operations in pCSP.
        The internal choice $P \sqcap Q$ is equivalent to $\tau.P + \tau.Q$ in RCCS.
        The external choice $s_1 ~\Box~ s_2$ can also be translated as the $+$ operator in RCCS when both $s_1$ and $s_2$ are guarded by an external action.
        In other cases, however, the behavior may differ. 
        For example, we have the transition
        $(\tau.P) ~\Box~ Q \myrightarrow{\tau} \llbracket P\rrbracket ~\Box~ Q$,
        while $\tau.P + Q \myrightarrow{\tau} \delta_P$, where the two operations are no longer equivalent.
        Intuitively, in pCSP, $\tau$ is viewed as an internal evolution of the compound process $\tau.P ~\Box~ Q$, allowing the system to retain the branch $Q$.
        In RCCS, by contrast, it is treated as an action generated by $\tau.P$, which causes the branch $Q$ to be lost after the execution.

        \item The behavior of $P \mid_L Q $ is essentially the same as $(L)(P \mid Q)$. 
        However, a synchronization can only occur on channels within $L$ in  the former, while there is no such restriction in the latter.
    \end{itemize}
\end{example}
Despite these differences, we can directly combine the pLTS in Figure~\ref{fig:pLTS-pCSP} with the definition in Figure~\ref{fig:pLTS2} to obtain the distribution-based semantics for pCSP.
The resulting transition relation
$\myrightarrow{}~\subseteq~\Distr(\mathsf{sCSP}) \times \Act \times (0,1] \times \Distr(\mathsf{sCSP})$
still satisfies the desirable properties stated in Section~\ref{subsec:linearity}.
For notation, we continue to follow the previous conventions; for example, $\rightsquigarrow$ denotes the silent transition sequence.
Moreover, the testing framework introduced in Section~\ref{sec:testing_over_distributions} can be applied to pCSP without any modification.

\paragraph{Testing pCSP processes}
To test probabilistic pCSP processes, Deng et al. introduced a special action $\omega \notin \Act$ to report success.
Following the convention in Section~\ref{subsec:external_characterization}, we use $\mathsf{pCSP}^\omega$ and $\mathsf{sCSP}^\omega$ to denote the sets of pCSP tests and state-based tests, respectively.
A \emph{result gathering function} $\mathbb{V} : \mathsf{sCSP}^\omega \to \mathcal{P}^+([0,1])$ associates each state-based test process with a non-empty set of outcomes, where $\mathcal{P}^+([0,1])$ denotes the collection of non-empty subsets of $[0,1]$.
Under the convention $\mathbb{V}(\rho) := \sum_{P \in \Supp(\rho)} \rho(P)\mathbb{V}(P)$, the function is defined inductively as
\begin{equation}
    \mathbb{V}(s):=
    \begin{cases}
        \{1\}, \quad & \text{if }s\myrightarrow{\omega},\\
        \bigcup\{\mathbb{V}(\rho): s\myrightarrow{\alpha}\rho\},\quad &\text{if }s\not\myrightarrow{\omega}\text{ but }s\myrightarrow{},\\
        \{0\}, \quad & \text{if }s\not\myrightarrow{}.
    \end{cases}
\end{equation}

For any pCSP process $P$ and test $T$, define $\mathcal{A}(T,P):=\mathbb{V}(\llbracket T\mid_{\Chan} P\rrbracket)$.
Note that the set $L$ in an expression of the form $P \mid_L Q$ is required to be finite.
Therefore, in this context, the subscript $\Chan$ should be understood as $\Chan \cap \Chan(T, P)$.
Intuitively, $\mathcal{A}(T, P)$ represents the testing outcome of $P$ with respect to the test~$T$.

\begin{example}
    \label{exp:pcsp-outcome}
    Let $P:=a.((b.d~\Box~c.e)\oplus_{\frac{1}{2}} (b.f ~\Box~c.g))$, $Q:=a.((b.d~\Box~c.g)\oplus_{\frac{1}{2}} (b.f ~\Box~c.e))$, and $T:=a.((b.d.\omega\oplus_{\frac{1}{2}}c.e.\omega)\sqcap (b.f.\omega \oplus_{\frac{1}{2}}c.g.\omega))$.
    It is easy to verify (cf. Example 3.3 in \cite{DBLP:journals/lmcs/DengGHM08}) that
    \begin{equation}
        \mathcal{A}(T,P)=\{0, \frac{1}{2}, 1\}, \quad
        \mathcal{A}(T, Q)=\{\frac{1}{2}\},
    \end{equation}
    which are both discrete sets.
    Note that $\omega$ is the only external action that $P \mid_\Chan T$ and $Q \mid_\Chan T$ can perform.
    Under our approach, we obtain
    \begin{equation}
        \PEquip{P \mid_\Chan T}=[0,1],\quad \PEquip{Q \mid_\Chan T}=\{\frac{1}{2}\}.
    \end{equation}
    Observe that in this example, our testing outcomes coincide with the convex closure of those obtained by Deng et al.
\end{example}

The probabilistic testing preorders are defined with two underlying preorders,  namely the Hoare preorder $\le_{\mathrm{Ho}}$ and the Smyth preorder $\le_{\mathrm{Sm}}$, given by
\begin{align}
    X \le_{\mathrm{Ho}} Y \iff \forall x\in X, \exists y \in Y, x\le y,\\
    X \le_{\mathrm{Sm}} Y \iff \forall y\in Y, \exists x \in X, x\le y.
\end{align}
Using these two preorders, the may and must preorders introduced by Deng et al. are defined as follows.
\begin{definition}[cf. Definition 3.2, \cite{DBLP:journals/lmcs/DengGHM08}]
    \label{def:may-must-pcsp}
    The \emph{may} and \emph{must} preorders are defined by
    \begin{align}
        P \sqsubseteq_{\mathrm{may}} Q \iff \text{for all tests }T, \mathcal{A}(T,P) \le_{\mathrm{Ho}} \mathcal{A}(T,Q),\label{eq:hp}\\
        P \sqsubseteq_{\mathrm{must}} Q \iff \text{for all tests }T, \mathcal{A}(T,P) \le_{\mathrm{Sm}} \mathcal{A}(T,Q).\label{eq:sm}
    \end{align}
\end{definition}
These two preorders induce the external characterizations of the probabilistic may and must equivalences, denoted (only in this subsection) by $\PMay$ and $\PMust$, respectively.

\paragraph{Box and Diamond equivalences in pCSP}
Now we extend our unifying internal characterization for testing equivalences to pCSP.
As argued above, probabilistic equipollence (Definition~\ref{def:probabilistic_equipollence}) remains well-defined for relations over $\Distr(\mathsf{sCSP})$.
The only modification needed is in the definition of extensionality.
To stay consistent with Definition~\ref{def:may-must-pcsp}, we restrict our attention here to fully probabilistic testing contexts.
Since the composition and localization operators are replaced here by $\mid_L$, we naturally say that a relation $\R$ over $\Distr(\mathsf{sCSP})$ is \emph{probabilistically extensional} if $(\mu_1\mid_L \nu) ~\R~ (\mu_2 \mid_L \nu)$ holds whenever $\mu_1~\R~\mu_2$, $\nu \in \Distr(\mathsf{sCSP})$, and $L \subseteq \Chan$.
Then we obtain the following internal characterization of probabilistic testing equivalences for pCSP.
\begin{definition}
    The \emph{probabilistic diamond (box, resp.) equivalence} $\PDiamond$ ($\PBox$, resp.) is the largest (strongly, resp.) probabilistically equipollent, probabilistically extensional relation on $\Distr(\mathsf{sCSP})$.
\end{definition}

We now obtain two equivalences over $\Distr(\mathsf{sCSP})$.
Since every pCSP process denotes a distribution, we write $P \PDiamond Q$ if $\llbracket P\rrbracket \PDiamond \llbracket Q \rrbracket$, and $P \PBox Q$ if $\llbracket P\rrbracket \PBox \llbracket Q \rrbracket$.
Note that the right-hand sides of Equations \ref{eq:hp} and \ref{eq:sm} can equivalently be expressed in terms of the supremum and infimum of the testing outcomes, respectively. From this observation, the following correspondence can be deduced. 

\begin{restatable}{lemma}{pcspOut}
    \label{lem:pcsp-outcome}
    Let $\mu\in \Distr(\mathsf{sCSP})$ be a distribution. 
    Then $\max \mathbb{V}(\mu)=\charMay_\omega(\mu)$ and $\min \mathbb{V}(\mu)=\charFS_\omega(\mu)$.
\end{restatable}

Lemma 23 shows that the external characterization given by Deng et al. via the result gathering function coincides with our external characterization based on 
$\charMay$ and $\charFS$ (cf. Definitions~\ref{def:pmay} and Definition~\ref{def:pfs} for the RCCS model).
An illustrative example has already been given in Example \ref{exp:pcsp-outcome}.
With this correspondence at hand, we obtain Theorem \ref{thm:pcsp}, which states that our internal characterization matches precisely the external characterization of Deng et al.
The proof follows the same reasoning as the proof of Theorem~\ref{thm:model_indenpent_characterizations} in Section~\ref{subsec:external_characterization} and is therefore omitted here.

\begin{theorem}
    \label{thm:pcsp}
    Let $P, Q \in \mathsf{pCSP}$ be two processes.
    The following statements hold.
    \begin{enumerate}
        \item $P \PMay Q$ if and only if $P \PDiamond Q$.
        \item $P \PMust Q$ if and only if $P \PBox Q$.
        \item $\PMust ~\subsetneq~ \PMay$.
    \end{enumerate}
\end{theorem}

Interestingly, $\PBox$ would normally correspond to fair testing rather than must testing. 
As we discussed in Example~\ref{ex:classical_testing}, fair and must equivalences are incomparable in the classical setting. 
However, in the present context, $\PBox$ happens to coincide with must testing. 
Indeed, recall that fair equivalence was introduced to address issues of divergence. 
Since all pCSP processes are free of divergent behaviors, the two equivalences coincide, which explains why the second statement of Theorem~\ref{thm:pcsp} holds.
Moreover, following the reasoning of Theorem~\ref{thm:probabilistic_box_diamond}, it is straightforward to show that $\PXBox ~\subsetneq~ \PXDiamond$. 
This, in turn, implies that $\PMust~\subsetneq~\PMay$, indicating that must testing is strictly more discriminating than may testing. 
This reproduces the same conclusion as Proposition 5.2 in \cite{DENG2007359}.

To conclude, this subsection illustrates the transferability of our approach by applying it to the pCSP model. 
In fact, our method does not rely on the syntactic details of the underlying process language; once a suitable pLTS is obtained, the full distribution-based testing framework applies directly.
Regarding probabilistic testing equivalences, we expect them to possess sufficient congruence properties, which motivates the requirement of extensionality. 
In principle, one could examine equivalences preserved by all operators of the new model, but selecting a minimal set of operators yields a cleaner and more concise definition.
	
\section{The world of testing equivalences}
\label{sec:comparison}
In Section~\ref{sec:model_independent}, we proved the correspondence between the unifying internal characterizations and the external characterizations. 
Now we proceed to establish the remaining theorems to complete the entire spectrum as summarized in Figure~\ref{fig:summary}.

\subsection{Comparison with classical testing equivalences}
\label{subsec:comparison_with_classical_testing}
We begin by establishing all the vertical arrows in Figure~\ref{fig:summary}, which illustrate the relationship between our probabilistic testing equivalences and the classical testing equivalences. 
The results can be classified into two categories based on the testing contexts.
\begin{enumerate}
    \item Our probabilistic testing equivalences in context $\Delta(\PCCS)$ are conservative generalizations of their classical counterpart (e.g., $\PCBox$ is the conservative generalization of $\CBox$).
    \item The probabilistic testing equivalences in context $\Distr(\PRCCS)$ are strictly finer than their classical testing counterpart, which intuitively demonstrates that probabilistic testing contexts provide greater distinguishing power than classical ones.
    Notably, the counterexamples that witness the strictness are just CCS processes (see Example~\ref{ex:extensionality_counterexp}).
\end{enumerate}

\paragraph{Internal characterizations}
The following lemma demonstrates that the characterization using testing outcomes is a generalization of (strong) observability.

\begin{restatable}{lemma}{equipClasGen}
\label{lem:equipollence_classical_generalization}
Let $P \in \PCCS$ be a CCS process, then we have
\begin{equation}
    \charMay_\L(P) = \begin{cases}
        1,  & \text{if $P \Downarrow$}, \\
        0,  & \text{otherwise}.
    \end{cases},\quad \text{and} 
    \quad 
    \charFS_\L(P) = \begin{cases}
        1,  & \text{if $P \observable$},\\
        0,  & \text{otherwise}.
    \end{cases}
\end{equation}
\end{restatable}

The intuition behind Lemma~\ref{lem:equipollence_classical_generalization} has already been illustrated in Example~\ref{ex:linearity_testing_outcome}.
In fact, the classical testing outcomes for CCS processes can be viewed as a discretization of probabilistic testing results, retaining only the values $0$ and $1$.
Since $P_1 \PXDiamond P_2$, it follows that $P_1 \CDiamond P_2$.
However, the converse does not hold in general, as both $\PXBox$ and $\PXDiamond$ require $\mathcal{D}$-extensionality, whose discriminating power may persist even when restricted to classical CCS processes.
This is demonstrated in the following example.
    
\begin{example}
\label{ex:extensionality_counterexp}
    Let $Q_5:=a.(\tau.b+\tau.c), Q_6:=a.b+a.c$.
    \begin{itemize}
        \item 
        Let $\R:=\left\{(R[Q_5],R[Q_6]):\text{$\R$ is a $1$-ary CCS term}\right\}$.
    It is easy to show that $\R$ is extensional, equipollent and strongly equipollent, which implies $Q_5\CBox Q_6$ and $Q_5\CDiamond Q_6$.
        \item 
        Let $\R^\Delta:=\left\{(\delta_{R[Q_5]},\delta_{R[Q_6]}):\text{$\R$ is a $1$-ary CCS term}\right\}$.
        Then $\R^\Delta$ is $\Delta$-extensional.
        By Lemma~\ref{lem:equipollence_classical_generalization} and the above conclusion, we also have $Q_5\PCBox Q_6$ and $Q_5\PCDiamond Q_6$ (by Corollary~\ref{cor:probabilistic_may_fs}).
        \item However, neither $Q_5\PBox Q_6$ nor $Q_5\PDiamond Q_6$ holds.
        It suffices to prove the latter one.
        Let $O=\bar{a}.(\frac{1}{2}\tau.\bar{b}.d\oplus \frac{1}{2}\tau.\bar{c}.d)$.
        Suppose $Q_5\PDiamond Q_6$, then by the $p$-extensionality, we have $\charMay_\L((abc)(Q_5\mid O))=\charMay_\L((abc)(Q_6\mid O))$.
        The following transition sequence shows that $\charMay_\L((abc)(Q_5\mid O))=1$.
        \begin{align}
            \delta_{(abc)(Q_5\mid O)}&\myrightarrow{\tau} \delta_{(abc)((\tau.b+\tau.c)\mid (\frac{1}{2}\tau.\bar{b}.d\oplus \frac{1}{2}\tau.\bar{c}.d))}\\
            &\myrightarrow{\tau}\frac{1}{2}\delta_{(abc)((\tau.b+\tau.c)\mid \bar{b}.d)}+\frac{1}{2}\delta_{(abc)((\tau.b+\tau.c)\mid \bar{c}.d)}\\
            &\rightsquigarrow\;\;\frac{1}{2}\delta_{(abc)(\mathbf{0}\mid d)}+\frac{1}{2}\delta_{(abc)(\mathbf 0\mid d)}.
        \end{align}
        On the other hand, since the descendants of $(abc)(Q_6\mid O)$ can only perform $d$, we must eliminate the prefix $ab$ or $ac$ to maximize the probability of performing an external action.
        For example, one such transition sequence takes the form of:
        \begin{align}
            \delta_{(abc)(Q_6\mid O)}&\myrightarrow{\tau} \delta_{(abc)(b\mid (\frac{1}{2}\tau.\bar{b}.d\oplus \frac{1}{2}\tau.\bar{c}.d))}\\
            &\myrightarrow{\tau}\frac{1}{2}\delta_{(abc)(b\mid \bar{b}.d)}+\frac{1}{2}\delta_{(abc)(b\mid \bar{c}.d)}\\
            &\rightsquigarrow\;\;\frac{1}{2}\delta_{(abc)(\mathbf{0}\mid d)}+\frac{1}{2}\delta_{(abc)(b\mid \bar{c}.d)}.
        \end{align}
        After considering all cases, we have $\charMay_\L((abc)(Q_6\mid O))=\frac{1}{2}$, leading to a contradiction.
        Therefore, one has that $\PDiamond\;\not\subseteq\; \CDiamond$ and $\PBox\;\not\subseteq\;\CBox$.
    \end{itemize}
\end{example}

This observation can be formally described as follows.

\begin{restatable}{proposition}{propClasGen}
\label{prop:classical_generalization}
    The following statements are valid.
    \begin{enumerate}
        \item
        $\Restr{(\PDiamond)}{CCS}\subsetneq\Restr{(\PCDiamond)}{CCS}=(\CDiamond)$.
        \item
        $\Restr{(\PBox)}{CCS}\subsetneq\Restr{(\PCBox)}{CCS}=(\CBox)$.
    \end{enumerate}
\end{restatable}

\paragraph{External characterizations}
For external characterizations, the proofs closely resemble Lemma~\ref{lem:equipollence_classical_generalization} and Proposition~\ref{prop:classical_generalization}.
We first show that our predicate-based testing framework generalizes the may and fair testing and then Proposition~\ref{prop:testing_classical_generalization} follows.

\begin{lemma}
\label{lem:testing_classical_generalization}
    Let $P\in\PCCS$ be a CCS process, $O$ be an classical observer, then we have
    \begin{enumerate}
        \item \label{item:testing_classical_generalization_1}
        $\charMay_\omega(P\mid O)=
        \begin{cases}
            1,&\text{if some testing sequence of $P$ by $O$ is DH-successful,}\\
            0,&\text{otherwise.}\\
        \end{cases}$
        \item \label{item:testing_classical_generalization_2}
        $\charFS_\omega(P\mid O)=
        \begin{cases}
            1,&\text{if the test of $P$ by $O$ is fair-successful,}\\
            0,&\text{otherwise.}\\
        \end{cases}$
    \end{enumerate}
\end{lemma}

\begin{restatable}{proposition}{propTestClassGen}
\label{prop:testing_classical_generalization}
    The following statements are valid.
    \begin{enumerate}
        \item 
        $\Restr{(\PMay)}{CCS}\subsetneq\Restr{(\PCMay)}{CCS}=(\May)$
        \item 
        $\Restr{(\PFS)}{CCS}\subsetneq\Restr{(\PCFS)}{CCS}=(\FS)$
    \end{enumerate}
\end{restatable}

A major application of Proposition~\ref{prop:classical_generalization} and~\ref{prop:testing_classical_generalization} is that we can now confirm the validity of
Proposition~\ref{prop:classical_model_independent_characterization}, which formalizes the interrelation between classical diamond/box equivalences and may/fair equivalences.
This proof points out that our study of equivalences on distributions can exactly cover some study of equivalence relations on $\PCCS$.

\begin{proofof}{Proposition~\ref{prop:classical_model_independent_characterization}}
    By Theorem~\ref{thm:congr}, both $\Restr{(\PCBox)}{CCS}$ and $\Restr{(\PCDiamond)}{CCS}$ are congruences.
    Corollary~\ref{cor:probabilistic_may_fs} shows that $\Restr{(\PCFS)}{CCS}=\Restr{(\PCBox)}{CCS}\subsetneq \Restr{(\PCDiamond)}{CCS}=\Restr{(\PCMay)}{CCS}$.
    By Proposition~\ref{prop:testing_classical_generalization} and Proposition~\ref{prop:classical_generalization}, we have $(\FS)=(\CBox)\;\subsetneq\;(\CDiamond)=(\May).$
    \qed
\end{proofof}

\subsection{Comparison with probabilistic weak bisimilarity}
\label{subsec:comparison_weak}
In this subsection, we further explore distribution-based semantics by comparing our probabilistic testing equivalences with the famous probabilistic weak bisimilarity \cite{turrini_PolynomialTimeDecision_2015}.

\subsubsection{Probabilistic weak bisimilarity}
Let $\mathsf{Tr} = \{(A,\alpha,\rho)\mid \text{$A \myrightarrow{\alpha} \rho$ can be derived in the pLTS in Figure \ref{fig:pLTS1}}\}$ be the set of transitions.
    For a transition $\mathsf{tr} = (A,\alpha,\rho)$, we denote by $\mathsf{src}(\mathsf{tr})$ the source process $A$, by $\mathsf{act}(\mathsf{tr})$ the action $\alpha$, and by $\rho_{\mathsf{tr}}$ the evolved distribution $\rho$.
    Let $\mathsf{Tr}(\alpha) = \{\mathsf{tr}\in \mathsf{Tr} \mid \mathsf{act}(\mathsf{tr}) = \alpha\}$.
    An \emph{execution fragment} of some process $A_0$ is a finite or infinite sequence of alternating states and actions $\omega = A_0 \alpha_0 A_1 \alpha_1 A_2 \alpha_2 \cdots$ such that $A_i \myrightarrow{\alpha_i} \rho_i$ and $\rho_i(A_{i+1}) >0$.
    If $\omega$ is finite, we denote by $\mathsf{last}(\omega)$ the last state of $\omega$.
    We denote by $\mathsf{frags}^{*}(A)$ and $\mathsf{frags}(A)$ the set of finite and all execution fragments of $A$, respectively.
    Given $\alpha\in \Act$, we define $\widehat{\alpha} = \alpha$ if $\alpha \in \L$, and $\widehat{\alpha} = \varepsilon$ (the empty string) if $\alpha =\tau$.
    The \emph{trace} of an execution fragment $\omega$ is the sub-sequence of external actions of $\omega$, i.e., $\mathsf{trace}(\omega) = \widehat{\alpha_0} \widehat{\alpha_1}\widehat{\alpha_2}\cdots$.
    
    In \cite{turrini_PolynomialTimeDecision_2015}, the notion of \emph{scheduler} is used  to resolve non-determinism.
    To a process $A$, a scheduler is a function $\sigma:\mathsf{frags}^{*}(A) \to \Distr(\mathsf{Tr}\cup \{\bot\})$ such that for each $\omega \in \mathsf{frags}^{*}(A), \sigma(\omega) \in \Distr(\{\mathsf{tr}\in\mathsf{Tr}\mid \mathsf{src}(\mathsf{tr})= \mathsf{last}(\omega)\}\cup \{\bot\})$. 
    We call a scheduler $\sigma$ (of $A$) $Dirac$ if for each $\omega\in \mathsf{frags}^{*}(A)$, $\sigma(\omega) = \delta_{\mathsf{tr}}$ for some $\mathsf{tr} \in \mathsf{Tr}$ or $\sigma(\omega) = \delta_{\bot}$.
    We call a scheduler $\sigma$ (of $A$) \emph{finite} if there exists a finite natural number $N$ such that $\sigma(\omega) = \delta_{\bot}$ for all $\omega$ with $|\omega| \ge N$.
    A scheduler $\sigma$ and a process $A$ induce a probability distribution $\rho_{\sigma,A}$ over finite execution fragments as follows.
    The basic measurable events are the \emph{cones} of finite execution fragments, where the cone of $\omega$ is defined by $C_{\omega} = \{\omega'\in \mathsf{frags}(A)\mid \text{$\omega$ is a prefix of $\omega'$}\}$.
    The probability $\rho_{\sigma,A}$ of a cone $C_{\omega}$ is defined recursively as follows:
    \begin{equation}
        \rho_{\sigma,A}(C_{\omega})= \begin{cases}
    	1, & \text{if $\omega = A$,} \\
    	0, & \text{if $\omega=B\neq A$}, \\
    	\rho_{\sigma,A}(C_{\omega'})\cdot\sum_{\mathsf{tr}\in \mathsf{Tr}(\alpha)}\sigma(\omega')(\mathsf{tr})\cdot \rho_{\mathsf{tr}}(B), & \text{if $\omega = \omega'\alpha B$.}
    \end{cases}
    \end{equation}
    Finally, for any $\omega \in \mathsf{frags}^{*}(A)$, $\rho_{\sigma,A}(\omega)$
    is defined as $\rho_{\sigma,A}(\omega) = \rho_{\sigma,A}(C_{\omega}) \cdot \sigma(\omega)(\bot)$, where $\sigma(\omega)(\bot)$ is the probability of choosing no further transition  (i.e., terminating) after $\omega$.

    In the context of weak bisimulation, we will encounter distributions with infinite support. 
    We use $\InfDistr(\PRCCS)$ to represent all countable support distributions over probabilistic processes.
    the standard lifting introduced earlier in Definition~\ref{def:lifting} from $\Distr(\PRCCS)$ to $\InfDistr(\PRCCS)$.
    Here, we slightly abuse the notation, as they are essentially consistent, with their specific meaning depending on the distribution provided.
    
    The next definition of weak combined transition is standard \cite{segala_ModelingVerificationRandomized_1995, turrini_PolynomialTimeDecision_2015}.
    The fact that state $A$ can weakly transfer to  distribution $\rho$ by executing an observable action $\alpha$ is defined as follows: if there exists a scheduler $\sigma$, from $A$ by doing $\alpha$ and  finite number of silent actions following $\sigma$, the probability of the final state being $B$ equals $\rho(B)$.
    \begin{definition}[Weak combined transition]
        \label{def:weak_combined_transition}
    	Given a process $A \in \PRCCS$, an action $\alpha \in \Act$ and a distribution $\rho \in \InfDistr(\PRCCS)$.
    	We say that there is a \emph{weak combined transition} from $A$ to $\rho$ labeled by $\alpha$, denoted by $A \stackrel{\alpha}{\Longrightarrow}_{c} \rho$, if there exists a scheduler $\sigma$ such that the following holds for the induced distribution $\rho_{\sigma,A}$:
    	\begin{enumerate}
    		\item $\rho_{\sigma,A}(\mathsf{frags}^{*}(A)) = 1$.
    		\item For each $\omega \in \mathsf{frags}^{*}(A)$, if $\rho_{\sigma,A}(\omega)>0$ then $\mathsf{trace}(\omega) = \widehat{\alpha}$.
    		\item For each process $B$, $\rho_{\sigma,A}\{\omega\in \mathsf{frags}^{*}(A)\mid \mathsf{last}(\omega)=B\} = \rho(B)$.
    	\end{enumerate}
    \end{definition}
    
    The next definition resembles the classical conception for probabilistic automata \cite{segala_ModelingVerificationRandomized_1995}.
    \begin{definition}[Probabilistic weak bisimulation]
    \label{def:PWB}
    	An equivalence $\mathcal{E}$ on $\PRCCS$ is a \emph{probabilistic weak bisimulation} if, for all $(A,B) \in \mathcal{E}$, if $A \myrightarrow{\alpha} \rho_A$, then there exists $\rho_B\in\InfDistr(\PRCCS)$ such that $B \stackrel{\alpha}{\Longrightarrow}_{c} \rho_B$ and $\rho_A\; \mathcal{E}^{\dagger} \;\rho_B$.

    \end{definition}

We write $A \PWeak B$ if there is a probabilistic weak bisimulation $\mathcal{E}$ such that $(A,B) \in \mathcal{E}$.

    \begin{theorem}[\cite{turrini_PolynomialTimeDecision_2015}]
    \label{thm:pwb-equivalence}
    	$\PWeak$ is an equivalence relation, and it is the largest probabilistic weak bisimulation.
    \end{theorem} 

\begin{example}
    Let $Q_1:=\frac{1}{2}\tau.a\oplus\frac{1}{2}\tau.b$, $Q_2:=\fix X.\left(\frac{1}{3}\tau.a\oplus\frac{1}{3}\tau.b\oplus\frac{1}{3}\tau.X\right)$.
    We now proceed to verify the claim that $Q_1 \PWeak Q_2$ in Example~\ref{ex:probabilistic_transition}.
    Define
    \begin{equation}
        \R:=\{(Q_1, Q_2), (a,a), (b,b), (\mathbf{0},\mathbf{0})\}.
    \end{equation}
    Next we need to show that $\R$ is a probabilistic weak bisimulation.
It suffices to consider the pair $(Q_1, Q_2)$, as detailed in Example~\ref{ex:probabilistic_transition} and Example~\ref{ex:oplus_testing_outcome}.
Note also that the infinite transition sequence illustrated in Figure~\ref{fig:infinite_case} can induce a weak combined transition.
Therefore we conclude that $Q_1 \PWeak Q_2$.
\end{example}

The schedule-based approach may involve a weak combined transition with an \emph{infinite} number of execution fragments, while our distribution-based approach uses \emph{finite} probabilistic transition sequences.
Therefore, it is possible to establish properties about probabilistic transition sequences by \emph{induction} on the length of the sequence (see Lemma~\ref{lem:linearity_of_probabilistic_transition_strengthened}).
In contrast, it is rather difficult to obtain similar results for weak combined transitions.

Moreover, the scheduler-based approach can be viewed as a typical method of resolving nondeterminism by means of probability: once a scheduler is fixed, all nondeterministic choices in a process are resolved according to the scheduler’s prescriptions, resulting in a single induced distribution.
For example, consider a simple process $P:=\tau.a + \tau.b$.
A scheduler may resolve the nondeterministic choice by triggering the two branches with probability $\tfrac{1}{2}$ each, thereby inducing the distribution $\tfrac{1}{2}\delta_a + \tfrac{1}{2}\delta_b$.
In this case, its behaviour becomes indistinguishable from that of the probabilistic process $\tfrac{1}{2}\tau.a \oplus \tfrac{1}{2}\tau.b$.
In contrast, our approach always preserves nondeterminism.
To simulate the behaviour described above, we would at least perform the following two probabilistic transitions:
\begin{equation}
    P \xrightarrow{(\tau,\frac{1}{2})} \frac{1}{2} \delta_P +\frac{1}{2}\delta_b \xrightarrow{(\tau,1)} \frac{1}{2}\delta_a+\frac{1}{2}\delta_b,
\end{equation}
where each step retains the underlying nondeterminism.
Consequently, in the three additional ``parallel worlds'' generated by these nondeterministic choices, we obtain the distributions $\delta_a, \delta_b$, and $\frac{1}{2}\delta_a+\frac{1}{2}\delta_b$.
The behavioural differences between the two operators, $+$ and $\oplus$, under our semantics have already been illustrated in Example~\ref{ex:semantics}.
Thus, the mechanisms of the two approaches differ conceptually.
Nevertheless, the resulting behaviours can mutually simulate each other. As explained above, a single scheduler step can be simulated through two or more probabilistic transitions in our framework, whereas a single probabilistic transition in our approach can be reproduced in the scheduler-based semantics using $\bot$.

Finally, it is worth noting that the scheduler-based approach distinguishes between all execution fragments, making it particularly suitable for formalizing trace-related equivalences, which typically possess stronger discriminating power. In fact, $\PWeak$ is strictly finer than the testing equivalences defined earlier, a result that will be established formally in the subsequent sections. Nevertheless, we emphasize that the strength of our approach lies in the use of the silent-transition relation $\rightsquigarrow$, which significantly streamlines the formalization of testing equivalences within our unified testing framework.

\subsubsection{A spectrum of probabilistic equivalences}

We can lift the equivalences $\PWeak$ on $\PRCCS$ to equivalences $(\PWeak)^{\dagger}$ on $\Distr(\PRCCS)$ in the standard way (see Definition \ref{def:lifting}).
For convenience, we often abbreviate them as $\PWeak$ when it is clear from the context that we are discussing equivalences over $\Distr(\PRCCS)$.
We now aim to establish the relationship between $\PWeak$ and our testing equivalences.

\begin{theorem}[Spectrum]
\label{thm:probabilistic_hierarchy}
    For the equivalences $\PWeak$, $\PBox$ and $\PDiamond$ on $\Distr(\PRCCS)$, we have the strict inclusions
    \begin{equation}
        \PWeak \;\subsetneq\; \PBox \;\subsetneq\; \PDiamond.
    \end{equation}
\end{theorem}

Before delving into the formal proofs, we present an example to illustrate the underlying intuition.

\begin{example}

\begin{figure}[htb]
\centering
    \begin{minipage}[c]{.6\textwidth}
        \centering
        
\begin{tikzpicture}
\tikzset{vertex/.style = {}}
\tikzset{edge/.style = {->,> = latex'}}
\node[vertex] (a) at (1.5,1) {\footnotesize$B_1$};
\node[vertex] (b) at  (0,0) {\footnotesize$A_1$};
\node[vertex] (c) at  (3,0) {\footnotesize$A_2$};
\draw[edge] (a) to node[left=2, pos=0] {\footnotesize$\frac{1}{2}\tau$} (b);
\draw[edge] (a) to node[right=5, pos=0] {\footnotesize$\frac{1}{2}\tau$} (c);
\node[vertex] (d) at  (-1,-1) {\footnotesize$a$};
\node[vertex] (e) at  (1,-1) {\footnotesize$\mathbf0$};
\draw[edge] (b) to node[left=2, pos=0] {\footnotesize$\frac{3}{4}\tau$} (d);
\draw[edge] (b) to node[right=3, pos=0] {\footnotesize$\frac{1}{4}\tau$} (e);
\node[vertex] (f) at  (2,-1) {\footnotesize$a$};
\node[vertex] (g) at  (4,-1) {\footnotesize$\mathbf0$};
\draw[edge] (c) to node[left=2, pos=0] {\footnotesize$\frac{1}{4}\tau$} (f);
\draw[edge] (c) to node[right=3, pos=0] {\footnotesize$\frac{3}{4}\tau$} (g);
\end{tikzpicture}

    \end{minipage}
\hspace{0.05\textwidth}
    \begin{minipage}[c]{.3\textwidth}
        \centering

\begin{tikzpicture}
\tikzset{vertex/.style = {}}
\tikzset{edge/.style = {->,> = latex'}}
\node[vertex] (a) at  (0,0) {\footnotesize$a$};
\node[vertex] (b) at  (1,1) {\footnotesize$B_2$};
\node[vertex] (c) at  (2,0) {\footnotesize$\mathbf 0$};
\draw[edge] (b) to node[left=2, pos=0] {\footnotesize$\frac{1}{2}\tau$} (a);
\draw[edge] (b) to node[right=3, pos=0] {\footnotesize$\frac{1}{2}\tau$} (c);
\end{tikzpicture}

    \end{minipage}
\caption{Two processes can be probabilistically may equivalent without being probabilistically weak bisimilar.}
\label{fig:exp_weak}
\end{figure}
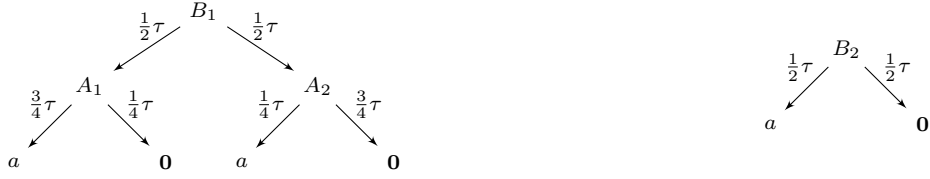

    Let $A_1:=\frac{3}{4}\tau.a\oplus\frac{1}{4}\tau.\mathbf{0}$, $A_2:=\frac{1}{4}\tau.a\oplus\frac{3}{4}\tau.\mathbf{0}$, $B_1:=\frac{1}{2}\tau.A_1\oplus\frac{1}{2}\tau.A_2$, and $B_2:=\frac{1}{2}\tau.a\oplus\frac{1}{2}\tau.\mathbf{0}$, as shown in Figure~\ref{fig:exp_weak}.
\begin{itemize}
    \item We claim that the transition $B_1 \myrightarrow{\tau} \frac{1}{2} \delta_{A_1} + \frac{1}{2} \delta_{A_2}$ cannot be simulated by $B_2$.
Otherwise, since the only possible transition for $B_2$ is $B_2 \myrightarrow{\tau} \frac{1}{2} \delta_a + \frac{1}{2} \delta_b$, or doing nothing, it would follow that $A_1$ must be simulated by $a$, $\mathbf{0}$, or $B_2$.
However, it is easy to verify that $A_1 \not\PWeak \mathbf{0}$ and $A_1 \not\PWeak a$.
Moreover, $A_1$ can reach action $a$ with probability $\frac{3}{4}$, while $B_2$ can reach $a$ with probability at most $\frac{1}{2}$.
Thus, none of these options suffice to simulate the transition.
We conclude that $B_1 \not\PWeak B_2$.
    \item By Lemma~\ref{lem:show_membership}, it is easy to show that $\delta_P\PBox\delta_Q$.
    The intuition is that although $B_2$, $a$, and $\mathbf{0}$ cannot individually simulate $B_1$, the transition 
    \begin{equation}
        B_2\myrightarrow{\tau}\frac{1}{2}\delta_a+\frac{1}{2}\delta_{\mathbf0}=\frac{1}{2}(\frac{3}{4}\delta_a+\frac{1}{4}\delta_\mathbf{0})+\frac{1}{2}(\frac{1}{4}\delta_a+\frac{3}{4}\delta_\mathbf{0})
    \end{equation}
    allows $A_1$ to be simulated by the distribution $\frac{3}{4}\delta_a+\frac{1}{4}\delta_\mathbf{0}$ and $A_2$ by $\frac{1}{4}\delta_a+\frac{3}{4}\delta_\mathbf{0}$.
\end{itemize}
In conclusion, probabilistic weak bisimulation allows one process to simulate another process, whereas probabilistic testing equivalences allow one distribution to simulate another distribution.
From this example, we may conjecture that $\PWeak\;\subsetneq\;\PBox$, where the strictness is already witnessed by $B_1$ and $B_2$.
\label{ex:pweak}
\end{example}

Now we formalize what we mean by ``simulate a distribution by another distribution'' in Example~\ref{ex:probabilistic_transition} and Example~\ref{ex:pweak}.
Given an equivalence relation $\E$ on $\PRCCS$, we define the \emph{$\E$-norm} of a function $g:\PRCCS\to\mathbb R$ as $|g|_{\E}:= \sum_{\mathcal{C} \in \PRCCS / \E} | \sum_{x\in\mathcal{C}}g(x)|$.
Building on this, we say that $\mu_2$ can \emph{$\mathcal{E}$-almost simulate} $\mu_1$ if for any transition $\mu_1 \rightsquigarrow \nu_1$ and any $\epsilon>0$, there exists a transition $\mu_2 \rightsquigarrow \nu_2$ such that $\vert\nu_1 - \nu_2\vert_{\mathcal{E}} \le \epsilon$.
Lemma~\ref{lem:PWeak_almost_bisimulation} states that any distributions $\mu_1, \mu_2$ satisfying $\mu_1 \PWeak \mu_2$ can $\PWeak$-almost simulate each other.

\begin{restatable}{lemma}{pweakAlmostBisim}
\label{lem:PWeak_almost_bisimulation}
    Given two distributions $\mu_1,\mu_2 \in \Distr(\PRCCS)$ and a number $\epsilon > 0$. 
    If $\mu_1 \PWeak \mu_2$ and $\mu_1 \ptran \nu_1$ with a degenerate witness $\pi_1 \in \tau^{*}$, then there exists a distribution $\nu_2$ such that $\mu_2 \ptran \nu_2$ and $|\nu_2 - \nu_1|_{\PWeak} \le \epsilon$.
\end{restatable}

Furthermore, we aim to leverage the power of $\PWeak$-almost simulation to demonstrate that $\PWeak$ is a strongly equipollent relation.
We have the following observations.
\begin{enumerate}
    \item If $P\PWeak Q$ and $P\myrightarrow{\ell}$ for some $\ell\in \L$, by Lemma~\ref{lem:PWeak_almost_bisimulation}, we can obtain a transition sequence $\delta_Q\rightsquigarrow\nu$ such that $\nu(\psi_\L)\ge1-\epsilon$ for any $\epsilon>0$.
    Thus $\charMay_\L(Q)=1=\charMay_\L(P)$.
    \item If $\mu_1\PWeak\mu_2$ and $\mu_1\ptran \nu_1$ with a degenerate witness, again by Lemma~\ref{lem:PWeak_almost_bisimulation}, we can obtain a transition sequence $\mu_2\rightsquigarrow\nu_2$ and $|\nu_2 - \nu_1|_{\PWeak} \le \epsilon$ for any small $\epsilon$.
    Intuitively, up to an arbitrarily small error, for every process $P \in \Supp(\nu_1)$ that contributes to $\nu_1(\psi_\L)$, i.e., $P\in\psi_\L$, there exists a process $Q\in\Supp(\nu_2)$ with almost the same probability $\nu_2(Q)\approx\nu_1(P)$ such that $\charMay_\L(Q) = 1$. 
    As a result, the values $\nu_1(\psi_\L)$ and $\charMay_\L(\nu_2)$ differ only by a negligible constant.
    Under the action of the supremum operator, such discrepancies are eliminated.
    Hence, we obtain $\charMay_\L(\mu_1) = \charMay_\L(\mu_2)$ without difficulty.
    \item 
    We have already established the $p$-equipollence of $\PWeak$.
    Building on this, a similar analysis should convince us that $\charFS_\L(\mu_1) = \charFS_\L(\mu_2)$ holds for any $\mu_1 \PWeak \mu_2$.
\end{enumerate}
Following the above reasoning, we immediately obtain Lemma~\ref{lem:PWeak_is_stronglyEquipollent}.
See Appendix~\ref{app:proof_strong_equipollence} for a formal proof.

\begin{restatable}{lemma}{pweakIsStr}
    \label{lem:PWeak_is_stronglyEquipollent}
    $\PWeak$ is a probabilistically strongly equipollent relation on $\Distr(\PRCCS)$.
\end{restatable}

It is well-known that the weak combined transitions are compositional~\cite{segala_ModelingVerificationRandomized_1995}.
Therefore, we can obtain the $p$-extensionality (i.e., $\Distr(\PRCCS)$-extensionality) of $\PWeak$ and deduce Theorem~\ref{thm:pweak_vs_pbox}.

\begin{restatable}{lemma}{pweakIsExt}
\label{lem:pweak_is_extensional}
    $\PWeak$ is $p$-extensional.
\end{restatable}

\begin{theorem}
\label{thm:pweak_vs_pbox}
    $\PWeak \;\subsetneq\; \PBox$.
\end{theorem}

\begin{proof}
    Lemma~\ref{lem:PWeak_is_stronglyEquipollent} and Lemma~\ref{lem:pweak_is_extensional} show that $\PWeak$ is a probabilistically strongly equipollent and $p$-extensional relation on $\Distr(\PRCCS)$.
    The strictness is shown in Example~\ref{ex:pweak}.
    \qed
\end{proof}

Theorem~\ref{thm:probabilistic_hierarchy} follows immediately from Theorem~\ref{thm:pweak_vs_pbox} and Corollary~\ref{cor:probabilistic_may_fs}.

\section{Conclusion}\label{sec:conclusion}

In this paper, we have proposed a new distribution-based semantics for the RCCS model and investigated several probabilistic testing equivalences. 
We have formalized a probabilistic testing framework with respect to process predicates and introduced two types of characterizations for testing equivalences: the unifying characterizations ($\PXDiamond$ and $\PXBox$) and the observer-based external characterizations ($\PXMay$ and $\PXFS$).
Both families of equivalences are parametrised by a given testing context~$\mathcal{D}$.
These equivalences can be regarded as generalizations of their classical counterparts under the classical testing context. 
We have also conducted an in-depth comparison between probabilistic weak bisimilarity, probabilistic branching bisimilarity, and our testing equivalences. 
It should be emphasized that our equivalence relations are congruences.
A similar conclusion has been corroborated in our case study on the pCSP model as well, further demonstrating the generalizability of our framework.

Our testing equivalences are defined with respect to two classes of characteristic quantities, $\charMay_\varphi$ and $\charFS_\varphi$.
It is important to note that our definitions (Definition~\ref{def:probabilistic_equipollence} and Definition~\ref{def:prob_may_fs_predicate}) require that the two related distributions to \emph{exactly} match on these characteristics.
However, this is a rather strong requirement, especially for probabilistic models such as probabilistic Turing machines or interactive proof systems~\cite{arora2006computational}, where small errors are often tolerated.
This observation suggests investigating \emph{weaker forms} of testing equivalence that permit small discrepancies in the characteristic values of related distributions.
For example, if two processes differ by no more than $0.001$ in their probability of eventual success under all tests, we may consider them interchangeable in most practical contexts. 
Exploring the properties and implications of such relaxed testing equivalences presents a promising, practical direction for future work.

Our new framework can be extended in several aspects. Probabilistic must equivalence is not the focus of this paper, but the testing framework developed in Section~\ref{sec:testing_over_distributions} remains applicable to its definition and analysis.
For instance, we can define a new characteristic $\chi^{\mathsf{must}}_{\varphi}$, which may take the form of $\inf\mathcal O^{\mu}_{\varphi}$, to provide internal and external characterizations of probabilistic must equivalence. 
Furthermore, the discussion of divergence in Section~\ref{subsec:structure_of_testing_outcome} indicates that the predicate-based framework can also be applied to \emph{divergence-sensitive} equivalences.
The probabilistic transition sequence introduced in Section~\ref{sec:rccs_model} can serve as a foundation for alternative characterizations of other probabilistic equivalences. 
The linearity and continuity of testing outcomes are also notable properties, offering a natural connection to recent research on metrics~{\cite{du2022behavioural,dal2023contextual,spork2024spectrum}}. 
Investigating a metric-based characterization of probabilistic testing equivalences thus represents another fruitful direction for future work.

One important direction for future work concerns decision algorithms for the proposed equivalences.
In classical testing equivalence, the may equivalence checking problem between finite-state processes can be reduced to the trace equivalence problem of nondeterministic automata, and is therefore in $\mathbf{PSPACE}$.
It is natural to explore whether similar ideas can be adapted to the probabilistic setting.
Indeed, several existing studies have already provided characterizations of probabilistic trace equivalence~\cite{bernardo2014revisiting,bernardo2022probabilistic_2022}.
However, the presence of probabilistic branching means that even after fixing a sequence of nondeterministic choices, a probabilistic trace typically becomes a nontrivial labelled tree rather than a simple sequence of actions.
Determining whether two processes generate the same set of such labelled trees is, intuitively, substantially more difficult.
Consequently, designing decision procedures for may or fair equivalences in probabilistic concurrent systems appears to be significantly more challenging.


\section*{Acknowledgments}
We thank BASICS members for their instructive discussions and feedback.
We also thank the anonymous referees for their valuable comments and instructive suggestions. Their efforts have greatly helped improve the quality of this manuscript.
The support from the National Science Foundation of China (62572319) is acknowledged.

\bibliographystyle{elsarticle-num} 
\bibliography{contents/testing}

\appendix
\setcounter{section}{0}
\setcounter{subsection}{0}
\renewcommand{\thesection}{\Alph{section}}
\renewcommand{\thesubsection}{\thesection.\arabic{subsection}}

\section{Proofs for Section \ref{sec:rccs_model}}
\label{appendix:proof_rccs}

\LinOfProbTransStr*
\begin{proof}
    The cases are trivial for $p = 0$ or $p = 1$. 
    Now we assume $p\in (0,1)$.
    \begin{enumerate}
        \item 
        Prove by induction on $|\pi_1|+|\pi_2|$.
        The base case is trivial.
        Assume that $|\pi_1|+|\pi_2\vert=k+1>0$.
        Then at least one of $|\pi_1|,|\pi_2|$ are non-zero. 
        W.l.o.g., we assume that $|\pi_1|>0$.
        Then the original transition sequence $\mu_1\myrightarrow{\pi}\nu_1$ can be divided into $\mu_1\myrightarrow{\pi_1'}\nu_1'\xrightarrow{(\alpha,q)}\nu_1$, where $\pi_1=\pi_1'\circ (\alpha,q)$ and $q\in(0,1]$.
        Since $|\pi_1'|+|\pi_2\vert=k$, by applying the induction hypothesis we have $p\mu_1+(1-p)\mu_2\myrightarrow{\pi'} p\nu_1'+(1-p)\nu_2$ and $\vert\pi'\vert\le k$.
        By definition, $\nu_1'\xrightarrow{(\alpha,q)}\nu_1$ implies there exists some $P\in\Supp(\nu_1')$ such that $P\myrightarrow{\alpha}\rho$ and $\nu_1=\nu_1'+\nu_1'(P)q(\rho-\delta_P)$.
        Note that $p\nu_1'(P)+(1-p)\nu_2(P)\neq 0$ for $p\in(0,1)$.
        We can define $r:=\frac{\nu_1'(P)pq}{p\nu_1'(P)+(1-p)\nu_2(P)}$.
        Still considering the aforementioned transition $P\myrightarrow{\alpha}\rho$, we have
        \begin{equation}
            \begin{aligned}
            p\nu_1'+(1-p)\nu_2\xrightarrow{(\alpha,r)}\nu:=p\nu_1'+(1-p)\nu_2+[p\nu_1'+(1-p)\nu_2](P)r(\rho-\delta_P)=p\nu_1+(1-p)\nu_2.
        \end{aligned}
        \end{equation}
        Let $\pi:=\pi'\circ (\alpha,r)$.
        We have $p\mu_1+(1-p)\mu_2\myrightarrow{\pi} p\nu_1+(1-p)\nu_2$ and $\vert\pi\vert\le k+1$.
        \item 
        Prove by induction on $|\pi|$.
        In the base case where $\vert\pi\vert=0$, set $\nu_1=\mu$, $\nu_2=\mu$, and the result follows.
        When $\vert\pi\vert=k+1$, we can divide the transition sequence corresponding to $p\mu_1+(1-p)\mu_2\myrightarrow{\pi} \nu$ into $p\mu_1+(1-p)\mu_2\myrightarrow{\pi'} \nu'\myrightarrow{(\alpha,q)}\nu$, where $\pi=\pi'\circ(\alpha,q)$ and $q\in(0,1]$.
        Since $\vert\pi'\vert=k$, by the induction hypothesis, there exists $\nu_1',\nu_2'\in\Distr(\PRCCS)$ such that $\mu_1\myrightarrow{\pi_1'}\nu_1'$, $\mu_2\myrightarrow{\pi_2'}\nu_2'$ and $\nu'=p\nu_1'+(1-p)\nu_2'$.
        Since $ \nu'\myrightarrow{(\alpha,q)}\nu$, there exists a process $P\in\Supp(\nu')$ such that $P\myrightarrow{\alpha}\rho$ and $\nu=\nu'+\nu'(P)q(\rho-\delta_P)$. 
        For $i=1,2$, let $\nu_i:=\nu_i'+\nu_i'(P)q(\rho-\delta_P)$.
        If $P\in\Supp(\nu_i')$, then $\nu_i'\xrightarrow{(\alpha,q)}\nu_i$; otherwise, $\nu_i'\myrightarrow{\varepsilon}\nu_i'=\nu_i$ holds trivially.
        These two cases can be unified as $\nu_i'\myrightarrow{\pi''_i}\nu_i$, where $\vert\pi''_i\vert\le 1$.
        Let $\pi_i:=\pi_i'\circ\pi_i''$.
        Then $\vert\pi_i\vert\le k+1$ and $\mu_i\myrightarrow{\pi_i}\nu_i$.
        We are done by verifying 
        \begin{equation}
            \begin{aligned}
            \nu&=\nu'+\nu'(P)q(\rho-\delta_P)\\
            &=[p\nu_1'+(1-p)\nu_2']+[p\nu_1'+(1-p)\nu_2'](P)q(\rho-\delta_P)\\
            &=p[\nu_1'+\nu_1'(P)q(\rho-\delta_P)]+(1-p)[\nu_2'+\nu_2'(P)q(\rho-\delta_P)]\\
            &=p\nu_1+(1-p)\nu_2.
        \end{aligned}
        \end{equation}
    \end{enumerate}
    \qed
\end{proof}
\section{Proofs for Section~\ref{sec:testing_over_distributions}}
\label{appendix:framework}

\linOfTestOut*

\begin{proof}
    For any $q\in \mathcal{O}^{p\mu_1+(1-p)\mu_2}_\varphi$, by definition, there exists a distribution $\nu$ such that $p\mu_1+(1-p)\mu_2\rightsquigarrow\nu$ and $\nu(\varphi)=q$.
    According to Corollary~\ref{cor:linearity_of_probabilistic_transition}~(\ref{item:linear_decomposition}), there exist two distributions $\nu_1,\nu_2\in \Distr(\PRCCS)$ such that $\mu_1\rightsquigarrow\nu_1$, $\mu_2\rightsquigarrow\nu_2$ and $\nu=p\nu_1+(1-p)\nu_2$. 
    Then $q=\nu(\varphi)=p\nu_1(\varphi)+(1-p)\nu_2(\varphi)\in p\mathcal{O}^{\mu_1}_{\varphi}+(1-p)\mathcal{O}^{\mu_2}_{\varphi}$.
    By the arbitrariness of $q$, we have $\mathcal{O}^{p\mu_1+(1-p)\mu_2}_\varphi\subseteq p\mathcal{O}^{\mu_1}_{\varphi}+(1-p)\mathcal{O}^{\mu_2}_{\varphi}$.
    
    For the other direction, consider any distributions $\nu_1$ and $\nu_2$ to which $\mu_1\rightsquigarrow\nu_1$ and $\mu_2\rightsquigarrow\nu_2$ hold.
    Again by Corollary~\ref{cor:linearity_of_probabilistic_transition}~(\ref{item:linear_combination}), we have $p\mu_1+(1-p)\mu_2\rightsquigarrow \nu:=p\nu_1+(1-p)\nu_2$.
    Hence, $p\nu_1(\varphi)+(1-p)\nu_2(\varphi)\in \mathcal{O}^{p\mu_1+(1-p)\mu_2}_\varphi$, which implies
    $p\mathcal{O}^{\mu_1}_{\varphi}+(1-p)\mathcal{O}^{\mu_2}_{\varphi}\subseteq\mathcal{O}^{p\mu_1+(1-p)\mu_2}_\varphi$.
    This finishes the proof of Equation~\ref{item:linearity_of_testing_outcome_1}.
    Since the supremum and infimum operations are commutative with the Minkowski sum and non-negative scalar product, Equation~\ref{item:linearity_of_testing_outcome_2} follows immediately.
    \qed
\end{proof}

\bdByDegenWit*

\begin{proof}
    Assume that $\mu\rightsquigarrow\nu$ is witnessed by $\pi$.
    We prove this lemma by induction on $\vert\pi\vert$.
    In the base case $\vert\pi\vert=0$, the conclusion holds by setting $\nu_1=\nu_2=\mu$.
    Assume that $\vert\pi\vert=k+1$.
    We can divide the transition sequence into $\mu\xrightarrow{(\tau,p)}\mu'\myrightarrow{\pi'}\nu$, where $\pi=(\tau,p)\circ\pi'$ and $p\in(0,1]$.
    By definition, there exists a process $P\in\Supp(\mu)$ such that $P\myrightarrow{\tau}\rho$ and $\mu'=\mu+\mu(P)p(\rho-\delta_P)$.
    Define $\mu_1':=\mu,\;\mu_2'=\mu+\mu(P)(\rho-\delta_P)$.        
    Note that $\mu'=(1-p)\mu_1'+p\mu_2'$.
    By Lemma~\ref{lem:linearity_of_probabilistic_transition_strengthened}, 
    there exist $\nu_1', \nu_2'$ such that
    $\mu'_1\myrightarrow{\pi_1''}\nu_1'$, $\mu'_2\myrightarrow{\pi_2''}\nu_2'$ and $\nu=(1-p)\nu_1'+p\nu_2'$,
    where $\vert\pi_1''\vert,\vert\pi_2''\vert\le k$. 
    Assume, without loss of generality, that $\nu_1'(\varphi)\leq \nu(\varphi)\leq \nu_{2}'(\varphi)$.
    Since $\vert\pi_1''\vert,\vert\pi_2''\vert\le k$,
    by the induction hypothesis, there exist two distributions $\nu_1,\nu_2$ such that (i) $\mu'_1\myrightarrow{\pi_1'}\nu_1$, $\mu'_2\myrightarrow{\pi_2'}\nu_2$ with $\pi_1',\pi_2'$ degenerate, and (ii) $\nu_1(\varphi)\le\nu_1'(\varphi)\le\nu(\varphi)\le\nu_2'(\varphi)\le\nu_2(\varphi)$.
    Let $\pi_1:=\pi_1'$ and $\pi_2:=\tau\circ\pi_2'$, which are both degenerate.
    Then $\mu\myrightarrow{\pi_1}\nu_1$ and $\mu\myrightarrow{\pi_2}\nu_2$ hold, which finishes the proof.
    \qed
\end{proof}

\section{Proofs for Section \ref{sec:model_independent}}
\label{appendix:proof_model_independent}

\propStrongBisim*

\begin{proof}
    The proofs of the first three statements are all standard.
    We only prove the last statement.
    For any distributions $\mu_1, \mu_2\in\Distr(\PRCCS)$ such that $\mu_1\Strong^\dag \mu_2$ and $\mu_1\rightsquigarrow\nu_1$ with witness $\pi$, we prove by induction on $\vert\pi\vert$.
    The case where $\pi=\varepsilon$ is trivial.
    Assume that $\vert\pi\vert=k+1$ and $\pi=(\tau,p)\circ\pi'$.
    We can rewrite the transition sequence as $\mu_1\xrightarrow{(\tau,p)}\mu_1'\myrightarrow{\pi'}\nu_1$. 
    Assume that the first transition is induced by $P\myrightarrow{\tau}\rho_P$ for some $P\in \Supp(\mu_1)$.
    Then $\mu_1'=\mu_1(P)p\rho_P+(1-\mu_1(P)p)\overline{\rho_P}$, where $\overline{\rho_P}:=\frac{\mu_1-\mu_1(P)p\delta_P}{1-\mu_1(P)p}$.
    According to Lemma~\ref{lem:linearity_of_probabilistic_transition_strengthened}, one has that $\rho_P\myrightarrow{\pi_1'}\nu_P$, $\overline{\rho_P}\myrightarrow{\pi_2'}\overline{\nu_P}$ for some $\nu_P$ and $\overline{\nu_P}$ such that $\vert\pi_1'\vert\le k$, $\vert\pi_2'\vert\le k$, and $\nu_1=\mu_1(P)p\nu_P+(1-\mu_1(P)p)\overline{\nu_P}$. 
    Note that $\mu_2([P])=\mu_1([P])\ge \mu_1(P)>0$.
    We decompose $\mu_2$ into 
    \begin{equation}
        \mu_2=\sum_{Q\in[P]}\frac{\mu_1(P)p}{\mu_1([P])}\mu_2(Q)\delta_Q+(1-\mu_1(P)p)\overline{\mu_2},
    \end{equation}
    where $\overline{\mu_2}$ is a distribution.
    Note that for all $Q\in[P]$, $\delta_Q\myrightarrow{\tau}\rho_Q\Strong^\dag\rho_P$ holds by definition. 
    By the induction hypothesis, we have $\rho_Q\rightsquigarrow\nu_Q\Strong^\dag\nu_P$ for some $\nu_Q$.
    Similarly, it can be verified that $\overline{\rho_P}\Strong^\dag\overline{\mu_2}$, which implies $\overline{\mu_2}\rightsquigarrow\overline{\nu_2}\Strong^\dag\overline{\nu_P}$ for some $\overline{\nu_2}$.
    Therefore, by Corollary~\ref{cor:linearity_of_probabilistic_transition}, 
    \begin{equation}
        \begin{aligned}
            \mu_2\rightsquigarrow\nu_2&:=\sum_{Q\in[P]}\frac{\mu_1(P)p}{\mu_1([P])}\mu_2(Q)\nu_Q+(1-\mu_1(P)p)\overline{\nu_2}\\
            &\Strong^\dag\, \sum_{Q\in[P]}\frac{\mu_1(P)p}{\mu_1([P])}\mu_2(Q)\nu_P+(1-\mu_1(P)p)\overline{\nu_P}\\
            &=\nu_1,
        \end{aligned}
    \end{equation}
    where $\Strong^\dag$ holds since the lifting must be closed under the convex combination.
    \qed
\end{proof}

\seqTrans*

\begin{proof}
    For any transition sequence $\mu\rightsquigarrow\nu$, we partition the finite set $\Supp(\nu)$ into $\Supp(\nu)\cap\psi_\L=\left\{A_1,\dots,A_n\right\}$ and $\Supp(\nu)\,\cap\,\overline{\psi_\L}=\left\{B_1,\dots,B_m\right\}$.
    Let $P\in\PRCCS$ be a process, define $P^L:=(L)(P\mid Q)$.
    The transition sequence $\mu\rightsquigarrow\nu$ can induce a corresponding one $\mu^L\rightsquigarrow\nu^L$, where
    $\Supp(\nu^L)=\{A_1^L,\dots,A_n^L\}\uplus\{B_1^L,\dots,B_m^L\}$ and $\nu^L(P^L)=\nu(P)$ for any $P\in\Supp(\nu)$.
    We also have the following observations.
    \begin{enumerate}
        \item For $1\le i\le m$, since $b\notin L$, we have $B_i^L\myrightarrow{b}$, which implies $B_i^L\in\psi_\L$ and $\charMay_\L(B_i^L)=1$.
        \item For $1\le i\le n$, $A_i\in\psi_\L$.
        There exists an external action $l_i\in L\cup \overline{L}$ and a distribution $\rho_i$ such that $A_i\myrightarrow{l_i}\rho_i$.
        Since $Q\myrightarrow{\overline{l_i}}\delta_{\mathbf0}$, we have $\delta_{A_i^L}\myrightarrow{\tau}(L)(\rho_i\mid \delta_{\mathbf 0})$ and $\charMay_\L((L)(\rho_i\mid \delta_{\mathbf 0}))=0$.
    \end{enumerate}
    Note that $\nu^L=\sum_{i=1}^n\nu(A_i)\delta_{A_i^L}+\sum_{i=1}^m\nu(B_i)\delta_{B_i^L}$.
    By Corollary~\ref{cor:linearity_of_probabilistic_transition}~(\ref{item:linear_combination}), we have 
    \begin{equation}
        \mu^L\rightsquigarrow\nu^L\rightsquigarrow\nu':=\sum_{i=1}^n\nu(A_i)(L)(\rho_i\mid \delta_{\mathbf 0})+\sum_{i=1}^m\nu(B_i)\delta_{R_i^L}.
    \end{equation}
    By Lemma~\ref{lem:linearity_of_testing_outcome}, we have 
    \begin{equation}
        \begin{aligned}
            \charFS_\L(\mu^L)
            \le\sum_{i=1}^n\nu(A_i)\charMay_\L((L)(\rho_i\mid \delta_{\mathbf 0}))+\sum_{i=1}^m\nu(B_i)\charMay_\L(B_i^L)
            =\sum_{i=1}^m\nu(B_i)
            =1-\nu(\psi_\L).
        \end{aligned}
    \end{equation}
    Hence, $\charMay_\L(\mu)=\sup\left\{\nu(\psi_\L):\mu\rightsquigarrow\nu\right\}\le 1-\charFS_\L(\mu^L)$.
    
    For the other direction, consider any transition sequence $\mu^L\rightsquigarrow\nu$ witnessed by $\pi$.
    By Corollary~\ref{cor:calculated_by_degenerate_witness} (\ref{item:calculated_by_degenerate_witness_2}), we assume that $\pi\in \tau^{*}$.
    According to the evolution of $Q$, all the processes in the support of any distribution in $\mu^L\rightsquigarrow\nu$ can be divided into two types: 
    \begin{equation}
        \mathcal{X}_1:=\{P^L:P\in\PRCCS\},\quad\text{and}\quad \mathcal{X}_2:=\{(L)(P\mid \mathbf{0}):P\in\PRCCS\}.
    \end{equation}
    It can be calculated similarly that $\sup\PEquip{{P}}=0$ for $P\in\mathcal{X}_2\cap\Supp(\nu)$ and $\sup\PEquip{{P}}=1$ for $P\in\mathcal{X}_1\cap\Supp(\nu)$.
    By Lemma~\ref{lem:linearity_of_testing_outcome}, $\charMay_\L(\nu)=1-\nu(\mathcal{X}_2)$.
    Define a new predicate $\psi_L:=\left\{P^L:P\myrightarrow{\ell}\text{ for some }\ell\in\L\right\}$.
    Note that processes only evolve from type $\mathcal{X}_1$ to type $\mathcal{X}_2$, but not vice versa. 
    Therefore, the probability of distributions satisfying $\mathcal{X}_1$ is non-increasing along the whole transition sequence. 
    Now we construct a new transition sequence starting from $\mu^L$ such that every distribution is defined over $\mathcal{X}_1$. 
    Let $\nu_0:=\nu$ and $\pi_0:=\pi$.
    Repeat the following procedure for $j=0,1,\dots$:
    \begin{enumerate}
        \item Divide $\mu^L\myrightarrow{\pi_j}\nu_j$ into $\mu^L\myrightarrow{\pi_j^1}\mu'\myrightarrow{\tau}\mu''\myrightarrow{\pi_j^2}\nu_j$ such that 
        (i) $\pi_j=\pi_j^1\circ\tau\circ\pi_j^2$,
        (ii) $\mu''(\mathcal{X}_1)=\nu_j(\mathcal{X}_1)$, and (iii) $\mu'(\mathcal{X}_1)>\mu''(\mathcal{X}_1)$. 
        If no such division exists, break from this procedure.
        Assume that $\mu'\myrightarrow{\tau}\mu''$ is caused by $P^L\myrightarrow{\tau}\rho$ for some process $P^L\in\psi_L$ and distribution $\rho$ such that $\rho(\mathcal{X}_2)=1$.
        This indicates that this $\tau$ action arises from the interaction between $P$ and $Q$.
        \item Since $\mu'=(\mu'-\mu'(P^L)\delta_{P^L})+\mu'(P^L)\delta_{P^L}$ and $\mu'\myrightarrow{\tau}\mu''$, we have $\mu''=(\mu'-\mu'(P^L)\delta_{P^L})+\mu'(P^L)\rho$.
        By Corollary~\ref{cor:linearity_of_probabilistic_transition}~(\ref{item:linear_decomposition}), there exists $\mu_1,\mu_2$ such that $\frac{\mu'-\mu'(P^L)\delta_{P^L}}{1-\mu'(P^L)}\rightsquigarrow\mu_1$, $\rho\rightsquigarrow\mu_2$, and $\nu_j=(1-\mu'(P^L))\mu_1+\mu'(P^L)\mu_2$. 
        Replacing $\delta_{P^L}\rightsquigarrow\rho\rightsquigarrow\mu_2$ by $\delta_{P^L}\rightsquigarrow\delta_{P^L}$ and again by this lemma, we obtain a new transition sequence $\mu'\myrightarrow{\pi^3_j} \nu_{j+1}:=(1-\mu'(P^L))\mu_1+\mu'(P^L)\delta_{P^L}$ for some word $\pi_j^3$.
        Define $\pi_{j+1}:=\pi_j^1\circ\pi_j^3$, and we obtain a new transition sequence $\mu^L\myrightarrow{\pi_{j+1}}\nu_{j+1}$.
        \item We have maintained $(\nu_{j+1}-\nu_j)(\psi_L\cup\mathcal{X}_2)=\mu'(P^L)(\delta_{P^L}-\mu_2)(\psi_L\cup\mathcal{X}_2)=\mu'(P^L)(1-1)=0.$
        Besides, by $\mu''(\mathcal{X}_1)=\nu_j(\mathcal{X}_1)$ and monotonicity, we have
        \begin{equation}
            \mu_1(\mathcal{X}_1)=\frac{\mu'-\mu'(P^L)\delta_{P^L}}{1-\mu'(P^L)}(\mathcal{X}_1)=\frac{\mu'(\mathcal{X}_1)-\mu'(P^L)}{1-\mu'(P^L)}. 
        \end{equation}
        Therefore, we can derive that
        $\nu_{j+1}(\mathcal{X}_1)=(1-\mu'(P^L))\mu_1(\mathcal{X}_1)+\mu'(P^L)\delta_{P^L}(\mathcal{X}_1)=\mu'(\mathcal{X}_1).$
        Although $\pi_{j+1}$ may be longer than $\pi_j$, the action used to split $\mu^L\myrightarrow{\pi_{j+1}}\nu_{j+1}$ during next iteration  can only appear within segment $\mu^L\myrightarrow{\pi_j^1}\mu'$.
    \end{enumerate}
    This process can be repeated at most a finite number of times; let us assume it is repeated $k$ times at most.
    Eventually, we will obtain a sequence $\mu^L\myrightarrow{\pi_k}\nu_k$ such that  i) $\nu_k(\mathcal{X}_1)=\mu^L(\mathcal{X}_1)=1$ and ii) 
    \begin{equation}
        \nu_k(\psi_L)=\nu_k(\psi_L\cup\mathcal{X}_2)=\nu_0(\psi_L\cup\mathcal{X}_2)\ge\nu(\mathcal{X}_2)=1-\charMay_\L(\nu).
    \end{equation}
    By monotonicity, $\mu^L\myrightarrow{\pi_k}\nu_k$ takes the form of $\mu^L=:\mu^L_0\myrightarrow{\tau}\mu_1^L\myrightarrow{\tau}\cdots\myrightarrow{\tau}\mu_n^L:=\nu_k$.
    Thus, it induces a transition sequence $\mu=:\mu_0\myrightarrow{\tau}\mu_1\myrightarrow{\tau}\cdots\myrightarrow{\tau}\mu_n$.
    Since $P\in\psi_\L$ if and only if $P^L\in\psi_L$, we have
    $\mu_n(\psi_\L)=\nu_k(\psi_L)\ge 1-\charMay_\L(\nu)$, which implies $1-\charMay_\L(\nu)\le \charMay_\L(\mu)$.
    Finally, taking all $\nu$ such that $\mu^L\rightsquigarrow\nu$ with degenrate witness into consideration, we have $1-\charMay_\L(\mu)\le\inf\left\{\sup \PEquip{\nu}:\mu^L\rightsquigarrow\nu\right\}=\charFS_\L(\mu^L)$.
    Finally, we have $\charMay_\L(\mu)+\charFS_\L(\mu^L)=1$.
    \qed
\end{proof}

\showMem*

\begin{proof}
Proposition~\ref{prop:property_strong_bisimulation} shows that i) $\Strong^\dag$ is $\mathcal{D}$-extensional; ii) if $\mu_1\Strong^\dag\mu_2$ and $\mu_1\rightsquigarrow\nu_1$, then $\mu_2\rightsquigarrow\nu_2\Strong^\dag\nu_1$ for some $\nu_2$ and $\nu_1(\psi_\L)=\nu_2(\psi_L)$.
Hence, $\Strong^\dag$ is probabilistically equipollent, which implies $\Strong^\dag\,\subseteq\,\PXDiamond$.

It is trivial that $\Strong^\dag\R^{\mathcal{D}}\Strong^\dag$ is $\mathcal{D}$-extensional.
Moreover, if $\R^{\mathcal{D}}$ is probabilistically equipollent, then $\Strong^\dag \R^{\mathcal{D}}\Strong^\dag$ is also probabilistically equipollent.
Since $\mu\,\Strong^\dag\,(\mu\mid\mathbf{0})$ holds for any $\mu\in \Distr(\mu)$, the inclusions $\R\,\subseteq\;\Strong^\dag
\R^{\mathcal{D}}\Strong^\dag\;\subseteq\;\PXDiamond$ hold.

As for statement (\ref{item:show_membership_3}), note that $\R$ is $\mathcal D$-extensional.
If $\R$ is probabilistically equipollent, then $\R \subseteq \PXDiamond$ holds.
Let $\vartheta=\delta_X$ and the result follows.
\qed
\end{proof}

\probMayFsExt*

\begin{proof}
We first prove two claims.
\begin{claim}\label{lem:strong_apply_to_test}
    Suppose $\mu_1,\mu_2\in\Distr(\PRCCS^\omega)$ are two distributions such that $\mu_1\Strong^\dag\mu_2$.
    Then 
    \begin{enumerate}
        \item 
        $\charMay_\omega(\mu_1)=\charMay_\omega(\mu_2)$, and
        \item
        $\charFS_\omega(\mu_1)=\charFS_\omega(\mu_2)$.
    \end{enumerate}
\end{claim}
\begin{proof}
    By Proposition~\ref{prop:property_strong_bisimulation}, if $\mu_1\rightsquigarrow\nu_1$ for some $\nu_1$, then $\mu_2\rightsquigarrow\nu_2\Strong^\dag \nu_1$ for some $\nu_2$.
    By definition, $\nu_1(\psi_\omega)=\nu_2(\psi_\omega)$.
    Therefore, $\PTest{\mu_1}\subseteq\PTest{\mu_2}$.
    By symmetry, $\PTest{\mu_1}=\PTest{\mu_2}$ and $\charMay_\omega(\mu_1)=\charMay_\omega(\mu_2)$.
    Now we also have $\charMay_\omega(\nu_1)=\charMay_\omega(\nu_2)$
    Therefore, $\charFS_\omega(\mu_1)=\charFS_\omega(\mu_2)$ holds.
    \qed
\end{proof}

\begin{claim}\label{lem:boundary_commutativity}
    Let $\mu_1,\mu_2\in\PRCCS^\omega$ be two distributions and $L\subset \Chan$ be a finite set of channels. 
    The following equalities hold.
    \begin{enumerate}
        \item 
        $\charMay((L)\mu_1\mid\mu_2)=\charMay(\mu_1\mid(L)\mu_2)$, and
        \item 
        $\charFS((L)\mu_1\mid\mu_2)=\charFS(\mu_1\mid(L)\mu_2)$.
    \end{enumerate}
\end{claim}
\begin{proof}
    Define a special set of processes $\mathcal{X}=\{(L)P_1\mid P_2\}\subseteq\PRCCS^\omega$.
    For $P:= (L)P_1 \mid P_2\in \mathcal{X}$, define $P^*:= P_1 \mid (L)P_2$.
    Moreover, if $P\in\psi_\omega$, then $P^*\in\psi_\omega$ since $\omega\not\in\Chan$.
    For $\mu\in\Distr(\mathcal{X})$, we can also define $\mu^*:=\{(P^*,\mu(P)):P\in\Supp(\mu)\}$ such that $\mu(\psi_\omega)=\mu^*(\psi_\omega)$.
    By induction, we can prove that if $\mu\in\Distr(\mathcal{X})$ and $\mu \xrightarrow{\tau^{*}} \nu$ for some distribution $\nu$, then $\nu\in\Distr(\mathcal{X})$ and $\mu^*\rightsquigarrow\nu^*$.
    By Corollary~\ref{cor:calculated_by_degenerate_witness} we have 
    \begin{equation}
        \charMay_\omega(\mu)
        =\sup\left\{\nu(\psi_\omega):\mu\xrightarrow{\tau^{*}}\nu\right\}\le\sup\left\{\nu^*(\psi_\omega):\mu^*\rightsquigarrow \nu^*\right\}\le\charMay_\omega(\mu^*).
    \end{equation}
    By symmetry, we obtain that $\charMay_\omega(\mu)= \charMay_\omega(\mu^*)$.
    Again by Corollary~\ref{cor:calculated_by_degenerate_witness}, we have 
    \begin{equation}
        \charFS_\omega(\mu)\ge\inf \left\{\sup \PTest{\nu^*}:\mu^*\rightsquigarrow \nu^*\right\}\ge \charFS_\omega(\mu^*).
    \end{equation}
    By symmetry, $\charFS_\omega(\mu)=\charFS_\omega(\mu^*)$.
    Since $\mu_1\mid(L)\mu_2=((L)\mu_1\mid\mu_2)^*$, the both results follow.
    \qed
\end{proof}

    Now back to the proof of Lemma~\ref{lem:probabilistic_may_fs_extensional}.
    Assume that $\mu_1,\mu_2\in \Distr(\PRCCS)$, $\nu\in\mathcal{D}$ and $L\in \Chan$ is finite.
    We can obtain that $(\mu_i\mid \nu)\mid o\Strong^\dag \mu_i\mid (\nu\mid o)$ by applying Proposition~\ref{prop:property_strong_bisimulation}.
    Since $\mathcal{D}$ is closed under the parallel composition and localization operations, both $\nu\mid o$ and $(L)o$ are also contained in $\mathcal{D}^\omega$.
    If $\mu_1\PXMay \mu_2$, by Lemma~\ref{lem:strong_apply_to_test} and Lemma~\ref{lem:boundary_commutativity}, the following equalities hold:
    \begin{enumerate}
        \item 
        $\charMay_\omega((\mu_1\mid\nu)\mid o)
        =\charMay_\omega(\mu_1\mid(\nu\mid o))
        =\charMay_\omega(\mu_2\mid(\nu\mid o))
        =\charMay_\omega((\mu_2\mid\nu)\mid o),$
        \item 
        $\charMay_\omega((L)\mu_1\mid o)
        =\charMay_\omega(\mu_1\mid (L)o)
        =\charMay_\omega(\mu_1\mid (L)o)
        =\charMay_\omega((L)\mu_2\mid o).$
    \end{enumerate}
    Therefore, $(\mu_1\mid\nu)\PXMay (\mu_2\mid\nu)$ and $(L)\mu_1\PXMay(L)\mu_2$, which implies that the equivalence $\PXMay$ is $\mathcal{D}$-extensional.
    Similarly, the equivalence $\PXFS$ is also $\mathcal{D}$-extensional.
    \qed
\end{proof}

\extEqip*

\begin{proof}
    For any $\nu$ such that $\mu\mid o_L\myrightarrow{\pi}\nu$ and $\pi\in\tau^{*}$, the process in $\Supp(\nu)$ can be divided into $\mathcal{X}_1:=\{P\mid O_L:P\in\PRCCS\}$ and $\mathcal{X}_2:=\left\{P\mid \omega:P\in \PRCCS\right\}$.
    Since $\mathcal{X}_2\subseteq\psi_\omega$ and $\mathcal{X}_1\subseteq\overline{\psi_\omega}$, we have $\nu(\psi_\omega)=\nu(\mathcal{X}_2)$.
    Let $\psi_L=\left\{P\mid O_L:P\in\psi_\L\right\}$.
    Note that the probability of satisfying $\mathcal{X}_1$ is non-increasing along the whole transition sequence.
    Repeat the procudure in the proof of Lemma~\ref{lem:sequence_transform}, we can obtain a transition sequence $\mu\mid o_L\rightsquigarrow \nu_k$ in $k$ steps for some $k$, such that i) all distributions are of type $\mathcal{X}_1$, and ii) $\nu_k(\psi_L\cup\mathcal{X}_2)=\nu(\psi_L\cup\mathcal{X}_2)$.
    Since $\nu_k(\mathcal{X}_2)=0$, from ii), we have 
    \begin{equation}
        \nu_k(\psi_L)=\nu_k(\psi_L\cup\mathcal{X}_2)=\nu(\psi_L\cup\mathcal{X}_2)\ge\nu(\mathcal{X}_2)=\nu(\psi_\omega).
    \end{equation}
    Furthermore, $\mu\mid o_L\rightsquigarrow\nu_k$ can induce a corresponding transition sequence $\mu\rightsquigarrow\nu'$ such that $\nu_k=\nu'\mid o_L$.
    Since $P\in\psi_\L\iff P\mid O_L\in\psi_L$, we have
    $\nu'(\psi_\L)=\nu_k(\psi_L)\ge\nu(\psi_\omega)$.
    Therefore, by Corollary~\ref{cor:calculated_by_degenerate_witness}, we have 
    \begin{equation}
        \charMay_\L(\mu)=\sup\left\{\nu'(\psi_\L):\mu\rightsquigarrow\nu'\right\}\ge \sup\left\{\nu(\psi_\omega):\mu\mid o_L\myrightarrow{\tau^{*}}\nu\right\}=\charMay_\omega(\mu\mid o_L).
    \end{equation}
    
    As for the other direction, any transition sequence in the form of $\mu\rightsquigarrow\nu$ can induce a corresponding one $\mu\mid o_L\rightsquigarrow\nu\mid o_L$ such that $(\nu\mid o_L)(\psi_L)=\nu(\psi_\L)$.
    \begin{enumerate}
        \item 
        For any process $P\in\psi_\L\cap\Supp(\nu)$, we have $P\myrightarrow{\ell}\rho_P$ for some $\rho_P\in\Distr(\PRCCS)$ and $\ell\in L\cup\overline{L}$.
        Thus, $\delta_{P\mid O_L}\myrightarrow{\tau}\rho_P\mid\delta_\omega$ and $(\rho_P\mid\delta_\omega)(\psi_\omega)=1$.
        \item 
        If $P\in\overline{\psi_L}\cap\Supp(\nu)$, then $\nu(\psi_\omega)=0$.
    \end{enumerate}
    Since $\nu\mid o_L=\sum_{P\in\Supp(\nu)}\nu(P)\delta_{P\mid O_L}$, by Corollary~\ref{cor:linearity_of_probabilistic_transition}, we have
    \begin{equation}
        \mu\mid o_L\rightsquigarrow\nu\mid o_L\rightsquigarrow \nu':=\sum_{P\in\psi_\L\cap\Supp(\nu)}\nu(P)(\rho_P\mid \delta_\omega)+\sum_{P\in\overline{\psi_\L}\cap\Supp(\nu)}\nu(P)\delta_{P\mid O_L}.
    \end{equation}
    Clearly, $\nu'(\psi_\omega)=\sum_{P\in\psi_\L\cap\Supp(\nu)}\nu(P)=\nu(\psi_\L)$.
    Therefore, 
    \begin{equation}
        \charMay_\omega\left({\mu\mid o_L}\right)=\sup\left\{\nu'(\psi_\omega):\mu\mid o_L\rightsquigarrow \nu'\right\}\ge \sup\left\{\nu(\psi_\L):\mu\rightsquigarrow\nu\right\}=\charMay_\L(\mu),
    \end{equation}
    which implies $\charMay_\L(\mu)=\charMay_\omega\left({\mu\mid o_L}\right)$.
    
    For the second equality, observe that
    \begin{enumerate}
        \item 
        It is trivial that $\sup\PTest{P}=1$ for any $P\in\mathcal{X}_2$.
        \item 
        If $P\myrightarrow{\ell}\rho$, then $\delta_{P\mid O_L}\myrightarrow{\tau}\rho\mid \delta_\omega\in\Distr(\mathcal{X}_2)$.
        Therefore, $\sup\PTest{Q}=1$ for any $Q\in\psi_L$.
    \end{enumerate}
    For any $\nu$ such that $\mu\mid o_L\myrightarrow{\pi}\nu$ and $\pi\in\tau^{*}$, we can obtain a transition sequence $\mu\mid o_L\rightsquigarrow\nu_k$ for some $k$, such that i) all distributions are of type $\mathcal{X}_1$, and ii) $\sup\PTest{\nu_k}=\sup\PTest{\nu}$.
    The invariant ii) can be maintained because processes can only evolve from $\psi_L$ to distributions in $\Distr(\mathcal{X}_2)$, but the linear functional $\mu\mapsto\sup\PTest{\mu}$ is equal to 1 for both.
    Moreover, there exists a transition sequence $\mu\rightsquigarrow\nu'$ such that $\nu_k=\nu'\mid o_L$.
    We have proved that $\charMay_\L({\nu'})=\charMay_\omega({\nu'\mid o_L})$, which implies $\charMay_\omega(\nu)=\charMay_\omega({\nu_k})=\sup\PEquip{\nu'}\ge\inf\left\{\sup\PEquip{\nu'}:\mu\rightsquigarrow\nu'\right\}=\charFS_\L(\mu)$.
    Therefore, by Corollary~\ref{cor:calculated_by_degenerate_witness}, $\charFS_\omega(\mu\mid o_L)\ge\charFS_\L(\mu)$.

    The other direction is relatively straightforward.
    For any transition sequence $\mu\rightsquigarrow\nu$, consider the corresponding sequence $\mu\mid o_L\rightsquigarrow \nu\mid o_L$ induced by it.
    We have proved that 
    \begin{equation}
        \charMay_\L({\nu})=\charMay_\omega({\nu\mid o_L})\ge\charFS_\omega(\mu\mid o_L).
    \end{equation}
    Therefore, $\charMay_\L({\mu})\ge\charFS_\omega(\mu\mid o_L)$, which completes the proof.
    \qed
\end{proof}

\corProbMayFsEqip*

\begin{proof}
Note that the distribution $o_L$ defined in Lemma~\ref{lem:external_equipollent} is in $\Delta(\PCCS^\omega)\subseteq\mathcal{D}^\omega$.
\begin{enumerate}
    \item 
    For any distributions $\mu_1,\mu_2\in\Distr(\PRCCS)$ such that $\mu_1\PXMay\mu_2$, 
    we define $L:=\Chan(\mu_1)\cup\Chan(\mu_2)$.
    Since $L$ must be finite, by Lemma~\ref{lem:external_equipollent}, we have
    \begin{equation}
        \charMay_\L({\mu_1})=\charMay_\omega({\mu_1\mid o_L})=\charMay_\omega(\mu_2\mid o_L)=\charMay_\L({\mu_2}).
    \end{equation}
    Hence, the equivalence $\PXMay$ is probabilistically equipollent.
    \item 
    For any distributions $\mu_1,\mu_2\in\Distr(\PRCCS)$ such that $\mu_1\PXFS\mu_2$, we can prove in a similar manner that 
    \begin{equation}
        \charFS_\L({\mu_1})=\charFS_\omega({\mu_1\mid o_L})=\charFS_\omega(\mu_2\mid o_L)=\charFS_\L({\mu_2}).
    \end{equation}
    Therefore, the equivalence $\PXFS$ is probabilistically strongly equipollent.
\end{enumerate}
By Lemma~\ref{lem:probabilistic_may_fs_extensional}, we have $\PXMay\,\subseteq\,\PXDiamond$ and $\PXFS\,\subseteq\,\PXBox$.
\qed
\end{proof}

\probDiaBoxPred*

\begin{proof}
    For any transition sequence $\mu\rightsquigarrow\nu$, we can proved by induction that $(L)\mu[\omega\mapsto a]\rightsquigarrow(L)\nu[\omega\mapsto a]$.
    For any process $P\in\Supp(\nu)$, $a$ is the only external action $(L)P[\omega\mapsto a]$ can perform.
    Since $P\myrightarrow{\omega}$ if and only if $(L)P[\omega\mapsto a]\myrightarrow{a}$, we have $\nu(\psi_\omega)=(L)\nu[\omega\mapsto a](\psi_\L)\le \charMay_\L((L)\mu[\omega\mapsto a])$.
    Therefore, $\charMay_\omega(\mu)\le\charMay_\L({(L)\mu[\omega\mapsto a]})$.

    Conversely, if $(L)\mu[\omega\mapsto a] \rightsquigarrow \nu'$, it is not difficult to deduce by induction that $\nu'=(L)\nu''$ for some $\nu''\in\Distr(\PRCCS)$ such that $\mu[\omega\mapsto a]\rightsquigarrow \nu''$.
    Note that the transformation $[a\mapsto \omega]$ satisfies the same properties with $[\omega\mapsto a]$ and under the condition that $a\notin\Chan(\mu)$, it is a left inverse of $[\omega\mapsto a]$.
    Hence, we can obtain that $\mu\rightsquigarrow \nu:=\nu''[a\mapsto\omega]$ and 
    \begin{equation}
        \charMay_\omega({\mu})\ge\nu(\psi_\omega)=(L)\nu[\omega\mapsto a](\psi_\L)=\nu'(\psi_\L).
    \end{equation}
    Therefore, $\charMay_\L((L)\mu[\omega\mapsto a])\le \charMay_\omega({\mu})$, which further implies $\charMay_\L((L)\mu[\omega\mapsto a])= \charMay_\omega({\mu})$.

    After establishing the above conclusions, we repeat the previous analysis. 
    Since $\charMay_\L((L)\nu[\omega\mapsto a])= \charMay_\omega({\nu})$, $\charFS_\omega(\mu)=\charFS_\L((L)\mu[\omega\mapsto a])$ follows.
    \qed
\end{proof}

\combChoCongr*

\begin{proof}
For any convex combination operation given by $I$ and $\{p_i\}_I$, let 
\begin{equation}
    \R:=\left\{\left(\sum_{i\in I}p_i\mu_{1,i}, \sum_{i\in I}p_i\mu_{2,i}\right) : \mu_{1,i}\PXDiamond\mu_{2,i}\text{~for all~}i\in I\right\}.
\end{equation}
Note that the convex combination operation is commutative with the localization and composition operations.
Given any pair $(\mu_1, \mu_2)\in \R^\mathcal{D}$, we can assume w.l.o.g. that $\mu_1:=\sum_{i\in I}p_i\nu_{1,i}$ and $\mu_2:=\sum_{i\in I}p_i\nu_{2,i}$, where $(\nu_{1,i},\nu_{2,i}) \in \{(\mu_{1,i},\mu_{2,i})\}^\mathcal{D}$.
Then by Lemma~\ref{lem:linearity_of_testing_outcome},
\begin{equation}
\charMay_\L(\mu_1)
= \sum_{i\in I}p_i\charMay_\L(\nu_{1,i})
= \sum_{i\in I}p_i\charMay_\L(\nu_{2,i})
= \charMay_\L(\mu_2),
\end{equation}
where the second equality holds by $\mu_{1,i}\PXDiamond\mu_{1,i}$ and the $\mathcal{D}$-extensionality of $\PXDiamond$.
By Lemma~\ref{lem:show_membership} (\ref{item:show_membership_2}), $\R\;\subseteq\;\PXDiamond$, which shows that the equality $\PXDiamond$ is preserved by the convex combination operation.

By the extensionality of $\PDiamond$, it is trivial to show that $\Restr{(\PDiamond)}{RCCS}$ is closed under the parallel composition and localization operations.
We abbreviate $\R^{\Distr(\PRCCS)}$ to $\R^\circ$.
\begin{enumerate}
    \item
    Let $\R:=\left\{\left(\delta_{\sum_{i\in I}a_i.P_{1,i}}, \delta_{\sum_{i\in I}a_i.P_{2,i}}\right) : P_{1,i}\PDiamond P_{2,i}\text{~for all~}i\in I\right\}$.
    For any $\mu_1,\mu_2$ such that $\mu_1\;\R^\circ\;\mu_2$, we assume that $\mu_1=\vartheta\left[\sum_{i\in I}a_i.P_{1,i}\right]$ and $\mu_2=\vartheta\left[\sum_{i\in I}a_i.P_{2,i}\right]$, where $\vartheta$ is a $1$-ary distribution containing only one occurrence of $X$.
    Observe that if $\mu_1\myrightarrow{\alpha}\nu_1$, then there exist two $1$-ary distributions $\varsigma_1, \varsigma_2$, $i\in I$ and $p\in [0,1]$ such that $\nu_1=p\varsigma_1\left[\sum_{i\in I}a_i.P_{1,i}\right]+(1-p)\varsigma_2\left[P_{1,i}\right]$ and $\mu_2\myrightarrow{\alpha}p\varsigma_1\left[\sum_{i\in I}a_i.P_{1,i}\right]+(1-p)\varsigma_2\left[P_{2,i}\right]$.
    \begin{claim}
        If $\mu_1\xrightarrow{\pi}\nu_1$ and $\pi\in\tau^{*}$, then $\nu_1(\psi_\L)\le \charMay_\L(\mu_2)$.
    \end{claim}
    \begin{claimproof}
        If $\vert\pi\vert=0$, the result follows trivially.
        Otherwise, assume that $\mu_1\myrightarrow{\tau}\mu_1'\myrightarrow{\pi'}\nu_1$.
        Based on our observation, there exist $\varsigma_1, \varsigma_2$, $i\in I$ and $p\in [0,1]$ such that $\mu_1'=p\varsigma_1\left[\sum_{i\in I}a_i.P_{1,i}\right]+(1-p)\varsigma_2\left[P_{1,i}\right]$ and $\mu_2\myrightarrow{\alpha}\mu_2':=p\varsigma_1\left[\sum_{i\in I}a_i.P_{1,i}\right]+(1-p)\varsigma_2\left[P_{2,i}\right]$.
        By Lemma~\ref{lem:linearity_of_probabilistic_transition_strengthened}, there exists two distributions $\nu_1',\overline{\nu_1}'$ such that
        $\varsigma_1\left[\sum_{i\in I}a_i.P_{1,i}\right]\myrightarrow{\pi_1}\nu_1'$, $\varsigma_2\left[P_{1,i}\right]\myrightarrow{\pi_2}\overline{\nu_1}'$, $\nu_1=p\nu_1'+(1-p)\overline{\nu_1}'$ and $\vert\pi_1\vert, \vert\pi_2\vert\le \vert\pi'\vert$.
        \begin{enumerate}
            \item Note that $\varsigma_1\left[\sum_{i\in I}a_i.P_{1,i}\right]\;\R^\circ\; \varsigma_1\left[\sum_{i\in I}a_i.P_{2,i}\right]$.
            By the induction hypothesis, one has that $\nu_1'(\psi_\L)\le \sup\PEquip{\varsigma_1\left[\sum_{i\in I}a_i.P_{2,i}\right]}$.
            \item 
            Since $\varsigma_2$ consists of only composition and localization operations, by the extensionality, we have $\varsigma_2\left[P_{1,i}\right]\PDiamond\varsigma_2\left[P_{2,i}\right]$, which implies that
            $\overline{\nu_1}'(\psi_\L)\le \sup\PEquip{\varsigma_2[P_{1,i}]}=\sup\PEquip{\varsigma_2[P_{2,i}]}$.
        \end{enumerate}
        By Lemma~\ref{lem:linearity_of_testing_outcome}, we have $\nu_1(\psi_\L)=p\nu_1'(\psi_\L)+(1-p)\overline{\nu_1}'(\psi_\L)\le  \sup\PEquip{\mu_2'}\le \charMay_\L(\mu_2).$
        \qed
    \end{claimproof}
    By symmetry and Corollary~\ref{cor:calculated_by_degenerate_witness}, one has that $\R^\circ$ is probabilistically equipollent.
    Therefore, $\R\subseteq\PDiamond$, which implies $\Restr{(\PDiamond)}{RCCS}$ is closed under the nondeterministic choice operation.
    \item 
    Let $\R:=\left\{\left(\delta_{\bigoplus_{i\in I}p_i.P_i}, \delta_{\bigoplus_{i\in I}p_i.Q_i}\right)\right\}$.
    Note that $\delta_{\bigoplus_{i\in I}p_i.P_i}\myrightarrow{\tau}\sum_{i\in I}p_i\delta_{P_i}$ and we have proved that $\PDiamond$ is preserved by convex combinations.
    Similarly, we can prove that $\R\subseteq\PDiamond$.
    Therefore, $\Restr{(\PDiamond)}{RCCS}$ is closed under the probabilistic choice operation.
\end{enumerate}

In conclusion, the equality $\Restr{(\PDiamond)}{RCCS}$ is preserved by the nondeterministic choice, probabilistic choice, parallel composition, and localization operations. 
\qed
\end{proof}

\addTau*

\begin{proof}
    We first prove the following claims.
    \begin{claim} Assume that $\vartheta$ is a $1$-ary distribution.
        \begin{enumerate}
        \item \label{item:algebraic_property_of_R1}
        If $\vartheta[\fix X.\tau.S]\myrightarrow{\alpha}\nu$, then there exists two $1$-ary distributions $\varsigma_1,\varsigma_2$ such that $\varsigma_1\PDiamond\varsigma_2$, $\nu=\varsigma_1[\fix X.\tau.S]$ and $\vartheta[\fix X.\tau.T]\myrightarrow{\alpha}\varsigma_2[\fix X.\tau.T]$.
        \item \label{item:algebraic_property_of_R2}
        $\vartheta[\fix X.\tau.S](\psi_\L)=\vartheta[\fix X.\tau.T](\psi_\L)$.
    \end{enumerate}
    \end{claim}
    \begin{claimproof}
        Assume that $\vartheta[\fix X.\tau.S]\myrightarrow{\alpha}\nu$ is induced by $R\in \Supp(\vartheta)$ such that $R[\fix X.\tau.S]\myrightarrow{\alpha}\rho$ and $\nu=\vartheta[\fix X.\tau.S]+\vartheta(R)(\rho-\delta_R[\fix X.\tau.S])$.
        It suffices to prove that there exists two $1$-ary distributions $\varrho_1,\varrho_2$ such that $\varrho_1\PDiamond\varrho_2$, $\rho=\varrho_1[\fix X.\tau.S]$ and $R[\fix X.\tau.T]\myrightarrow{\alpha}\varrho_2[\fix X.\tau.T]$, which is standard by induction on the inference of $R[\fix X.\tau.S]\myrightarrow{\alpha}\rho$.
        Let $\varsigma_i:=\vartheta+\vartheta(R)(\varrho_i-\delta_R)$.
        Since the equality $\PDiamond$ is closed under the convex combination operation, one has that $\varsigma_1\PDiamond\varsigma_2$,  $\nu=\varsigma_1[\fix X.\tau.S]$ and $\vartheta[\fix X.\tau.T]\myrightarrow{\alpha}\vartheta[\fix X.\tau.T]+\vartheta(R)(\varrho_2[\fix X.\tau.T]-\delta_R[\fix X.\tau.T])=\varsigma_2[\fix X.\tau.T]$.
        
        By statement (\ref{item:algebraic_property_of_R1}), one can conclude that for any $\ell\in\L$,  $R[\fix X.\tau.S]\myrightarrow{\ell}$ if and only if $R[\fix X.\tau.T]\myrightarrow{\ell}$.
        Therefore,
        \begin{equation}
            \vartheta[\fix X.\tau.S](\psi_\L)=\sum_{R[\fix X.\tau.S]\myrightarrow{\ell}}\vartheta(R)=\sum_{R[\fix X.\tau.T]\myrightarrow{\ell}}\vartheta(R)=\vartheta[\fix X.\tau.T](\psi_\L).
        \end{equation}
        \qed
    \end{claimproof}

    Let $\R := \left\{(\vartheta[\fix X.\tau.S], \vartheta[\fix X.\tau.T]) : S \PDiamond T, \text{~and~}\vartheta\text{ is $1$-ary}\right\}$.
    For any $1$-ary distribution $\vartheta$ such that $\vartheta[\fix X.\tau.S]\myrightarrow{\tau^k}\nu_1$, we prove by induction that $\nu_1(\psi_\L)\le \charMay_\L({\vartheta[\mu X.\tau.T]})$.
    \begin{enumerate}
        \item 
        $k=0$.
        Then $\nu_1=\vartheta[\fix X.\tau.S]$.
        By statement (\ref{item:algebraic_property_of_R2}), one has that $\vartheta[\fix X.\tau.S](\psi_\L)=\vartheta[\fix X.\tau.T](\psi_\L)$.
        Hence, $\nu_1(\psi_L)\le \sup\PEquip{\vartheta[\fix X.\tau.S]}=\charMay_\L({\vartheta[\mu X.\tau.T]})$.
        \item 
        $k=i+1$.
        By statement (\ref{item:algebraic_property_of_R1}), we can w.l.o.g. assume that $\vartheta[\fix X.\tau.S]\myrightarrow{\tau}\varsigma_1[\fix X.\tau.S]\myrightarrow{\tau^i}\nu_1$ and $\vartheta[\fix X.\tau.T]\myrightarrow{\tau}\varsigma_2[\fix X.\tau.T]$ for some $1$-ary distributions $\varsigma_1, \varsigma_2$ such that $\varsigma_1\PDiamond\varsigma_2$.
        Applying the induction hypothesis to $\varsigma_1[\fix X.\tau.S]\myrightarrow{\tau^i}\nu_1$, we can obtain that $\nu_1(\psi_\L)\le\sup\PEquip{\varsigma_1[\mu X.\tau.T]}$.
        Since $\varsigma_1\PDiamond \varsigma_2$, we have $\varsigma_1[\fix X.\tau.T]\PDiamond\varsigma_2[\fix X.\tau.T]$.
        Then
        \begin{equation}
            \nu_1(\psi_\L)\le \sup\PEquip{\varsigma_1[\mu X.\tau.T]}=\sup\PEquip{\varsigma_2[\mu X.\tau.T]}\le \sup\PEquip{\vartheta[\mu X.\tau.T]}=\charMay_\L({\vartheta[\mu X.\tau.T]}).
        \end{equation}
    \end{enumerate}
    By symmetry and Corollary~\ref{cor:calculated_by_degenerate_witness}, one has $\charMay_\L({\vartheta[\mu X.\tau.S]})=\charMay_\L({\vartheta[\mu X.\tau.T]})$.
    By Lemma~\ref{lem:show_membership}, we conclude that $\fix X.\tau.S\PDiamond\fix X.\tau.T$.
    
    Let $\R:=\{(\vartheta[\fix X.T], \vartheta[\mu X.\tau.T]) : \mu \text{ is $1$-ary}\}$.
    It suffices to prove that $\R$ is probabilistically equipollent.
    By induction on the inference, the following statements are valid.
    \begin{enumerate}
        \item If $\vartheta[\fix X.T]\myrightarrow{\alpha}\nu$, then there exists a $1$-ary distribution $\varsigma$ such that $\nu=\varsigma[\mu X.T]$ and $\vartheta[\fix X.\tau.T]\rightsquigarrow\myrightarrow{\alpha}\varsigma[\fix X.\tau.T]$.
        \item If $\vartheta[\fix X.\tau.T]\myrightarrow{\alpha}\nu$, then there exists a $1$-ary distribution $\varsigma$ such that $\nu=\varsigma[\fix X.\tau.T]$ and $\vartheta[\fix X.T]\myrightarrow{\alpha}\varsigma[\fix X.T]$ or $\vartheta[\fix X.T]\Strong^\dag \varsigma[\fix X.T]$.
        \item $\vartheta[\fix X.\tau.T](\psi_\L)\le \vartheta[\fix X.T](\psi_\L)\le \charMay_\L({\vartheta[\fix X.\tau.T]})$.
    \end{enumerate}
    Therefore, for any transition sequence $\vartheta[\fix X.T]\myrightarrow{\tau^k}\nu_1$, we can prove by induction on $k$ that $\nu_1(\psi_\L)\le\charMay_\L({\vartheta[\fix X.\tau.T]})$.
    Hence, by Corollary~\ref{cor:calculated_by_degenerate_witness}, $\charMay_\L({\vartheta[\fix X.T]})\le \charMay_\L({\vartheta[\fix X.\tau.T]})$.
    The analog statement in the symmetric case also holds.
    By Lemma~\ref{lem:show_membership}, we have $\fix X.T\PDiamond \fix X.\tau.T$.
    \qed
\end{proof}

\fixPoint*

\begin{proof}
Consider any $n$-ary terms $S, T$ such that $S\PDiamond T$.
For any processes $P_2,\dots, P_n\in\PRCCS$, it is trivial that $S[X,P_2,\dots, P_n]\PDiamond T[X,P_2,\dots, P_n]$.
By the conclusion of case $n=1$, we have $\mu X.S[X,P_2,\dots,P_n]\PDiamond\mu X.T[X,P_2,\dots,P_n]$, which implies $\fix X.S\PDiamond\fix X.T$.
\qed
\end{proof}

\pcspOut*

\begin{proof}
    It suffices to consider an arbitrary process $s \in \mathsf{sCSP}$.
    We can first prove by induction that
    \begin{enumerate}
        \item\label{item:pcsp_outcome_1} if $\delta_s \myrightarrow{\tau^*} \nu$, then $\nu(\psi_\omega)\le \max\mathbb{V}(\nu) \in \mathbb{V}(s)$, and
        \item if $p \in\mathbb{V}(s)$, then there exists a distribution $\nu$ such that $\delta_s\rightsquigarrow \nu$ and $\mathbb{V}(\nu)=\{p\}$ and $\nu(\psi_\omega)=p$.
    \end{enumerate}
    By definition, this yields $\max \mathbb{V}(s)=\charMay_\omega(s)$, and therefore the first statement follows.
    For the second statement, since$\mathbb{V}(s)$ is a finite set, we can take $\delta_s\rightsquigarrow \nu$ with $\mathbb{V}(\nu)=\{\min \mathbb{V}(s)\}$.
    Since $\charMay_\omega(\nu)=\max \mathbb{V}(\nu)=\min \mathbb{V}(s)$, we have  $\min \mathbb{V}(s)\ge \charFS_\omega(s)$.
    Moreover, consider any transition sequence $\delta_s \myrightarrow{\tau^*} \nu$.
    By Property~(\ref{item:pcsp_outcome_1}) and the first statement above, we have $\charMay_\omega(\nu)=\max\mathbb{V}(\nu)\ge \min \mathbb{V}(s)$.
    hen Corollary~\ref{cor:calculated_by_degenerate_witness} implies $\min \mathbb{V}(s) \le \charFS_\omega(s)$. 
    This proves the second statement.
\end{proof}

\section{Proofs for Section \ref{sec:comparison}}
\label{appendix:proof_comparison}

\equipClasGen*

\begin{proof}
    Given a CCS process $P$, we use $\mathsf{Reach}_{\tau}(P) = \left\{P': P \Longrightarrow P'\right\}$ to denote the set of processes reachable from $P$ through internal actions.

    For $\charMay_\L(P)$, there are two possible cases according to whether $P \Downarrow$.
    \begin{enumerate}
        \item If $P \Downarrow$, then there exists $P' \in \psi_\L$ with $P \Longrightarrow P'$. Now we see that $\delta_{P'}(\psi_\L) = 1$ for some transition sequence $\delta_{P} \rightsquigarrow \delta_{P'}$, and thus $
        \charMay_\L(P) = 1$.
        \item If $P \not\Downarrow$, then $P' \notin \psi_\L$ for all $P'$ with $P \Longrightarrow P'$, which implies that $\mathsf{Reach}_{\tau}(P)\subseteq \overline{\psi_\L}$. Consider any transition sequence $\delta_{P} \rightsquigarrow \nu$ with witness $\pi$, we can prove that $\Supp(\nu) \subseteq \mathsf{Reach}_{\tau}(P)$ by induction on the length of $\pi$. 
        Then $\nu(\psi_\L) = \sum_{P' \in \psi_\L} \nu(P') = 0$.
        Therefore, $\charMay_\L(P)=0$.
    \end{enumerate}
    \indent We then prove the equality for $\charFS_\L(P)$. 
    There are two possible cases according to whether $P \observable$.
    \begin{enumerate}
        \item If $P \observable$, then $P' \Downarrow$ for all $P'$ with $P \Longrightarrow P'$. 
        According to the above conclusion, 
        $\charMay_\L({P'})=1$ for all $P'$ for all $P'\in\mathsf{Reach}_{\tau}(P)$.
        Consider any transition sequence $\delta_{P} \rightsquigarrow \nu$.
        According to Lemma \ref{lem:linearity_of_testing_outcome}, we have 
        \begin{equation}
            \charMay_\L(\nu)= \sum_{P'\in\Supp(\nu)}\nu(P')\charMay_\L({P'})
        = \sum_{P'\in\Supp(\nu)}\nu(P')\cdot 1 = 1.
        \end{equation}
        By the arbitrariness of the transition sequences, we have $\charFS_\L(\mu)=
        \inf \left\{\charMay_\L(\nu): \delta_{P} \rightsquigarrow \nu\right\} = 1$.
        \item If $P \not\observable$, then $P' \not\Downarrow$ for some $P'$ with $P \Longrightarrow P'$. 
        By our first equality, we have
        $\charMay_\L({P'})=0$.
        Now we see that $\charMay_\L({P'})=0$ for the transition sequence $\delta_{P}\rightsquigarrow\delta_{P'}$.
        Thus, $\charFS_\L(\mu)=0$.
        \qed
    \end{enumerate}
\end{proof}

\propClasGen*

\begin{proof}
    Define the $\Delta$-lifting of a relation $\R$ over $\PCCS$ by $\R^{\Delta}:=\{(\delta_P,\delta_Q) : P~\R~Q\}$.
    Then $\R^\Delta$ can be viewed as a relation over $\Distr(\PRCCS)$.
    By definition, we have the following observations.
    \begin{enumerate}
        \item \label{item:delta_extensionality_1} 
        If $\R$ is a $\Delta$-extensional relation on $\Distr(\PRCCS)$, then $\Restr{(\Strong^\dag\R\Strong^\dag)}{CCS}$ is extensional on $\PCCS$. 
        \item \label{item:delta_extensionality_2} 
        If $\E$ is an extensional equivalence on $\PCCS$ then $\Strong^\dag\E^{\Delta}\Strong^\dag$ is $\Delta$-extensional on $\Distr(\PRCCS)$.
    \end{enumerate}
    By Lemma~\ref{lem:equipollence_classical_generalization}, we can deduce that $\Restr{(\PCDiamond)}{CCS}$ is an extensional, equipollent relation on $\PCCS$, which is contained in $\CDiamond$.
    Again by this lemma, we obtain that $\Strong^\dag(\CDiamond)^{\Delta}\Strong^\dag$ is $\Delta$-extensional and probabilistic equipollent and thus contained in $\PCDiamond$.
    By definition, we have $(\CDiamond)\subseteq\Restr{((\Strong^\dag(\CDiamond)^{\Delta}\Strong^\dag)}{CCS}\,\subseteq\,\Restr{(\PCDiamond)}{CCS}.$
    Therefore, $\Restr{(\PCDiamond)}{CCS}=(\CDiamond)$.
    Since $p$-extensionality implies $\Delta$-extensionality, one has that $\Restr{(\PDiamond)}{CCS} \,\subseteq \,\Restr{(\PCDiamond)}{CCS}$.
    Similarly, $\Restr{(\PBox)}{CCS} \,\subseteq \,\Restr{(\PCBox)}{CCS}=(\CBox)$ also holds.
    The strictness is witnessed by Example~\ref{ex:extensionality_counterexp}.
    \qed
\end{proof}

\propTestClassGen*

\begin{proof}
    By Corollary~\ref{cor:probabilistic_may_fs} and Proposition~\ref{prop:classical_generalization}, we can easily prove the two strict containment relations.
    Next, we will show that $\May$ and $\Restr{(\PCMay)}{CCS}$ are coincident.
    \begin{enumerate}
        \item 
        Given any $P,Q\in\PCCS$ such that $P\May Q$, for any $o\in\Delta(\PCCS)^{\omega}$, it is clear that $o=\delta_O$ for some $O\in\PCCS^{\omega}$.
        By the definition of $\May$ and Lemma~\ref{lem:testing_classical_generalization}, we have $\charMay_\omega({\delta_P\mid o})=\charMay_\omega({P\mid O})=\charMay_\omega({Q\mid O})=\charMay_\omega({\delta_Q\mid o})$.
        Therefore, $\delta_P\PCMay\delta_Q$ holds, which implies $\May\subseteq\Restr{(\PCMay)}{CCS}$.
        \item 
        Given any $P,Q\in\PCCS$ such that $P~\Restr{(\PCMay)}{CCS}~Q$ and any classical observer $O\in\PCCS^{\omega}$, the equality $\charMay_\omega({P\mid O})=\charMay_\omega({Q\mid O})$ holds by definition.
        By Lemma~\ref{lem:testing_classical_generalization}, $P\May Q$ and thus $\Restr{(\PCMay)}{CCS}\subseteq\May$.
    \end{enumerate}
    Following a similar argument, we can prove that $\Restr{(\PCFS)}{CCS}=(\FS)$.
    \qed
\end{proof}

\pweakAlmostBisim*

\begin{proof}
	Assume that $\mu_1 \ptran \nu_1$ is witnessed by $\pi_1 \in\tau^{*}$. 
    We prove this lemma by induction on $|\pi_1|$.
	\begin{enumerate}
		\item In the base case $|\pi_1| = 0$, we have $\nu_1 = \mu_1$. Since $\mu_1 \PWeak \mu_2$ , $\mu_1(\mathcal{C}) = \mu_2(\mathcal{C})$ holds for all $\mathcal{C} \in \PRCCS / \PWeak$ (see Definition \ref{def:lifting}), which implies that $|\mu_2 - \mu_1|_{\PWeak} = 0$.
		Now by setting $\nu_2 = \mu_2$, we see that $\mu_2 \ptran \nu_2$ and $|\nu_2 - \nu_1|_{\PWeak} = |\mu_2 - \mu_1|_{\PWeak} = 0$, thus the results follows.
		\item Assume that the statements holds for $|\pi_1| = k \ge 0$.
		When $|\pi_1| = k+1$, we can divide the transition sequence into $\mu_1 \myrightarrow{\pi_1'} \nu_1' \myrightarrow{\tau} \nu_1$, where $\pi_1 = \pi_1' \circ \tau$ and $\pi_1'$ is a degenerate witness.
		Since $|\pi_1'| = k$ and $\pi_1'$ is a degenerate witness, by the induction hypothesis, for any number $\epsilon > 0$, there exists $\nu_2'$ such that $\mu_2 \ptran \nu_2'$ and $|\nu_2' - \nu_1'|_{\PWeak} \le \epsilon$.
		Since $\nu_1' \myrightarrow{\tau} \nu_1$, there exists a process $P \in \Supp(\nu_1')$ such that $P \myrightarrow{\tau} \rho_P$ and $\nu_1 = \nu_1' + \nu_1'(P)(\rho_P - \delta_{P})$.
		Since $|\nu_2' - \nu_1'|_{\PWeak} \le \epsilon$, we have $|\nu_2'([P]) - \nu_1'([P])| \le \epsilon$.
		Thus there exists a set of processes $\{Q_1, \cdots, Q_m\} \subseteq \Supp(\nu_2')$ and a number $p \in (0,1]$ such that $Q_i \in [P]$ for all $i \in [m]$ and $|p\cdot(\sum_{i=1}^{m} \nu_2'(Q_i)) - \nu_1'(P)| \le \epsilon$.
		
	Now consider any process $Q \in \{Q_1, \cdots, Q_m\}$. 
        Since $P \PWeak Q$ and $P \myrightarrow{\tau} \rho_P$, there exists $\rho_{Q}\in\InfDistr(\PRCCS)$ such that $Q \stackrel{\tau}{\Longrightarrow}_{c} \rho_Q$ and $\rho_P \PWeak \rho_Q$. 
    Suppose the weak combined transition $Q \stackrel{\tau}{\Longrightarrow}_{c} \rho_Q$ is induced by  a scheduler $\sigma_Q$ defined on $\mathsf{frags}^{*}(Q)$, we define a set of schedulers $\{\sigma_Q^{\upharpoonright n}\}_{n \in \mathsf{N}}$ as follows:
    \begin{equation}
        \sigma_Q^{\upharpoonright n}(\omega)= \begin{cases}
        \sigma_Q(\omega), & \text{if $|\omega| < n$,} \\
        \delta_{\bot}, & \text{if $|\omega| \ge n$.}
    \end{cases}
    \end{equation}
	Now each scheduler $\sigma_Q^{\upharpoonright n}$ will induce a new weak combined transition $Q \stackrel{\tau}{\Longrightarrow}_{c} \rho_Q^{\upharpoonright n}$ for some distribution $\rho_Q^{\upharpoonright n}\in\Distr(\PRCCS)$.
		By definition, for each $Q' \in \PRCCS$, we have $\lim_{n\to\infty}\rho_Q^{\upharpoonright n}(Q') = \rho_{Q}(Q') $.
		Thus for any $\epsilon > 0$, there exists $N$ such that the distribution $\rho_Q^{\upharpoonright N}$ (induced by the scheduler $\sigma_Q^{\upharpoonright N}$) satisfying that 
		$|\rho_Q^{\upharpoonright N} - \rho_{Q}|_{\PWeak} \le \epsilon$.
		We further show how to construct a probabilistic transition sequence
		$\delta_Q \ptran \rho_Q^{\upharpoonright N}$ according to the scheduler $\sigma_Q^{\upharpoonright N}$.
		Consider any fragment $\omega$ with $|\omega| < N$, w.l.o.g, assume that $\mathsf{last}(\omega) = B$ and $B$ can perform $k$ different internal transitions $\{\mathsf{tr}_{i} = B \myrightarrow{\tau} \rho_i \mid i \in [k]\}$. 
		Suppose $\sigma_Q^{\upharpoonright N}(\omega)(\mathsf{tr}_{i}) = q_i$ for all $i \in [k]$ and $\sigma_Q^{\upharpoonright N}(\omega)(\bot) = 1 - \sum_{i \in [k]}q_i$, then $B$ will arrive at the distribution $\mu_B = \sum_{i \in [k]}q_i \rho_i + (1 - \sum_{i \in [k]}q_i) \delta_B$ in next step according to the scheduler $\sigma_Q^{\upharpoonright N}$.
        Since $\delta_B\ptran\rho_i$ for all $i\in[k]$,
        by Corollary~\ref{cor:linearity_of_probabilistic_transition}~(\ref{item:linear_combination}), 
        \begin{equation}
            \delta_B=\sum_{i \in [k]}q_i \delta_B + (1 - \sum_{i \in [k]}q_i) \delta_B\ptran \sum_{i \in [k]}q_i \rho_i + (1 - \sum_{i \in [k]}q_i) \delta_B=\mu_B.
        \end{equation}
        Repeat the above procedure and by induction on the length of $\omega$, we can prove that $\delta_Q \ptran \rho_Q^{\upharpoonright N}$.
		
		Now we have proved that for each $Q_i \in \{Q_1, \cdots, Q_m\}$ and any $\epsilon > 0$, there exists a distribution $\rho_{Q_i}$ such that $\delta_{Q_i} \ptran \rho_{Q_i}$ and $|\rho_P - \rho_{Q_i}| \le \epsilon$.
		Since $\nu_2' = \sum_{i \in [m]}\nu_2'(Q_i) \delta_{Q_i} + (1-r) \frac{\nu_2' - \sum_{i \in [m]}\nu_2'(Q_i) \delta_{Q_i}}{1-r}$ (where $r = 1 - \sum_{i \in [m]}\nu_2'(Q_i)$) and $\delta_{Q_i} \ptran (1-p)\delta_{Q_i} + p \rho_{Q_i}$, according to Corollary \ref{cor:linearity_of_probabilistic_transition}, we have
		$
			\nu_2'\ptran\nu_2' + p\sum_{i \in [m]}\nu_2'(Q_i) (\rho_{Q_i} - \delta_{Q_i}). 
		$
		Let $\nu_{2} = \nu_2' + p\sum_{i \in [m]}\nu_2'(Q_i) (\rho_{Q_i} - \delta_{Q_i})$, then we have 
            \begin{equation}
                \begin{aligned}
			&\vert\nu_2 - \nu_1|_{\PWeak} \\
			 =~  &\vert[\nu_2' + p\sum_{i \in [m]}\nu_2'(Q_i) (\rho_{Q_i} - \delta_{Q_i})] - [\nu_1' + \nu_1'(P)(\rho_P - \delta_{P})]\vert_{\PWeak} \\
			\le~ & |\nu_2' - \nu_1'|_{\PWeak} 
			+
			|p\sum_{i \in [m]} \nu_2'(Q_i) - \nu_1'(P)|\cdot |\rho_P - \delta_{P}|_{\PWeak}
			+
			p\sum_{i \in [m]}\nu_2'(Q_i) \cdot|(\rho_{Q_i} - \delta_{Q_i}) - (\rho_P - \delta_{P})|_{\PWeak} \\
			\le~ &\epsilon + \epsilon \cdot 2 +1\cdot (\epsilon + 0)\\
			=~ &4\epsilon
		\end{aligned}
            \end{equation}
	\qed
	\end{enumerate}
\end{proof}

\subsection{Proof of Lemma~\ref{lem:PWeak_is_stronglyEquipollent}}
\label{app:proof_strong_equipollence}

Given an induced probability distribution $\rho_{\sigma,A}$ of process $A$ under scheduler $\sigma$ and an external action $\ell \in \L$, we say an execution fragment $\omega\in \mathsf{frags}^{*}(A)$ is \emph{$\ell$-observable} if $\rho_{\sigma,A}(\omega)>0$ and $\mathsf{trace}(\omega) = \ell$, and say $\omega$ is \emph{$\ell$-fireable} if $\rho_{\sigma,A}(\omega)>0$, $\mathsf{trace}(\omega) = \epsilon$, and $\sigma(\omega)(\mathsf{tr}) > 0$ for some $\mathsf{tr}$ with $act(\mathsf{tr})  = \ell$. The $\ell$-observable probability of $\rho_{\sigma,A}$ is then defined by 
\begin{equation}
    \mathsf{P}^{\ell}(\rho_{\sigma,A}) :=
\sum\{
\rho_{\sigma,A}(\omega)
\mid \text{$\omega$ is an $\ell$-observable execution fragment}
\}.
\end{equation}

\begin{lemma}
\label{lem:PWeak_observable_sup}
	Let $P,Q \in \PRCCS$ be two processes with $P \PWeak Q$. 
	If $P \in \psi_\L$, then $\charMay_\L({Q})= 1$.
\end{lemma}

\begin{proof}
	Assume that $P,Q \in \PRCCS$ are two processes with $P \PWeak Q$.
	If $P \in \psi_\L$, then $P \myrightarrow{\ell} \rho_P$ for some external action $\ell \in \L$. Since $P \myrightarrow{\ell} \rho_P$ and $P \PWeak Q$, there exists distribution $\rho_{Q}\in\InfDistr(\PRCCS)$ such that $Q \stackrel{\ell}{\Longrightarrow}_{c} \rho_Q$ and $\rho_P \PWeak \rho_Q$.
	Suppose the weak combined transition $Q \stackrel{\ell}{\Longrightarrow}_{c} \rho_Q$ is induced by  a scheduler $\sigma_Q$ defined on $\mathsf{frags}^{*}(Q)$, we define a family of schedulers $\{\sigma_Q^{\upharpoonright n}\}_{n \in \mathsf{N}}$ as follows:
	\begin{equation}
	    \sigma_Q^{\upharpoonright n}(\omega)= \begin{cases}
		\sigma_Q(\omega), & \text{if $|\omega| < n$,} \\
		\delta_{\bot}, & \text{if $|\omega| \ge n$.}
	\end{cases}
	\end{equation}
	Now each scheduler $\sigma_Q^{\upharpoonright n}$ will induce a probability distribution $\rho_{\sigma_Q^{\upharpoonright n},Q}$ over finite execution fragments.
	According to Definition \ref{def:weak_combined_transition}, we have 
	\begin{equation}
	    \lim_{n\to\infty} \mathsf{P}^{\ell}(\rho_{\sigma_Q^{\upharpoonright n},Q})
	~=~ \rho_{\sigma_Q,Q}\{\omega\in \mathsf{frags}^{*}(Q)\mid \mathsf{trace}(\omega)=\ell\}
	~=~ 1.
	\end{equation}
	For each finite scheduler $\sigma_Q^{\upharpoonright n}$, by cutting the $\ell$-fireable prefix for all $\ell$-observable execution fragments, we obtain a new finite scheduler $(\sigma')_Q^{\upharpoonright n}$ defined as follows:
	\begin{equation}
	    (\sigma')_Q^{\upharpoonright n}(\omega)= \begin{cases}
		\sigma_Q^{\upharpoonright n}(\omega), & \text{if $|\omega| < n$, $\mathsf{trace}(\omega) = \epsilon$, and $\omega$ is not $\ell$-fireable,} \\
		\delta_{\bot}, & \text{otherwise.}
	\end{cases}
	\end{equation}
	According to the induced probability distribution $\rho_{(\sigma')_Q^{\upharpoonright n},Q}$ over finite execution fragments, we can define a distribution $\rho_Q^{\upharpoonright n}$ over $\PRCCS$ as follows:
	\begin{equation}
	    \rho_Q^{\upharpoonright n}(Q') :=
	\rho_{(\sigma')_Q^{\upharpoonright n},Q}\{\omega\in \mathsf{frags}^{*}(Q)\mid \mathsf{last}(\omega)=Q'\} \text{ for all process $Q' \in \PRCCS$.}
	\end{equation}
	Since each $\ell$-observable execution fragment $\omega$ has an $\ell$-fireable prefix $\omega'$, we see that 
	\begin{align}
		\mathsf{P}^{\ell}(\rho_{\sigma_Q^{\upharpoonright n},Q}) 
		\le
		\rho_{(\sigma')_Q^{\upharpoonright n},Q}\{\omega\in \mathsf{frags}^{*}(Q)\mid \mathsf{last}(\omega) \in \psi_\L\}=\sum\{
		\rho_Q^{\upharpoonright n}(Q')
		\mid \text{$Q' \in \PRCCS\cap \psi_\L$} \} 
		= \rho_Q^{\upharpoonright n}(\psi_\L). 
	\end{align}
	Thus 
	\begin{equation}
	    \lim_{n\to\infty} \rho_Q^{\upharpoonright n}(\psi_\L)
	~=~ \lim_{n\to\infty} \mathsf{P}^{\ell}(\rho_{\sigma_Q^{\upharpoonright n},Q}) 
	~=~ 1.
	\end{equation}
	Then for each $n \in \mathsf{N}$, by a similar argument as in the proof of Lemma \ref{lem:PWeak_almost_bisimulation}, we can construct a probabilistic transition sequence $\delta_Q \ptran \rho_Q^{\upharpoonright n}$ according to the finite scheduler $(\sigma')_Q^{\upharpoonright n}$.
	Since $\delta_Q \ptran \rho_Q^{\upharpoonright n}$ for all $n \in \mathsf{N}$ and $\lim_{n\to\infty} \rho_Q^{\upharpoonright n}(\psi_\L) = 1$, we conclude that $\charMay_\L({Q})= 1$.
    \qed
\end{proof}

\begin{lemma}
\label{lem:PWeak_is_equipollent}
    $\PWeak$ is a probabilistically equipollent relation on $\Distr(\PRCCS)$.
\end{lemma}

\begin{proof}
Let $\mu_1,\mu_2 \in \Distr(\PRCCS)$ be two distributions such that $\mu_1 \PWeak \mu_2$.
Consider any distribution $\nu_1$ such that $\mu_1 \ptran \nu_1$ with witness $\pi$. By Corollary \ref{cor:calculated_by_degenerate_witness}, we can assume, w.l.o.g., that the witness $\pi$ is degenerate.
According to Lemma \ref{lem:PWeak_almost_bisimulation}, for any $\epsilon > 0$, there exists a distribution $\nu_2$ such that $\mu_2 \ptran \nu_2$ and $|\nu_2 - \nu_1|_{\PWeak} \le \epsilon$.
Assume that $(\psi_\L\cap \Supp(\nu_1))/\PWeak=\{\mathcal{C}_i\}_{i\in[n]}$.
Since $|\nu_2 - \nu_1|_{\PWeak} = \sum_{\mathcal{C} \in \PRCCS / \PWeak} | \nu_2(\mathcal{C}) -\nu_1(\mathcal{C})| \le \epsilon$, we have 
\begin{equation}
    \nu_2(\psi_\L)\ge\sum_{i\in [n]} |\nu_2(\mathcal{C}_i)|\ge \sum_{i\in [n]} |\nu_1(\mathcal{C})| - \sum_{i\in [n]} | \nu_2(\mathcal{C}) -\nu_1(\mathcal{C})| \ge\nu_1(\psi_L) - \epsilon.
\end{equation}
According to Lemma \ref{lem:PWeak_observable_sup}, $\charMay_\L({Q})= 1$ holds for all $Q \in \bigcup_{i \in [n]} \mathcal{C}_i$. 
By Lemma \ref{lem:linearity_of_testing_outcome}, we have 
\begin{equation}
	\charMay_\L(\mu_2)\ge\charMay_\L(\nu_2)\ge \sum_{Q\in\bigcup_{i \in [n]} \mathcal{C}_i}\nu_2(Q)\charMay_\L({Q})
	= \sum_{Q\in\bigcup_{i \in [n]} \mathcal{C}_i}\nu_2(Q)=\nu_2(\psi_\L)\ge \nu_1(\psi_L) - \epsilon
\end{equation}
Since $\nu_1$ and $\epsilon$ are arbitrary, we have $\charMay_\L(\mu_1) \le \charMay_\L(\mu_2)$, and $\charMay_\L(\mu_1) = \charMay_\L(\mu_2)$ holds by symmetry.
\qed
\end{proof}

\pweakIsStr*

\begin{proof}
    For any $\mu_1,\mu_2$ such that $\mu_1\PWeak\mu_2$ and a probabilistic transition sequence $\mu_1\rightsquigarrow\nu_1$ with a degenerate witness $\pi\in\tau^{*}$, by Lemma~\ref{lem:PWeak_almost_bisimulation}, we can deduce that there exists a distribution $\nu_2$ such that $\mu_2\rightsquigarrow\nu_2$ and $\vert \nu_1-\nu_2\vert_{\PWeak}\le \epsilon$.
    Lemma~\ref{lem:PWeak_is_equipollent} shows that for any $P\in \PRCCS$ and $Q\in [P]$, we have $\charMay_\L({Q})=\charMay({P})$.
    By Lemma~\ref{lem:linearity_of_testing_outcome}, for $i=1,2$, we have
    \begin{equation}
        \charMay_\L({\nu_i})=\sum_{[P]\in\PRCCS/\PWeak}\sum_{Q\in [P]}\nu_i(Q)\charMay_\L({Q})
        =\sum_{[P]\in\PRCCS/\PWeak}\charMay_\L({P})\cdot\nu_i([P]).
    \end{equation}
    Therefore, the following equalities are valid.
    \begin{equation}
        \begin{aligned}
        \vert\charMay_\L({\nu_1})-\charMay_\L({\nu_2})\vert&=\left\vert\sum_{[P]\in\PRCCS/\PWeak}\charMay_\L({P})\cdot(\nu_1([P])-\nu_2([P]))\right\vert\le\vert\nu_1-\nu_2\vert_{\PWeak}\le\epsilon,
    \end{aligned}
    \end{equation}
    which implies that $\charFS_\L(\mu_2)\le\charFS_\L(\mu_1)$.
    By symmetry, $\charFS_\L(\mu_2)=\charFS_\L(\mu_1)$ holds.
    Hence, we can conclude that $\PWeak$ is probabilistically strongly equipollent.
    \qed
\end{proof}

\subsection{Proof of Lemma~\ref{lem:pweak_is_extensional}}
Lemma~\ref{lem:communication}, a rephrasing of the concatenation lemma in \cite{segala_ModelingVerificationRandomized_1995}, is used to prove Lemma~\ref{lem:pweak_is_extensional}.

\begin{lemma}[\cite{segala_ModelingVerificationRandomized_1995}]
    \label{lem:communication}
    Let $P,Q\in \PRCCS$ be two processes.
    If $P\stackrel{\ell}{\Longrightarrow}_{c}\rho_P$ and $Q\stackrel{\overline{\ell}}{\Longrightarrow}_{c}\rho_{Q}$ for some $\ell\in \L$ and distributions $\rho_P,\rho_Q\in\InfDistr(\PRCCS)$, then $P\mid Q \Longrightarrow_{c}\rho_P\mid \rho_Q$.
\end{lemma}

\pweakIsExt*

\begin{proof}
    It suffices to prove that the original $\PWeak$ on $\PRCCS$ is closed under the parallel composition and localization operations.
    Define $\E:=\left\{(P_1\mid Q,P_2\mid Q):P_1\PWeak P_2\text{~and~} Q\in\PRCCS\right\}\cup\PWeak$ on $\PRCCS$.
    Note that $\E$ is an equivalence.
    We will prove that $\E$ is a probabilistic weak bisimulation.
    
    First, we claim that for any distribution $\mu_1,\mu_2,\nu\in\InfDistr(\PRCCS)$ such that $\mu_1(\PWeak)^\dag\mu_2$, $(\mu_1\mid\nu)~\E^\dag~(\mu_2\mid\nu)$ holds.
    In fact, given $A,B\in\PRCCS$, let $[A]\mid B$ be the set $\left\{ A' \mid B :A\PWeak A'\right\}$.
    Since $\mu_1(\PWeak)^\dag\mu_2$, we have
    \begin{equation}
        (\mu_1\mid \nu)([A]\mid B)=\mu_1([A]_{\PWeak})\cdot\nu(B)=\mu_2([A]_{\PWeak})\cdot\nu(B)=(\mu_2\mid\nu)([A]\mid B).
    \end{equation}
    Note that $[A]\mid B \subseteq [A\mid B]_{\E}$ holds by definition.
    Therefore, $\left\{[A]\mid B: A,B\in \PRCCS\right\}$ is a refinement of $\PRCCS/\E$.
    Furthermore, for any $\mathcal{C}\in\PRCCS/\E$, we have 
    \begin{equation}
        (\mu_1\mid \nu)(\mathcal{C})=\sum_{[A]\mid B\in \mathcal{C}}(\mu_1\mid \nu)([A]\mid B)=\sum_{[A]\mid B\in \mathcal{C}}(\mu_2\mid \nu)([A]\mid B)=(\mu_2\mid \nu)(\mathcal{C}),
    \end{equation}
    which implies that $(\mu_1\mid\nu)~\E^\dag~(\mu_2\mid\nu)$

    Given $P_1,P_2,Q\in \PRCCS$ such that $P_1\PWeak P_2$, if $P_1\mid Q\myrightarrow{\alpha}\mu_1$ for some distribution $\mu_1$, then there are three cases.
    \begin{enumerate}
        \item 
        $Q\myrightarrow{\alpha}\rho$ for some $\rho$ and $\mu_1=\delta_{P_1}\mid\rho$.
        Then $P_2\mid Q\stackrel{\alpha}{\Longrightarrow}_{c}\mu_2:=\delta_{P_2}\mid \rho$.
        Since $\delta_{P_1}\PWeak\delta_{P_2}$, by our claim, we have $\mu_1~\E^{\dag}~\mu_2$.
        \item 
        $P_1\myrightarrow{\alpha}\rho_1$ for some $\rho_1$ and $\mu_1=\rho_1\mid\delta_Q$.
        Since $P_1\PWeak P_2$, there exists $\rho_2\in\InfDistr(\PRCCS)$ such that $P_2 \stackrel{\alpha}{\Longrightarrow}_{c} \rho_2$ and $\rho_1\; \PWeak \;\rho_2$. 
        Hence, $P_2\mid Q \stackrel{\alpha}{\Longrightarrow}_{c} \mu_2:= \rho_2\mid\delta_Q$ and by our claim, $\mu_1~\E^{\dag}~\mu_2$ holds.
        \item 
        $P_1\myrightarrow{\ell}\rho_1$, $Q\myrightarrow{\overline{\ell}}o$ for some distributions $\rho_1,o\in\Distr(\PRCCS)$ and $\ell\in\L$, $\mu_1=\rho_1\mid o$ ,and $\alpha=\tau$.
        Since $P_1\PWeak P_2$, there exists $\rho_2\in\InfDistr(\PRCCS)$ such that $P_2 \stackrel{\ell}{\Longrightarrow}_{c} \rho_2$ and $\rho_1\; (\PWeak)^{\dagger} \;\rho_2$. 
        By Lemma~\ref{lem:communication}, we have $P_2\mid Q\stackrel{\tau}{\Longrightarrow}_{c}\mu_2:=\rho_2\mid o$.
        Therefore, $\mu_1~\E^\dag~\mu_2$.
    \end{enumerate}
    By symmetry, we conclude that $\E$ is a probabilistic weak bisimulation, therefore contained in the largest one, i.e., $\PWeak$.
    Hence, for any $P_1,P_2,Q\in \PRCCS$ such that $P_1\PWeak P_2$, we have $P_1\mid Q \PWeak P_2\mid Q$.

    It is trivial to show that $\left\{((a)P_1,(a)P_2):P_1\PWeak P_2\right\}\cup\PWeak$ is a probabilistic weak bisimulation.
    Therefore, for any processes $P_1,P_2\in\PRCCS$ such that $P_1\PWeak P_2$, we have $(a)P_1\PWeak (a)P_2$.
    Now we can conclude that $\PWeak$ is $p$-extensional.
    \qed
\end{proof}


\end{document}